\documentclass[11pt]{article}
\usepackage{brainberry}

\newcommand{\binaryDistance}{d_b}
\newcommand{\dtv}{d_{\-{TV}}}
\newcommand{\dkl}[2]{D_{\-{KL}}\left(#1 \mid #2\right)}

\newcommand{\vbl}{\textnormal{vbl}}
\newcommand{\vblg}{\textnormal{vbl}_{\-{g}}}
\newcommand{\vblb}{\textnormal{vbl}_{\-{b}}}
\newcommand{\vbls}{\textnormal{vbl}^{\sigma}}
\newcommand{\vblsg}{\textnormal{vbl}^{\sigma}_{\-{g}}}
\newcommand{\vblsb}{\textnormal{vbl}^{\sigma}_{\-{b}}}

\newcommand{\degreeConstant}{p_{\textnormal{hd}}}
\newcommand{\badclauseConstant}{\epsilon_{\textnormal{bd}}}
\newcommand{\sbadvars}{\widetilde{V}_{\textnormal{bad}}}
\newcommand{\sbadclauses}{\widetilde{\+C}_{\textnormal{bad}}}
\newcommand{\sgoodvars}{\widetilde{V}_{\textnormal{good}}}
\newcommand{\sgoodclauses}{\widetilde{\+C}_{\textnormal{good}}}
\newcommand{\badvars}{V_{\textnormal{bad}}}
\newcommand{\goodvars}{V_{\textnormal{good}}}

\newcommand{\badclauses}{\+C_{\textnormal{bad}}}
\newcommand{\goodclauses}{\+C_{\textnormal{good}}}

\newcommand{\intersectionclause}{\+C_{\textnormal{intersect}}}
\newcommand{\goodlowerbound}{k_{\textnormal{gl}}}

\newcommand{\frozenclauses}[1][\sigma]{\+C^{#1}_{\textnormal{frozen}}}
\newcommand{\blockedclause}{\+C^{\sigma}_{\textnormal{blocked}}}
\newcommand{\frozenConstant}{\zeta_{\textnormal{frozen}}}

\newcommand{\frozenRho}{\rho_{\textnormal{frozen}}}
\newcommand{\frozenEta}{\eta_{\textnormal{frozen}}}
\newcommand{\aliveVars}{V^\sigma_{\textnormal{alive}}}
\newcommand{\true}{\textnormal{True}}
\newcommand{\false}{\textnormal{False}}
\newcommand{\badInterior}[1][\sigma]{\+C_{\textnormal{com}}^{#1}}
\newcommand{\currentComponent}{\+C^{\sigma}_{\textnormal{ext}}}

\newcommand{\nextVar}{\textnormal{NextVar}(\tau_S, c_0)}
\newcommand{\prunedBadInterior}[1][\sigma]{\overline{\+C}_{\textnormal{com}}^{#1}}
\newcommand{\revealedLowerbound}{k_{\textnormal{revealed}}}

\newcommand{\BetaConstant}{\beta_{\textnormal{ind}}}

\newcommand{\Pinning}{\textsf{Reveal}}
\newcommand{\iid}{i.i.d.\@ }

\newboolean{doubleblind}
\setboolean{doubleblind}{false}

\title{Learning CNF formulas from uniform random solutions \\ in the local lemma regime}
\date{}

\ifthenelse{\boolean{doubleblind}}{\author{Author(s)}}{
\author{Weiming Feng\thanks{School of Computing and Data Science, The University of Hong Kong.  Email: \texttt{wfeng@hku.hk}} 
\and  Xiongxin Yang\thanks{Department of Computer Science, University of California, Santa Barbara. Email: \texttt{xiongxinyang@ucsb.edu}}  
\and  Yixiao Yu\thanks{State Key Laboratory for Novel Software Technology, New Cornerstone Science Laboratory, Nanjing University. Emails: 
\texttt{yixiaoyu@smail.nju.edu.cn, zhangyiyao@smail.nju.edu.cn}}  
\and  Yiyao Zhang\footnotemark[3]}
}
\begin{document}

\maketitle

\thispagestyle{empty}
\setcounter{page}{0}

\begin{abstract}
  We study the problem of learning a $n$-variables $k$-CNF formula $\Phi$ from its \iid uniform random solutions, which is equivalent to learning a Boolean Markov random field (MRF) with $k$-wise hard constraints.
  Revisiting Valiant's algorithm (Commun.\@ ACM'84), we show that it can \emph{exactly} learn 
  (1) $k$-CNFs with bounded clause intersection size under \emph{Lov\'asz local lemma} type conditions, 
  from $O(\log n)$ samples;
  and (2) random $k$-CNFs near the \emph{satisfiability threshold}, 
  from $\widetilde{O}(n^{\exp(-\sqrt{k})})$ samples.
  These results significantly improve the previous $O(n^k)$ sample complexity.
  We further establish new information-theoretic lower bounds on sample complexity 
  for both exact and approximate learning from \iid uniform random solutions.
\end{abstract}

\thispagestyle{empty}
\newpage

\tableofcontents

\thispagestyle{empty}
\pagebreak

\setcounter{page}{1}
\section{Introduction}

The CNF (conjunctive normal form) formula is one of the most fundamental objects in computer science.
One canonical form of CNF formula is the $k$-CNF formula, which is defined by a set of $n$ Boolean variables $V = \{v_1, v_2, \ldots, v_n\}$ and a set of $m$ clauses $\+C = \set{c_1, c_2, \ldots, c_m}$.
The formula is a conjunction of all clauses in $\+C$ and each clause is a disjunction of $k$ distinct literals in $\set{v_i,\neg v_i \mid v_i \in V}$.
Given a $k$-CNF formula $\Phi = (V, \+C)$, 
a solution $X \in \{\true, \false\}^V$ is an assignment of all variables such that 
all clauses in $\+C$ are satisfied. 
Let $\mu_\Phi$ denote the uniform distribution over all satisfying assignments of $\Phi$.

In this paper, we study the problem of 
properly learning $k$-CNF formulas $\Phi$ from \emph{\iid uniform random solutions}.  
Given $T$ \iid samples drawn from $\mu_\Phi$, 
the goal of the learning algorithm is to construct a CNF formula $\widehat{\Phi}$ such that:
for \emph{exact learning},
$\widehat{\Phi}$ has the same set of satisfying assignments as $\Phi$, 
i.e. $\mu_\Phi = \mu_{\widehat{\Phi}}$;
or for \emph{approximate learning}, 
the total variation distance between $\mu_\Phi$ and $\mu_{\widehat{\Phi}}$ can be 
controlled by an error bound $\varepsilon > 0$.
The number of samples $T$ required by the algorithm is referred to as its \emph{sample complexity}, 
and its total running time as the \emph{computational complexity}.

This problem naturally arises in various domains, such as statistical physics and data science.  
A particularly notable motivation comes from the task of 
learning graphical models or equivalently, \emph{Markov random fields} (MRFs) 
with \emph{hard constraints} from \emph{\iid Gibbs samples}.  
Indeed, a CNF formula can be viewed as an MRF over a Boolean variable set $V$, 
where each clause imposes a $k$-wise hard constraint on its variables.  
The uniform distribution $\mu_\Phi$ corresponds to the \emph{Gibbs distribution} induced by this MRF.

In 1984, 
Valiant introduced the framework of probably approximately correct (PAC) learning~\cite{Valiant84} 
and showed that the concept class of $k$-CNF formulas is PAC-learnable 
via a very simple and classical learning algorithm based on the elimination of inconsistent clauses.

\begin{tcolorbox}[
  title=Valiant's Algorithm~\cite{Valiant84},
  colback=white,
  colframe=black,
  coltitle=black,
  colbacktitle=gray!20,
  fonttitle=\bfseries,
  boxrule=1pt,
  arc=3pt
]
\textbf{Input}: number of variables $n$, clause size $k$, $T$ \iid samples $X_1, \ldots, X_T$ from $\mu_\Phi$.
\begin{itemize}[leftmargin=0.3cm]
  \item Let $\widehat{\Phi} = (V, \+C)$ be a CNF formula containing all $2^k \cdot \binom{n}{k}$ possible size-$k$ clauses.
  \item For each clause $c \in \+C$ defined on a $k$-variable set $\vbl(c) \subseteq V$, if there exists a sample $X_i$ for $i \in [T]$ such that $X_i(\vbl(c))$ violates $c$, then remove $c$, i.e., $\+C \gets \+C \setminus \set{c}$.
  \item Return the CNF formula $\widehat{\Phi} = (V, \+C)$.
\end{itemize}
\end{tcolorbox}
\noindent
In the context of learning $k$-CNF formulas from uniform random solutions, 
the proof of PAC-learnability of $k$-CNFs in~\cite{Valiant84} also implies the following approximate learning result.

\begin{theorem}[\cite{Valiant84}, Theorem A]\label{thm:valiant}
  Let $k \geq 2$ be a constant integer. For any $\varepsilon > 0$ and $\delta > 0$, Valiant's algorithm approximately (within total variation distance error at most $\varepsilon$) learns any satisfiable $k$-CNF formula from \iid uniform solutions with probability at least $1 - \delta$ in sample complexity $T = O_k(\frac{n^k+\log (1 / \delta)}{\varepsilon})$ and computational complexity $O_k(n^k T)$.
\end{theorem}

After Valiant's work, 
there has been significant progress on the PAC-learning of Boolean formulas 
(e.g., DNF formulas~\cite{Bshouty96,TaruiT99,KlivansOS04,sellieLearningRandomMonotone2008,sellieExactLearningRandom2009,AlmanNPS25} 
and decision trees~\cite{EhrenfeuchtH89,MehtaR02,BlancLQT22}), but the results for CNF formulas are limited except for some specific classes of CNF formulas~\cite{AngluinFP92,AriasBT17,HermoO20}.
For the research on learning MRFs, 
many works focused on MRFs 
with \emph{soft} constraints~\cite{ChowL68,KargerS01, BreslerMS13,Bresler15,vuffrayInteractionScreeningEfficient2016,KlivansM17,HamiltonKM17,WuSD19,gaitondeUnifiedApproachLearning2024,chandrasekaranLearningSherringtonKirkpatrickModel2025} 
and MRFs with \emph{pair-wise} hard constraints~\cite{BreslerGS14a,BlancaCSV20}.
However, beyond Valiant's classical work~\cite{Valiant84}, 
we are not aware of any result for properly learning general $k$-CNF formulas from uniform random solutions. 
Moreover,
existing MRF learning algorithms do not directly extend to this setting; 
see \Cref{sec:obstacles-in-applying-previous-mrf-learning-algorithms} for a discussion of the technical challenges.

The classical result in \Cref{thm:valiant} applies to all satisfiable $k$-CNF formulas.
However, its sample complexity is prohibitively large, as $k$ appears in the exponent of $n$.
We revisit Valiant's algorithm and release its power for two natural and important classes of CNF formulas:
CNF formulas satisfying a Lov\'asz local lemma type condition 
and random CNF formulas near the satisfiability threshold. 
For both cases,
we show that the required number of samples can be significantly reduced compared to the general setting.
We remark that our results tackle the problem of \emph{exactly} learning CNF formulas, which is more challenging than approximate learning.
In addition,
we establish new information-theoretic lower bounds on the sample complexity
for learning CNF formulas satisfying the local lemma condition.

\subsection{Our results: Learning CNF formulas in the local lemma regime}

We consider the following class of CNF formulas with degree and intersection constraints.

\begin{definition}[($k,d,s$)-CNF formula]
  Let $k, d, s$ be three positive constant integers. 
  A CNF formula $\Phi = (V, \+C)$ is said to be a ($k,d,s$)-CNF formula if 
  every clause $c_i \in \+C$ contains exactly $k$ variables which are denoted as $\vbl(c_i)$,
  each variable $x \in V$ appears in at most $d$ different clauses, 
  and for any two distinct clauses $c_i, c_j \in \+C$ share at most $s$ variables, 
  i.e., $|\vbl\tuple{c_i} \cap \vbl\tuple{c_j}| \leq s$.

  In particular,
  when $s = k$, there are no constraints on the size of the intersection of two clauses, 
  we denote $(k,d,k)$-CNF formulas as $(k,d)$-CNF formulas.
\end{definition}

The problem of learning ($k,d,s$)-CNF formulas can be formulated as follows. 

\begin{problem}\label{problem:learning-cnf} Learning a ($k,d,s$)-CNF $\Phi = (V,\+C)$ formula from \iid uniform solutions.
  \begin{itemize}
    \item \textbf{Input}: Number of variables $n$, parameters $k,d,s$, a confidence parameter $\delta > 0$, an error bound $\varepsilon > 0$, and $T = T(n, k, d, s, \varepsilon, \delta)$ \iid uniform random solutions $X_1,\ldots,X_T$ from $\mu_\Phi$.
    \item \textbf{Output}: The output satisfies the following requirements with probability at least $1 - \delta$:
          \begin{itemize}
            \item For exact learning ($\varepsilon = 0$), output the CNF formula $\widehat{\Phi}$ such that $\mu_\Phi = \mu_{\widehat{\Phi}}$\footnote{Two CNF formulas $\Phi$ and $\widehat{\Phi}$ may have different set of clauses but they have the same set of satisfying assignments.}.
            \item For approximate learning ($\varepsilon > 0$), output a CNF formula $\widehat{\Phi}$ such that the total variation distance between $\mu_\Phi$ and $\mu_{\widehat{\Phi}}$ is at most $\varepsilon$.
            
          \end{itemize}
  \end{itemize}
  \end{problem}

We study \Cref{problem:learning-cnf} when the input CNF formula satisfies a Lov\'asz local lemma type condition.
The local lemma~\cite{Locallemma75} is a classical condition in combinatorics to guarantee the existence of certain combinatorial objects. For the $(k,d)$-CNF formula $\Phi$, the local lemma condition says if 
\begin{align*}
k \geq \log d + \log k + \log e = \log d + o(k),
\end{align*}
where $\log$ denotes $\log_2$, then the formula $\Phi$ must have a satisfying assignment. Later on, the local lemma was widely used in theoretical computer science, including construction algorithms for constraint satisfaction problems~\cite{MoserT10} and sampling algorithms for CNF formulas~\cite{Moitra19,fengFastSamplingCounting2021,he2021perfect, jain2021sampling, Feng0Y21,JainPV21,HeWY22,0011WY23,WangY24}. These algorithms can serve as the oracle for generating the \iid solutions of CNF formulas.

Our result discovers that for CNF formulas satisfying 
some local lemma type conditions $k = \Omega(\log d)$, 
the size of the intersection between two clauses, the parameter $s$, 
plays a crucial role in the sample complexity of learning CNF formulas. 
With a proper bound on the intersection size, 
Valiant's algorithm achieves the optimal $\Theta(\log n)$ sample complexity. 
Without the bounds of intersection size, 
the learning problem requires at least a polynomial in $n$ number of samples.

\subsubsection{CNF formulas with bounded intersection size}
We now give our results for the \emph{exact} learning of CNF formulas. 
The following result considers CNF formulas with \emph{sublinear} size intersection $s=o(k)$.

\begin{restatable}[]{theorem}{learningSublinearIntersection}
  \label{thm:main-bounded}
  Let $\eta \in (0, 1)$ be a constant.
  For any integers $k,d,s$ satisfying $s = k^{1-\eta}$ and  $k\ge \log d+ o(k) + O_\eta(1)$, Valiant's algorithm exactly learns any $(k,d,s)$-CNF formula from \iid uniform solutions with probability at least $1 - \delta$ with sample complexity $T = O_{k,\eta}(\log\frac{n}{\delta})$ and computational complexity $O_{k,\eta}(n^k \log \frac{n}{\delta})$.
\end{restatable}

The above theorem shows that for CNF formulas with sublinear intersection $s = k^{1-\eta}$, 
under the near-optimal (up to $o(k)$ additive term) 
local lemma condition $k \gtrsim \log d$, 
Valiant's algorithm can learn the CNF formula \emph{exactly} with logarithmic sample complexity.
An important class of CNF formulas is the \emph{linear} $k$-CNF formulas, 
where the intersection size between any two clauses is at most 1. 
We have the following corollary for exactly learning linear $k$-CNF formulas.

\begin{corollary}[Linear CNF formulas]\label{cor:linear-cnf-lower-bound}
  For $k \geq \log d + o(k)$, 
  the result in \Cref{thm:main-bounded} holds for linear $(k,d)$-CNF formulas 
  with sample complexity $T = O_{k}(\log\frac{n}{\delta})$ 
  and computational complexity $O_{k}(n^k \log \frac{n}{\delta})$.
\end{corollary}

Our next theorem shows that the $O(\log n)$ sample complexity is \emph{tight} for exact learning $k$-CNF formulas with sublinear intersection. In fact, the hard instance satisfies $d = 1$ and $s = 0$, which means even if all clauses are disjoint, $\Omega(\log n)$ sample complexity is required.

\begin{restatable}[]{theorem}{lowerBoundSublinearIntersection}
  \label{theorem:lower-bound-simple-thm}
  Let $k \geq 2$ be a constant integer. Any algorithm that exactly learns an $n$-variable $(k,1,0)$-CNF formula from \iid uniform solutions with probability at least $\frac{1}{3}$ requires $\Omega_k(\log n)$ samples.
\end{restatable}

We then consider CNF formulas with \emph{linear} size intersection, where two clauses share $s = O(k)$ variables. We show the following result for Valiant's algorithm on exactly learning CNF formulas.

\begin{restatable}[]{theorem}{learningLinearIntersection}
  \label{thm:main-linear-bounded}
  Let $\zeta \in (0, 1)$ be a constant.
  For any integers $k,d,s$ satisfying $s = \zeta k$ and $k \geq C \log d + o(k) + O_\zeta(1)$, where
  \begin{align*}
    C \triangleq
    \begin{cases}
      \frac{1}{1 - \sqrt{2 \zeta}},
       & \zeta \in (0, 3 - 2\sqrt{2}),
      \\ \frac{2}{1 - \zeta}, & \zeta \in [3 - 2\sqrt{2}, 1),
    \end{cases}
  \end{align*}
  Valiant's algorithm exactly learns any $(k,d,s)$-CNF formula from \iid uniform solutions with probability at least $1 - \delta$ with sample complexity $T = O_{k,\zeta}(\log\frac{n}{\delta})$ and computational complexity $O_{k,\zeta}(n^k \log \frac{n}{\delta})$.
\end{restatable}

The above theorem shows that for CNF formulas with linear intersection size $s = \zeta k$, 
under a relaxed local lemma condition $k \ge \Omega_\zeta(\log d) + o_\zeta(k)$, 
Valiant's algorithm can still exactly learn the formula using only $O(\log n)$ samples. 
It is worth noting that the constant $C(\zeta)$ satisfies $C(\zeta) \to 1$ as $\zeta \to 0$,  
indicating that our condition approaches the true local lemma regime when $\zeta$ is small.  
However, $C(\zeta) \to \infty$ as $\zeta \to 1$, 
which means the result no longer applies to CNF formulas whose clauses may arbitrarily intersect. 
Indeed, our next two lower bound results show that without any bound on the intersection size, 
the $O(\log n)$ sample complexity is information-theoretically impossible.

\subsubsection{CNF formulas without intersection size bound}

We establish two lower bound results showing that the assumption of bounded intersection size 
is \emph{necessary} for any learning algorithm to achieve \emph{logarithmic} sample complexity.

In particular, one can construct two CNF formulas $\Phi_1$ and $\Phi_2$ 
that both satisfy the local lemma condition but allow pairs of clauses to share too many variables, 
such that the total variation distance between $\mu_{\Phi_1}$ and $\mu_{\Phi_2}$ is at most $\exp(-\Omega(n))$. 
Hence, any exact learning algorithm would require \emph{exponentially} many samples 
to distinguish between $\Phi_1$ and $\Phi_2$.
We have the following lower bound result.

\begin{restatable}[]{theorem}{lowerExactExp}
  \label{proposition:lower-bound-general-thm}
  Let $k \geq 2$ be a constant integer. Any algorithm that exactly learns an $n$-variable $(k,k,k-1)$-CNF formula from \iid uniform solutions with probability $\frac{1}{3}$ requires $\exp(\Omega_k(n))$ samples.
\end{restatable}

Combining the above lower bound result on exact learning 
with Valiant's algorithmic results on approximate learning (\Cref{thm:valiant}), 
we obtain a sharp separation between the sample complexities of exact and approximate learning CNF formulas.
While exact learning may require $\exp(\Omega(n))$ many samples, 
the approximate learning can be achieved with only $O(n^k)$ samples.

Furthermore, even for the problem of \emph{approximately} learning CNF formulas $\Phi$ 
with total variation distance error bound $\varepsilon$, 
we show that if two clauses in $\Phi$ share too many variables, 
then any approximate learning algorithm must require a \emph{polynomial} number of samples.

\begin{restatable}[]{theorem}{lowerApproxPoly}
  \label{theorem:lower-bound-general-thm}
  Fix a constant integer $k \geq 2$ and a constant error bound $\varepsilon_0 \in (0, \frac{1}{400 \cdot 2^k})$.
  Any algorithm that approximately learns an $n$-variable $(k,k,k-1)$-CNF formula 
  from \iid uniform solutions with total variation distance error at most $\varepsilon_0$ 
  and success probability $\frac{1}{3}$ requires at least 
  $\Omega_{k,\varepsilon_0}((\frac{n}{\log n})^{1 - \frac{2}{k}})$ samples.
\end{restatable}

Our upper and lower bound results show that the size of the intersection plays a crucial role 
in the sample complexity of learning CNF formulas in the local lemma regime.
The hard instance in the above theorem satisfies $d = k$; 
these CNF formulas satisfy a very strong local lemma condition because $k = d \gg \log d$. 
For these CNF formulas without intersection bound for clauses, 
even \emph{approximately} learning requires $\widetilde\Omega({n^{1 - \frac{2}{k}}})$ samples. 
The lower bound is close to linear in $n$ when $k$ is a large constant. 
However, if an intersection bound is assumed, 
Valiant's algorithm can \emph{exactly} learn the CNF formula using $O(\log n)$ samples 
under a mild local lemma condition $k = \Omega(\log d)$.

\subsection{Our results: Learning random CNF formulas near the satisfiability threshold}
The random CNF formula is a fundamental model in probability, physics, and computer science.
We use $\Phi = \Phi(k, n, m = \floor{\alpha n})$ to denote a random $k$-CNF formula on $n$ variables 
$V = \{v_1, \ldots, v_n\}$ and $m = \floor{\alpha n}$ random clauses $\+C = \set{c_1, \ldots, c_m}$.
Each clause of the formula is an independent disjunction of $k$ literals 
chosen uniformly and independently from $\set{v_1, \ldots, v_n, \neg v_1, \ldots, \neg v_n}$.
Note that each clause has exactly $k$ literals (repetitions allowed), 
and there are $(2n)^{km}$ possible formulas.
The parameter $\alpha \in \mathbb{R}^+$ is called the \emph{density} of the formula.
A fundamental problem for random CNF formulas is to determine a condition of the density $\alpha$ 
such that the formula is satisfiable with high probability.
Building on a long line of works~\cite{kirousis1998approximate,friedgut1999sharp,achlioptasAsymptoticOrderRandom2002,achlioptas2003threshold,coja2014asymptotic,ding2022satisfiability},
Ding, Sly, and Sun~\cite{ding2022satisfiability} answered this question 
and proved that there exists a sharp threshold 
$\alpha_{\star}(k)= 2^k \ln 2 - (1 + \ln 2) / 2 + o_k(1)$  such that
$$\forall \eps > 0, \quad \lim_{n \to \infty} \Pr{\Phi(k, n, m = \floor{\alpha n}) \text{ is satisfiable}} = \begin{cases}
    1 & \text{if } \alpha \le \alpha_{\star}(k) - \eps, \\
    0 & \text{if } \alpha \ge \alpha_{\star}(k) + \eps.
  \end{cases}$$

The random CNF formula shares some similarities with the CNF formulas in the local lemma regime.  
The above satisfiability condition can be rewritten as $k \geq \log \alpha + O(1)$, which is very similar to the local lemma condition $k \geq \log d + o(k)$ with the difference that the degree $d$ is replaced by the density (average degree) $\alpha$. 
It was discovered that some algorithmic techniques developed for CNF formulas in the local lemma regime can be extended to random CNF formulas~\cite{GalanisGGY21,heImprovedBoundsSampling2023, ChenGGGHMM24, chenCountingRandomSAT2025}. Recently, \cite{chenCountingRandomSAT2025} designed an algorithm for sampling uniform solutions of random CNF formulas near the satisfiability threshold. 

Inspired by this connection,
we further analyze Valiant's algorithm on the problem of exact learning random CNF formulas $\Phi = \Phi(k, n, m = \floor{\alpha n})$. The problem is formulated as follows.

\begin{problem}\label{problem:learning-random-cnf} Exact learning a random CNF formula from \iid uniform solutions.
\begin{itemize}
  \item \textbf{Input}: Parameters $n, k,\alpha$ of the random formula, a confidence parameter $\delta > 0$, and $T = T(n, k, \alpha, \delta)$ \iid uniform random solutions $X_1,\ldots,X_T$ from the distribution $\mu_\Phi$, where $\Phi = \Phi(k, n, m = \floor{\alpha n})$ is a random $n$-variable $k$-CNF formula with density $\alpha$.
  \item \textbf{Output}: A CNF formula $\widehat{\Phi}$ satisfies that,
   with probability at least $1 - o(\frac{1}{n})$ over the choice of $\Phi$, it holds that $\mu_\Phi = \mu_{\widehat{\Phi}}$ with probability at least $1 - \delta$, where the probability is taken over the randomness of $X_1,\ldots,X_T$ and the independent randomness $\+R$ inside the learning algorithm (assume $\+R = \emptyset$ if the learning algorithm is deterministic). Formally,
        \begin{align*}
          \Pr[\Phi]{ \Pr[X_1,\ldots,X_T,\+R]{ \mu_\Phi = \mu_{\widehat{\Phi}} } \ge 1 - \delta } \ge 1 - o\left(\frac{1}{n}\right).
        \end{align*}
\end{itemize}
\end{problem}

Note that Valiant's algorithm is deterministic and thus $\+R = \emptyset$ in our analysis. 
We prove the following result for Valiant's algorithm on exactly learning random CNF formulas.

\begin{restatable}[]{theorem}{learningRandom}
  \label{theorem:learning-random-cnf}
  Let $\alpha \in \mathbb{R}^+$ and $k \in \mathbb{N}$ be two constants satisfying $k \ge 10^5$, 
  $\alpha \le 2^{k - \widetilde{O}(k^{4/5})}$.
  For any $n \geq n_0(k,\alpha)$ sufficiently large, Valiant's algorithm solves \Cref{problem:learning-random-cnf} of exact learning with sample complexity $T = O_k(n^{\exp(-\sqrt{k})} \log \frac{n}{\delta})$ and computational complexity $O_k(n^{k + \exp(-\sqrt{k})} \log \frac{n}{\delta})$.
\end{restatable}

Our result holds for random CNF formulas satisfying $k \geq \log \alpha + o(k)$, which is very close to the satisfiability threshold $k \geq \log \alpha + O(1)$. The coefficient of $\log \alpha$ is tight, but some $o(k) = \widetilde{O}(k^{4/5})$ additive terms are required. Compared to the $O(n^k)$ sample complexity in \Cref{thm:valiant}, we give a much better sample complexity $\widetilde{O}(n^{\exp(-\sqrt{k})})$, where the exponent goes to $0$ as $k$ becomes large. 
We remark that the exponent $\exp(-\sqrt{k})$ is not critical. One can improve it to $\exp(-k^c)$ for some $\frac{1}{2}< c < 1$ by a more careful analysis.
Compared to our sample complexity lower bounds in \Cref{proposition:lower-bound-general-thm} and \Cref{theorem:lower-bound-general-thm}, our result shows that typical random CNF formulas are significantly easier to learn than adversarial CNF formulas in the local lemma regime. Compared with the $O(\log n)$ sample complexity in \Cref{cor:linear-cnf-lower-bound} and \Cref{thm:main-linear-bounded}, our result for random CNF formulas requires more samples. The reason is that although the typical random CNF formula has some good structural properties (e.g., bounded average degree and bounded intersection size), it still can have many variables with \emph{unbounded degree}. Hence, we need to apply a different and more involved analysis for random CNF formulas. See technique overview in \Cref{sec:overview} for more details. 

\subsection{Related works and open problems}

\paragraph{Related works}

Despite the work discussed before, there are other related works on the problem of learning CNF formulas.
A line of work studied the problem of one-shot learning of CNF formulas. The problem considers CNF formulas with an \emph{external field}. The learning algorithm is required to cover the external field with $\emph{one}$ sample~\cite{daganLearningIsingModels2021,BhattacharyaRamanan/2021/ParameterEstimationUndirected, GalanisKalavasisKandiros/2024/LearningHardConstrained,GalanisGoldbergZhang/2025/OneShotLearning}. Recent work~\cite{ChauhanPanageas/2025/LearningIsingModels} extended the problem to learning the temperature of an Ising model truncated by a CNF formula.

Moreover, De, Diakonikolas, and Servedio~\cite{DeDS15} studied the problem of learning Boolean functions from the uniform distribution of satisfying assignments. Instead of CNF formulas, they considered linear threshold functions and DNF formulas. These functions are not defined by local hard constraints, which are very different from CNF formulas. Furthermore, Fotakis, Kalavasis, and Tzamos~\cite{FotakisKT22} studied the problem of estimating the parameters of $n$-dimensional Boolean product distributions, where samples are truncated by a set $S \subseteq \{0, 1\}^n$. Their algorithm is based on the membership oracle of $S$.

Additionally, several recent works have investigated learning MRFs from a wide variety of local Markov chains (e.g., Glauber dynamics) rather than from \iid samples~\cite{gaitondeUnifiedApproachLearning2024,gaitondeBypassingNoisyParity2025,gaitondeBetterModelsAlgorithms2025}. This approach circumvents the assumption of sample oracles that generate \iid samples from MRFs and overcomes the $n^{\Theta(k)}$ computational complexity barrier associated with learning from \iid samples.
However, the solution space of a CNF formula can be \emph{disconnected} under the moves of Glauber dynamics.
It would still be interesting to study the problem of learning CNF formulas from a suitable Markov chain dynamics.

\paragraph{Open problems}
We list some open problems for learning CNF formulas. 

\begin{itemize}
  \item \textbf{Tight trade-off in exact learning.}
  In \Cref{thm:main-linear-bounded}, we prove that exact learning of CNF formulas with intersection size $s = \zeta k$ is possible using $O_{\zeta,k}(\log n)$ samples, under a relaxed local lemma condition $k \ge C(\zeta) \cdot \log d + o_{\zeta}(k)$ for $\zeta \in (0,1)$. An important open problem is to determine the precise trade-off between the parameter $C(\zeta)$ in the local lemma condition and the sample complexity achievable by exact learning algorithms.
  \item \textbf{Approximate learning in the local lemma regime.}
  Our \Cref{theorem:lower-bound-general-thm} establishes a lower bound of $\widetilde{\Omega}_k(n^{1 - 2/k})$ samples for any approximate learning algorithm. It remains an interesting question to further strengthen this lower bound and to design an approximate learning algorithm whose sample complexity improves upon Valiant's classical $O(n^k)$ bound.
  \item \textbf{Learning random CNF formulas.}
  Our current results apply when the clause density satisfies $\alpha \le 2^{k - \widetilde{O}(k^{4/5})}$. A natural direction for future work is to extend this regime to $\alpha \le \tfrac{2^{k}}{\mathrm{poly}(k)}$. It would also be interesting to study whether the sample complexity $\widetilde{O}(n^{\exp(-\sqrt{k})})$ can be further reduced to sub-polynomial or even polylogarithmic in $n$.
\end{itemize}

\section{Technical overview}\label{sec:overview}
Let $\Phi = (V, \+C)$ be a $k$-CNF formula where every clause contains distinct $k$ literals.
The $k$-CNF formula is a canonical example of a Markov random field with hard constraints, 
where every clause poses a local hard constraint on $k$ variables. 
In the paper, we show that a very simple and natural marginal lower bound condition, 
denoted as the \emph{resilience property}, 
plays a crucial role in the sample complexity of proper learning $k$-CNF formulas. 
For any clause $c^*$ with $k$ variables $\vbl(c)$, 
only one assignment $\sigma^* \in \set{\text{True}, \text{False}}^k$ violates $c^*$. 
We call $\sigma^*$ the \emph{forbidden assignment} of $c^*$. 
The resilience property says that for any clause $c^* \notin \+C$, 
the probability that $X_{\vbl(c^*)} = \sigma^*$ is either $0$ or bounded away from $0$ by a certain quantity $\theta$ for a uniform random solution $X \sim \mu_{\Phi}$.
To cover the application of random CNF formulas, instead of $k$-CNF formulas, we state the definition for a slightly more general case where each clause contains \emph{at most} $k$ distinct variables.
\begin{definition}[$\theta$-resilience]\label{def:resilience}
  Given a parameter $\theta \in (0, 1)$,
  a CNF formula $\Phi = (V, \+C)$ with each clause containing at most $k$ variables is said to be \emph{$\theta$-resilient}
  if for any clause $c^*\notin \+C$ with $k$ variables and forbidden assignment $\sigma^*$,
  the probability that a uniform random solution $X$ of $\Phi$
  violates $c^*$ is either $0$ or at least $\theta$, i.e.,
  $$\Pr[X\sim\mu_{\Phi}]{X_{\vbl(c^*)} = \sigma^*} = 0 \text{ or } \Pr[X\sim\mu_{\Phi}]{X_{\vbl(c^*)} = \sigma^*} \ge \theta.$$
\end{definition}

This property appeared in previous work~\cite{BreslerGS14a,BlancaCSV20} 
on learning MRFs with \emph{pair-wise} hard constraints such as graph coloring 
and weighted independent set (hardcore model). 
We study the role of this property in both the algorithm and the hardness of learning CNF formulas.
\begin{itemize}
  \item On the algorithmic side, 
  it is straightforward to show that the $\theta$-resilient condition implies that 
  Valiant's algorithm can exactly learn CNF formulas with sample complexity $O(\frac{1}{\theta}\log n)$. 
  Our main contribution is to show that for the class of CNF formulas studied in this paper, 
  the resilience property \emph{can be established} with a \emph{large enough} $\theta$.
  Unlike the MRFs with pair-wise hard constraints considered in previous work, 
  the higher-order interactions make the resilience property much harder to establish. 
  We exploit the Lov\'asz local lemma condition and several structural properties of CNF formulas to 
  establish the desired resilience property.
  \item On the hardness side, consider a CNF formula $\Phi$ that lacks the $\theta$-resilience property, i.e., there exists a clause $c^*$ with forbidden assignment $\sigma^*$ such that $0 < \mathbb{P}_{X\sim\mu_{\Phi}}[X_{\vbl(c^*)} = \sigma^*] < \theta$.
  A simple observation gives an $\Omega(\frac{1}{\theta})$ sample complexity lower bound of exact learning. 
  We further show that the lack of the resilience property can also imply a sample complexity lower bound for 
  \emph{approximate} learning. 
  Furthermore, for $k$-CNF formulas satisfying the local lemma condition 
  but without a bound on the interaction size of two clauses, 
  we can construct a hard instance to make it lack the resilience property, which proves our hardness result.
\end{itemize}

\subsection{Sample complexity of Valiant's algorithm}
The following sample complexity bound for Valiant's algorithm is straightforward to establish.
\begin{proposition}\label{pps:resilience}
  For any satisfiable and $\theta$-resilient CNF formula $\Phi$ with each clause containing at most $k$ variables,
  Valiant's algorithm exactly learns $\Phi$
  from \iid uniform solutions with probability at least $1 - \delta$
  with sample complexity $T = O(\frac{k}{\theta}\log \frac{n}{\delta})$
  and computational complexity $O_k(n^k \cdot T)$.
\end{proposition}

\begin{proof}
  Note that since all clauses contain at most $k$ variables, for a clause with $i < k$ variables, we can extend the clause to a clause of size $k$ by adding $k - i$ distinct variables not in the clause, each with a literal of either polarity.
  Enumerating all possible extensions yields $2^{k - i} \binom{n - i}{k - i}$ size-$k$ clauses for each size-$i$ clause.
  After the extension, we obtain a $k$-CNF formula and denote the resulting set of clauses by $\+C$. We enumerate all possible size-$k$ clauses over $V$ in Valiant's algorithm. Fix a clause $c^*$. If $\oPr_{X \sim \mu_{\Phi}}[X_{\vbl(c^*)} = \sigma^*] = 0$, then it will never be eliminated by Valiant's algorithm. Otherwise, let $\+C^*$ be the set of all clauses such that for each $c^* \in \+C^*$, $\oPr_{X \sim \mu_{\Phi}}[X_{\vbl(c^*)} = \sigma^*] \ge \theta$. For any $c^* \in \+C^*$, let $\^I_{c^*}$ be the indicator that $c^*$ is not violated by any sample $X_i$ for $i \in [T]$.
  By the resilience property and independence of samples, we have
  \begin{equation*}
    \E{\^I_{c^*}} = \Pr{\^I_{c^*}=1} = \tuple{1-\Pr[X\sim\mu_{\Phi}]{X_{\vbl(c^*)} = \sigma^*}}^T \le \tuple{1-\theta}^T.
  \end{equation*}
  Since the number of clauses that do not appear in $\+C^*$ is at most $(2n)^k$,
  by Markov's inequality,
  \begin{equation*}
    \Pr{\sum_{c^* \in \+C^*} \^I_{c^*} \ge 1}
    \le \E{\sum_{c^* \notin \+C} \^I_{c^*}}
    \le (2n)^k \tuple{1-\theta}^T
    \le (2n)^k \exp\tuple{-T\theta}.
  \end{equation*}
  When $T\ge \frac{1}{\theta} \tuple{k\ln (2n)-\ln \delta}$, this probability is at most $\delta$,
  which means Valiant's algorithm eliminates all such clauses with probability at least $1 - \delta$.

  Now we argue that, suppose all clauses in $\+C^*$ are eliminated, then the output formula $\hat{\Phi}$ is equivalent to $\Phi$. Note that any clause in $\+C$ will never be eliminated and hence $\Phi$ is implied by $\hat{\Phi}$. On the other hand, for any solution $\sigma$ of $\Phi$, all clauses forbidding $\sigma$ are in $\+C^*$ and hence eliminated. Therefore, $\sigma$ is also a solution of $\hat{\Phi}$ and hence $\hat{\Phi}$ is implied by $\Phi$. This completes the proof.
\end{proof}

We next show how to establish the resilience property. Consider a $(k,d,s)$-CNF formula $\Phi = (V,\+C)$ with a local lemma condition. Fix an arbitrary clause $c^* \notin \+C$ with variable set $\vbl(c^*) = \set{v^*_1, \ldots, v^*_k}$ and forbidden assignment $\sigma^* = (\sigma^*_1, \ldots, \sigma^*_k)$. We show that $X_{\vbl(c^*)} = \sigma^*$ has a constant probability $\Omega(1)$ for a uniform random solution $X \sim \mu_{\Phi}$.

\subsubsection{Structured CNF formulas in the local lemma regime}

\paragraph{Local lemma and local uniformity}
The Lov\'asz local lemma guarantees a local uniformity property for the distribution $\mu_{\Phi}$ 
if the CNF formula satisfies a local lemma condition $k \gtrsim \log d$. 
Let $X \sim \mu_{\Phi}$. For any variable $v \in V$, 
the \emph{local uniformity property} (\Cref{lem:local-uniformity}) states that the marginal distribution of $X_v$ 
is close to the uniform distribution over $\set{\true, \false}$. 
This property was first observed in the algorithmic local lemma~\cite{HaeuplerSS11} and then 
widely used in local-lemma-based sampling and approximate counting algorithms~\cite{Moitra19}.

One natural idea is to establish the resilience property by 
recursively using the local uniformity property. 
Specifically, given $X \sim \mu_{\Phi}$, 
we first reveal the value of $X_{v_1^*}$. 
The local uniformity property guarantees that $\mathbb{P}[X_{v_1^*} = \sigma^*_1] \approx \frac{1}{2}$. 
Conditional on $X_{v_1^*} = \sigma^*_1$, 
we can simplify the CNF formula by removing the variable $v_1^*$ 
and all clauses satisfied by $v_1^*$. 
We then keep applying the same process to the next variable in the simplified formula.
However, 
this straightforward approach fails to establish the resilience property because of the following reasons.
\begin{itemize}
  \item The revealing process will keep removing variables 
  so that the number of unrevealed variables in a clause may become smaller than $\log d$ at some point. 
  Then, the local lemma condition breaks down, and the local uniformity property disappears.
  \item Indeed, in the proof of our lower bound result in \Cref{proposition:lower-bound-general-thm}, 
  for CNF formulas satisfying a very strong local lemma condition $k = d \gg \log d$ 
  but without a bound on the intersection size of two clauses, 
  we can construct a hard instance $\Phi$ to show that the resilience property fails. 
  Specifically, we can show that 
  $0 < \mathbb{P}_{X \sim \mu_{\Phi}}[X_{\vbl(c^*)} = \sigma^*] < \exp(-\Omega(n))$.
\end{itemize}

To prove the desired resilience property, we must use the local uniformity property combined with structural properties of the CNF formulas we are interested in. We carefully design a process to reveal the values of $X$ at some variables (including variables outside $\vbl(c^*)$) to guarantee that the local uniformity property holds with constant probability throughout the process.

\paragraph{CNF formulas with bounded intersection}
Here, we give a proof overview of \Cref{thm:main-bounded},
the case of sub-linear intersection size $s = o(k)$ with the local lemma condition $k \geq \log d + o(k)$.
The formal proof is given in \Cref{sec:sublinear-intersection}.
We will apply a similar analysis to the case of linear intersection size $s = \zeta k$ in \Cref{sec:linear-intersection},
which will prove \Cref{thm:main-linear-bounded}.

Fix a $(k,d,s)$-CNF formula $\Phi = (V, \+C)$ and a size-$k$ clause $c^*\notin\+C$ with forbidden assignment $\sigma^*$.
For clarity, we assume $\vbl(c^*) \neq \vbl(c)$ holds for all $c \in \+C$ here. 
Other corner cases will be handled in the formal proof (\Cref{lem:resilience-bounded-simple}).
To show the lower bound of the probability $\mathbb{P}_{X \sim \mu_{\Phi}}[X_{\vbl(c^*)} = \sigma^*] $,
we design a process to reveal the values of $X \sim \mu_\Phi$ at certain variables.

We first find a set of clauses $\tilde{\+C} \subseteq \+C$ such that 
each clause $c \in \tilde{\+C}$ shares at least $t$ variables with $c^*$, 
where $t = o(k)$ is a properly chosen threshold depending on $s$. 
By the bounded intersection assumption, 
every two clauses in $\tilde{\+C}$ share at most $s$ variables with each other.  
The bounded intersection allows a combinatorial argument of set families (\Cref{lem:subset-intersect-new}), 
with which we can show that the set $\tilde{\+C}$ contains at most $o(k)$ clauses. 
In summary, we establish the structural property that there are only $o(k)$ clauses 
that can share more than $o(k)$ number of variables with $c^*$.

Next, we show the revealing process on $X$ to establish the desired resilience property. The revealing process consists of the following two steps.
We first reveal $X_S$ on a subset of variables $S \subseteq V \setminus \vbl(c^*)$ outside $\vbl(c^*)$ such that $|S| = o(k)$ and with a constant probability, all
the clauses in $\tilde{\+C}$ are satisfied by $X_S$ (\Cref{cond:pin-sequence}).
The existence of such a set $S$ is once again guaranteed by the bounded intersection between two clauses and the constant probability bound is provided by the local uniformity property. 
Furthermore, since we only reveal $o(k)$ variables, 
the number of unrevealed variables in each clause is at least $k - o(k)$. 
Therefore, the local lemma condition always holds throughout the process.
Assume the above good event happens, 
in the second step, 
we reveal the value of $v_i^* \in \vbl(c^*)$ one by one from $i = 1$ to $k$. 
Note that all clauses in $\tilde{\+C}$ are satisfied in the first step, 
and we can remove them from the CNF formula. 
The remaining clauses contain at most $o(k)$ variables with $\vbl(c^*)$. 
Hence, during the second step, all clauses contain at least $k - o(k)$ unrevealed variables. 
By the local uniformity property, 
we can show that $X_{v_i^*} = \sigma^*_i$ with a constant probability for all $i \in [k]$. 
This establishes the desired resilience property with $\theta = \Omega_k(1)$.

\subsubsection{Random CNF formulas near the satisfiability threshold}
We now move to the case of random CNF formulas. We first give some quick observations about the structure of typical random CNF formulas. With high probability, each clause contains at least $k - 2$ variables and at most $k$ variables, and two clauses share at most $3$ variables with each other. Hence, it behaves like a linear $k$-CNF formula. However, we cannot use the same technique as above because although the average degree $\alpha \lesssim 2^{k}$ is small, the maximum degree $d \approx \frac{\log n}{\log \log n}$ is unbounded. Hence, the standard local lemma condition $k \gtrsim \log d$ does not hold. We need a more careful and involved analysis of random CNF formulas.

Fix a feasible configuration $\sigma^*$ in $\Lambda = \{v_1^*, \ldots, v_k^*\}$ in $\mu_{\Phi}$.
We show that $\mathbb{P}_{X\sim\mu_{\Phi}}[X_{\vbl(c^*)} = \sigma^*] \gtrsim n^{-\exp(-\sqrt{k})}$. 
Using the chain rule, the probability can be decomposed as follows
\begin{equation*}
  \Pr[X\sim\mu_{\Phi}]{X_{\vbl(c^*)} = \sigma^*} 
  = \prod_{i = 1}^k \Pr[X\sim\mu_{\Phi}]{X_{v_i^*} = \sigma_i^* \mid \forall j < i, X_{v_j^*} = \sigma_j^*} \defeq \prod_{i=1}^k p_i.
\end{equation*}
For each conditional probability $p_i$, it suffices to show that 
$p_i \gtrsim n^{-\exp(-k^{4/5})}$. 
Let $\pi$ denote the distribution $\mu_{\Phi}$ conditioned on 
$v_j^* = \sigma_j^*$ for all $j < i$, and let $X \sim \pi$. 
We then design a revealing process to obtain a lower bound 
on the probability that $X_{v_i^*} = \sigma_i^*$.

Our revealing process on $X$ consists of two steps: pre-revealing and conditional-revealing. 
At a high level, 
the pre-revealing step reveals $X$ on a subset $S$ where $v_i^* \notin S$, 
and with a constant probability, 
the pre-revealing result $X_S$ satisfies certain ``nice'' properties (see the definition in \Cref{def:nice-pinning}). 
The conditional-revealing step reveals the value of $v_i^*$ conditional on $X_S$. 
If $X_S$ is ``nice'', 
then $v_i^*$ takes the value $\sigma_i^*$ with probability 
at least $n^{-\exp(-k^{4/5})}$ in the conditional-revealing step.

\paragraph{Classify variables} 
Before describing the detailed revealing process, 
we classify all variables in $V$ into \emph{good} variables and \emph{bad} variables. 
The random formula contains high-degree variables whose degree is significantly larger than 
the average degree $\alpha$.
Furthermore, 
since we consider the conditional distribution $\pi$ instead of $\mu_{\Phi}$, 
the values of all $v_j^*$ for $j < i$ are fixed by an adversary. 
We will find all \emph{bad variables} that contain all high-degree variables, 
fixed value variables, and other variables that are significantly affected by them. 
The procedure in \Cref{alg:identify-bad} for finding bad variables is inspired by the 
previous works \cite{GalanisGGY21,heImprovedBoundsSampling2023,ChenGGGHMM24,chenCountingRandomSAT2025} 
on sampling random CNF formula solutions. 
Additionally, we need to use the \emph{bounded intersection} property (\Cref{lem:subset-intersect-new}) 
to control the effect of fixed value variables $v_1^*, \ldots, v_{i-1}^*$.

\paragraph{Pre-revealing step}
The first step is a standard ``BFS'' revealing process starting from $v_i^*$. 
We keep revealing values of some \emph{good variables} 
and removing all clauses that are satisfied by the current revealing results. 
The process stops once we can find some set of clauses $\+C' \subseteq \+C$ such that, 
conditional on the revealing results $X_S$, 
the distribution of $X_{v_i^*}$ depends only on variables and clauses in $\+C'$ 
but not on other variables and clauses.
See \Cref{sec:conditional-independence} for the formal analysis.
Roughly speaking, the revealing result $X_S$ is ``nice'' if
\begin{itemize}
  \item almost all clauses $c \in \+C'$ contains at least $2 k^{4/5}$ unrevealed (either good or bad) variables.
  \item the size of $\+C'$ is bounded by $\log n$.
\end{itemize}
The first item will be established in \Cref{sec:lower-bound-good-variables},
while the second item will be established in \Cref{sec:size-of-associated-component}.
Since the pre-revealing step only reveals \emph{good variables}, 
some local-lemma-based analysis can show that 
the revealing result $X_S$ is ``nice'' with a 
constant probability.

The formal analysis of the pre-revealing step is 
given in the proof of \Cref{lemma:random-cnf-marginal-lower-bound-nice-probability}.

\paragraph{Conditional-revealing step}
Our purpose now is to lower bound the probability of $v_i^*$ taking the value $\sigma_i^*$ conditional on a ``nice'' $X_S$. 
However, after the pre-revealing step, 
most of the good variables in $\+C'$ are revealed. 
Some clauses may only have $2 k^{4/5} = o(k)$ unrevealed variables. 
Some unrevealed variables can be the bad variables with an unbounded degree.
Hence, the local uniformity argument no longer works for analyzing the variable $v_i^*$ 
because the local lemma condition totally breaks down. 

We overcome this challenge by using the structural property of the clause set $\+C'$. 
We show the following property (\Cref{property:degree-one-variable}) for a typical random CNF formula.
For any subset $\widehat{\+C} \subseteq \+C$ of clauses with size  $2 \leq |\widehat{\+C}| < \log n$, 
one can always find two clauses $c_1,c_2 \in \+C$ such that 
\begin{align}\label{eq:degree-one-variable-property-overview}
\forall i \in \set{1, 2}, \quad  \Big | \vbl(c_i) \setminus 
  \cup_{c' \in \widehat{\+C} \setminus \set{c_i}} \vbl(c') \Big | \ge k - k ^{4/5}.
\end{align}
In words, if we consider the sub-formula induced by the clause set $\widehat{\+C}$, then both clauses $c_1$ and $c_2$ contain many degree-one variables. Suppose $c_i$ is a clause that \emph{forbids} some assignment $\tau \in \set{\true, \false}^{\vbl{(c_i)}}$. Let $S_i$ be the set of degree-one variables that only belong to $c_i$. These variables behave like variables in a \emph{monotone} CNF formula: under any condition, each $v \in S_i$ takes the satisfying value $\neg \tau(v)$ with probability at least $\frac{1}{2}$. Hence, if we reveal all variables in $S_i$, then $c_i$ is satisfied with probability at least $1 - (\frac{1}{2})^{|S_i|}$. Intuitively, the degree-one variable can prove a \emph{one-sided} marginal lower bound, which turns out to be enough for our analysis.

Back to the conditional-revealing step. Suppose $X_S$ is ``nice''. 
Using~\eqref{eq:degree-one-variable-property-overview}, 
even if many variables are revealed in the pre-revealing step, 
we can still find a clause $c \in \+C'$ such that 
$c$ contains at least $k^{4/5}$ \emph{unrevealed degree-one} variables. 
By revealing all these degree-one variables, $c$ is satisfied 
with probability at least $1 - (\frac{1}{2})^{k^{4/5}}$ and we can then remove $c$.
Note that~\eqref{eq:degree-one-variable-property-overview} holds for \emph{all} subsets of clauses 
with size at most $\log n$. 
We use this argument recursively to remove all clauses in $\+C'$ 
with probability at least $(1 - (\frac{1}{2})^{k^{4/5}})^{\log n} \approx n^{-\exp(-k^{4/5})}$. 
After that, $v_i^*$ becomes a isolated variable and it takes the value $\sigma_i^*$ with probability $\frac{1}{2}$. 

In the formal analysis, 
we need to pay some special attention if $v_1^*,v_2^*,\ldots,v_i^*$ are one of the degree-one variables. 
We may also need to deal with the last clause separately. 
The formal analysis of the conditional-revealing step is 
given in the proof of \Cref{lemma:random-cnf-marginal-lower-bound-nice}.

Finally, 
we remark that~\eqref{eq:degree-one-variable-property-overview} is related to 
the locally tree-like property proved in~\cite{ChenGGGHMM24}. 
If clauses $\+C'$ form a tree, then due to the bounded intersection between clauses, 
one can find clauses $c_1,c_2 \in \+C'$. 
However, the locally tree-like property says that $\+C'$ is a tree with a constant number of extra clauses. 
We believe it is possible to derive our above proof from the locally tree-like property, 
but one would need to analyze the effect of these extra clauses very carefully, 
especially when $\+C'$ is reduced to a constant size. 
However, the property in~\eqref{eq:degree-one-variable-property-overview} 
provides a more direct route to the desired bound, resulting in a simpler proof.

\subsection{Lower bound of sample complexity}
The $\Omega(\log n)$ sample complexity lower bound in \Cref{theorem:lower-bound-simple-thm} can be established by Fano's inequality on $(k, 1, 0)$-CNF formulas with disjoint clauses. We remark that $\Omega(\log n)$ is standard for learning MRFs, which also appeared in~\cite{SanthanamW12,BreslerMS13}.   

For CNF formulas satisfying a strong local lemma condition $k = d \gg \log d$ but with large $s = k - 1$ intersections, in  \Cref{definition:lower-bound-gadgets}, we construct a $(k, k, k-1)$-CNF formula $\Phi = (V, \+C)$ with $k\ell$ variables which violates the resilience property. Specifically, there exists an $k$-variable clause $c^* \notin \+C$ with forbidden assignment $\sigma^*$ such that 
\begin{align}\label{eq:lower-bound-resilience-property-overview}
  0 < \mathbb{P}_{X \sim \mu_{\Phi}}[X_{\vbl(c^*)} = \sigma^*] = \exp(-\Theta_k(\ell)).
\end{align}

To obtain the lower bound of exact learning in \Cref{proposition:lower-bound-general-thm}, we simply take $k\ell = n$.
This implies that even distinguishing $\Phi$ from the perturbed formula $\Phi' = (V, C \cup \set{c^*})$ requires exponentially many samples, since their total variation distance $\dtv(\mu_{\Phi}, \mu_{\Phi'}) \le \exp(-\Omega_k(n))$.

Next, we sketch the proof of the lower bound of approximate learning in \Cref{theorem:lower-bound-general-thm}. 
We first provide some intuition.
Let $\varepsilon_0 > 0$ be the desired error bound, and let $M$ be an integer.
We use the above $(k, k, k-1)$-CNF formula $\Phi$ with $\ell  = \Theta_k(\log \frac{M}{\varepsilon_0})$ variables as a \emph{gadget} to construct a family $\+X$ of CNF formulas.
Each CNF formula $\Phi_{\text{hard}} = (V_{\text{hard}},\+C_{\text{hard}}) \in \+X$ contains $n = M\cdot k\ell =\Theta_k(M\log \frac{M}{\varepsilon_0})$ variables and $M$ disjoint set of clauses $\+C_1, \+C_2, \dots, \+C_M$ where $\+C_{\text{hard}} = \uplus_{i=1}^M \+C_i$. Each set of clauses $\+C_i$ either forms the gadget $\Phi$ or the gadget $\Phi'$, where $\Phi'$ is obtained from $\Phi$ by adding the clause $c^*$ in~\eqref{eq:lower-bound-resilience-property-overview}.
Hence, there are $2^M$ different CNF formulas in the family $\+X$.
Note that $\mu_{\Phi_{\text{hard}}}$ is a product distribution of $M$ independent components, the distribution on each  component is either $\mu_{\Phi}$ or $\mu_{\Phi'}$.
Using~\eqref{eq:lower-bound-resilience-property-overview}, we have $\dtv(\mu_{\Phi}, \mu_{\Phi'})\approx \exp(-\Theta_k(\ell))$.

Consider the following problem. Let $\mu_{\Phi_{\text{hard}}}$ be a CNF formula in $\+X$. Given \iid uniform solutions from $\mu_{\Phi_{\text{hard}}}$, the algorithm needs to learn a $\Phi_{\text{out}} \in \+X$ such that $\dtv(\mu_{\Phi_{\text{hard}}}, \mu_{\Phi_{\text{out}}}) \le \varepsilon_0$. We prove the information-theoretic lower bound on the sample complexity of this problem. The sample complexity lower bound can be easily extended to the case when the algorithm is allowed to output an arbitrary CNF formula. The intuition of our proof is based on the following two facts.
\begin{itemize}
  \item First, to approximately learn $\Phi_{\text{hard}}$, the algorithm needs to correctly learn at least a \emph{linear} portion of $M$ gadgets in the CNF formula $\Phi_{\text{hard}}$. Intuitively, to satisfy this property, the total variation distance $\dtv(\mu_{\Phi}, \mu_{\Phi'})$ between two types of gadgets should be \emph{large} enough. Otherwise, suppose the total variation distance $\dtv(\mu_{\Phi}, \mu_{\Phi'})$ is too small. Then all CNF formulas in $\+X$ are almost the same, and the approximate learning problem is trivial because the algorithm can output an arbitrary CNF formula in $\+X$.
  \item Next, even to learn a single gadget is $\Phi$ or $\Phi'$, the algorithm needs at least $\approx \frac{1}{\dtv(\mu_{\Phi}, \mu_{\Phi'})}$ samples to distinguish two types of gadgets. To obtain a better lower bound, we hope that the total variation distance $\dtv(\mu_{\Phi}, \mu_{\Phi'})$ is as \emph{small} as possible.
\end{itemize}
By setting  $\ell = \Theta_k(\log \frac{M}{\varepsilon_0})$, we can balance the above two constraints. Moreover, by properly choosing the constant (depending on $k$) hidden in $\Theta_k(\cdot)$, we can guarantee $\dtv(\mu_{\Phi}, \mu_{\Phi'})\approx (\frac{\varepsilon_0}{M})^{1 - o_k(1)}$. Recall that $n = \Theta_k(M \log\frac{M}{\varepsilon_0})$. Hence, we need roughly $\widetilde{\Omega}_{k, \varepsilon_0}(n^{1-o_k(1)})$ samples.

The above construction and analysis resemble \emph{Assouad's Lemma} in~\cite{MR777600}, which is often used to derive lower bounds on the sample complexity in the context of the \emph{minimax risk}. 
However, in our lower bound results, we consider learning CNF formulas with $\varepsilon_0$-error and $\frac{1}{3}$ success probability.
In \Cref{sec:sample-complexity-of-approx-learning-cnf-formulas-in-the-local-lemma-regime}, we formalize the above proof idea by analyzing the size of $\varepsilon_0$-balls of $\+X$ under the total variation distance metric and applying a distance-based variant of \emph{Fano's inequality} (\Cref{lemma:distance-fano-inequality}) to obtain the desired lower bound.

\subsection{Obstacles in applying previous MRF learning algorithms}\label{sec:obstacles-in-applying-previous-mrf-learning-algorithms}

Finally, 
we discuss some technical challenges in 
applying previous MRF learning algorithms to our setting, learning CNF formulas from \iid uniform solutions.

Bresler, Mossel, and Sly~\cite{BreslerMS13} proposed an algorithm to learn MRFs by enumerating all neighbors of each variable. Consider a $(k,d)$-CNF formula, for any $v \in V$, let $N(v)$ denote all neighbors $u$ of $v$ such that $\{u,v\} \subseteq \vbl(c)$ for some clause $c \in \+C$. Note that $N(v)$ is the \emph{Markov blanket} of $v$ and $|N(v)| \leq kd$. Their technique needs \emph{at least} the following condition. For each $u \in N(v)$, there is an assignment $\sigma$ on $N(v) \setminus \set{u}$ such that $\sigma$ occurs with a constant probability in $\mu_{\Phi}$ and conditional on $\sigma$, $u$ has a constant influence on $v$. 
As $\sigma$ can be a configuration of about $kd \approx k\cdot 2^k$ variables, verifying the constant probability lower bound for $\sigma$ seems more challenging than verifying our resilience property on $k$ variables. Furthermore, even if one can verify their condition, their algorithm runs in time at least $n^{O(kd)}$, but Valiant's algorithm runs in time $n^{O(k)}$. 
The maximum degree $d \approx 2^k$ in the local lemma regime and $d \approx \frac{\log n}{\log \log n}$ for random CNF formulas.

A faster algorithm based on the \emph{correlation decay} was also proposed in~\cite{BreslerMS13}. This algorithm requires that for two variables $u$ and $v$, their correlation is small if $u$ and $v$ are far away from each other in the underlying hypergraph of MRF, and their correlation is large if $u$ and $v$ are in the same clause. The correlation decay (weak spatial mixing) property indeed holds for CNF formulas in the local lemma regime~\cite{Moitra19}. However, we give a counterexample in \Cref{sec:counterexample-correlation-lower-bound} to show that the correlation between $u,v$ can be 0 even if they are in the same clause.

Bresler~\cite{Bresler15} proposed an algorithm to learn general Ising models. Later works improve the sample complexity and extend the result to MRFs~\cite{KlivansM17,HamiltonKM17,WuSD19}. However, these techniques work for MRFs with \emph{soft constraints} because they require a bound on the strength of local interactions, i.e. bounded width assumption. Also, all the above techniques require a bounded degree on the underlying graph of MRFs, which is not the case for random CNF formulas. Subsequent works~\cite{gaitondeUnifiedApproachLearning2024,chandrasekaranLearningSherringtonKirkpatrickModel2025} extend to learning the Sherrington-Kirkpatrick model, which is beyond the bounded width assumption. Their results rely on the concentration properties of the interaction matrix, which seems not applicable to CNF formulas. It is interesting to see (but not clear now) if these techniques can be generalized to the problems studied in this paper.

For MRFs with \emph{pair-wise} hard constraints, e.g., the hardcore model and graph coloring, \cite{BreslerGS14a,BlancaCSV20} proposed algorithms based on the resilience property. 
Verifying the resilience property for MRFs with pair-wise hard constraints is not challenging. 
However, CNF formulas are MRFs defined by \emph{high-order} local interactions, and we need new techniques to deal with them.

\section{Resilience of CNF formulas in the local lemma regime}\label{sec:resilience-bounded-intersection}

In this section, we establish the resilience property of $(k, d, s)$-CNF formulas.
We first prove that when the intersection size between any two clauses is sublinear in $k$,
the $(k, d, s)$-CNF formula is $O(1)$-resilient under the optimal local lemma condition,
that is, when $k \gtrsim \log d$. Formally,
\begin{lemma}\label{lem:resilience-bounded-sublinear}
  Let $\eta \in (0, 1)$ be a constant.
  For any integers $k,d,s$ satisfying $s = k^{1-\eta}$ and  $k\ge \log d+ o(k) + O_\eta(1)$ (in particular, $k \geq 2^{\frac{2}{\eta}}$),
  the $(k, d, s)$-CNF formula is $O_{k,\eta}(1)$-resilient.
\end{lemma}
\noindent
Combining this lemma with \Cref{pps:resilience},
it is straightforward to prove \Cref{thm:main-bounded}.

\learningSublinearIntersection*

Moreover, we show that if the local lemma condition is relaxed,
the $O(1)$-resilience still holds even when the intersection size is linear in $k$,
which directly implies \Cref{thm:main-linear-bounded}.
Formally,

\begin{lemma}\label{lem:resilience-bounded-linear}
  Let $\zeta \in (0, 1)$ be a constant.
  For any integers $k,d,s$ satisfying $s = \zeta k$ and $k \geq C \log d + O(C \log k +\frac{1}{\sqrt{\zeta}})$, where
  \begin{align}\label{eq:C-bounded-linear}
    C \triangleq
    \begin{cases}
      \frac{1}{1 - \sqrt{2 \zeta}},
       & \zeta \in (0, 3 - 2\sqrt{2}),
      \\ \frac{2}{1 - \zeta}, & \zeta \in [3 - 2\sqrt{2}, 1),
    \end{cases}
  \end{align}
  the $(k, d, s)$-CNF formula is $O_{k,\zeta}(1)$-resilient.
\end{lemma}

\learningLinearIntersection*

In the rest of this section, 
we focus on proving \Cref{lem:resilience-bounded-sublinear} and \Cref{lem:resilience-bounded-linear}.

\subsection{Preliminaries of Lov\'asz local lemma}

Before we prove the resilience property of CNF formulas, we introduce some standard tools. 

Let $\+R = \{R_1,R_2,\ldots,R_k\}$ be a set of mutually independent random variables.
For any event $E$, we use $\vbl(E) \subseteq \+R$ to denote the set of random variables that $E$ depends on. Define a set of bad events $\+B = \{B_1,B_2,\ldots,B_m\}$. For any event $B \in \+B$, define the neighborhood of $B$ as $\Gamma(B) = \{B' \in \+B \,|\, B' \neq B \land \vbl(B') \cap \vbl(B) \neq \emptyset\}$. For any event $E \notin \+B$, similarly define $\Gamma(E) = \{B \in \+B \,|\, \vbl(B) \cap \vbl(E) \neq \emptyset\}$. Let $\Pr[\+R]{\cdot}$ denote the product distribution over $\+R$. We use the following version of Lov\'asz local lemma in~\cite{HaeuplerSS11}.

\begin{theorem}[\cite{HaeuplerSS11}]\label{thm:lovasz-local-lemma}
  If there exists a function $x: \+B \to (0, 1)$ such that for any $B \in \+B$,
  \begin{align*}
    \Pr[\+R]{B} \leq x(B) \prod_{B' \in \Gamma(B)} (1 - x(B')),
  \end{align*}
  then it holds that $\Pr[\+R]{\wedge_{B \in \+B} \bar{B}} \ge \prod_{B \in \+B} (1 - x(B)) > 0$.

  Moreover, for any event $E$, it holds that
  \begin{align*}
    \Pr[\+R]{E \,|\, \wedge_{B \in \+B} \bar{B}} \leq \Pr[\+R]{E} \cdot \prod_{B \in \Gamma(E)} (1 - x(B))^{-1}.
  \end{align*}
\end{theorem}

For CNF formula $\Phi = (V, \+C)$, 
consider the product distribution $\+R$ that every variable takes True or False independently with probability $1/2$ and the bad events $\+B = \{B_c \,|\, c \in \+C\}$, 
where $B_c$ is the event that the clause $c$ is not satisfied. 
Using the Lov\'asz local lemma~\Cref{thm:lovasz-local-lemma}, 
the following local uniformity property for CNF formulas is well-known~\cite{Moitra19,fengFastSamplingCounting2021}.

\begin{lemma}[\cite{Moitra19,fengFastSamplingCounting2021}]\label{lem:local-uniformity}
  Let $\Phi =(V, \+C)$ be a CNF formula.
  Assume each clause contains at least $k_1$ variables and at most $k_2$ variables,
  and each variable belongs to at most $d$ clauses.
  For any $t \geq k_2$, if $2^{k_1} \ge 2\e dt$,
  then there exists a satisfying assignment for $\Phi$ and for any $v \in V$,
  \begin{align*}
    \max\set{\Pr[X\sim\mu_{\Phi}]{X_v=1}, \Pr[X\sim\mu_{\Phi}]{X_v=0}}
    \le \frac{1}{2} \exp\tuple{\frac{1}{t}}.
  \end{align*}
\end{lemma}

\subsection{A general approach to establish resilience property}

We first give a general approach to establish the resilience property in \Cref{def:resilience}.
Then we use this approach to establish \Cref{lem:resilience-bounded-simple} and \Cref{lem:resilience-bounded-condition} in \Cref{sec:sublinear-intersection} and \Cref{sec:linear-intersection}, respectively.

Let $c^*$ be the clause in \Cref{def:resilience},
whose variables are $\vbl(c^*) = \set{v^*_1, \ldots, v^*_k}$ and forbidden assignment is $\sigma^* = (\sigma^*_1, \ldots, \sigma^*_k)$.
Suppose we want to verify the resilience property of a CNF formula $\Phi=(V, \+C)$ with respect to $c^*$.
We start the proof by a simple case that there exists a clause $c' \in \+C$
such that $\vbl(c') = \vbl(c^*)$ but $c' \neq c^*$.

\begin{lemma}\label{lem:resilience-bounded-simple}
  Let $\Phi = (V, \+C)$ be a $(k, d, s)$-CNF formula satisfying $k\ge \log d+ \log k + \log (2\mathrm{e}) +s$,
  and $c^*\notin \+C$ be a clause.
  If there exists a clause $c' \in \+C$ with $\vbl(c') = \vbl(c^*)$ but $c' \neq c^*$,
  then it holds that
  \begin{equation*}
    \Pr[X\sim \mu_{\Phi}]{X_{\vbl(c^*)} = \sigma^*}
    \ge \tuple{1 - \frac{1}{2}\exp\tuple{\frac{1}{k}}}^k.
  \end{equation*}
\end{lemma}

\begin{proof}
  In this case, $\sigma^*$ is a satisfying assignment of the clause $c'$ and there exists a variable,
  say $v^*_1$ without loss of generality,
  such that $c'$ is satisfied when $X_{v^*_1} = \sigma^*_1$.
  Since each clause in $\+C$ contains exactly $k$ variables,
  the marginal lower bound
  $\Pr[X\sim \mu_{\Phi}]{X_{v_1^*} = \sigma^*_1} \ge 1-\exp\tuple{k^{-1}}/2$
  follows from \Cref{lem:local-uniformity} with $2^k \ge 2\e dk$.
  Conditioning on the event that $X_{v^*_1} = \sigma^*_1$,
  the CNF formula can be simplified by
  removing clauses that have been satisfied
  and removing $v^*_1$ from clauses containing $v^*_1$ with forbidden value $\sigma^*_1$.
  We then pin $v^*_i$ with $\sigma^*_i$ from $i=2$ to $k$.
  Since any other clause in $\+C$ shares at most $s$ variable with $c'$
  (and therefore also with $c^*$),
  the size of each clause is always at least $k - s$ and at most $k$
  during the simplification process.
  By \Cref{lem:local-uniformity} with $2^{k-s} \ge 2\e dk$,
  \begin{equation*}
    \Pr[X\sim \mu_{\Phi}]{ X_{v^*_i} = \sigma^*_i \mid \sigma^*_{\le i - 1}}
    \ge 1 - \frac{1}{2}\exp\tuple{\frac{1}{k}},
  \end{equation*}
  where $\sigma^*_{\le i - 1} = \bigwedge_{j=1}^{i-1} (X_{v^*_j} = \sigma^*_j)$.
  Note that if $k\ge \log d+ \log k + \log (2\mathrm{e}) +s$, then the conditions $2^{k-s} \ge 2\e dk$ are always satisfied,
  which completes the proof.
\end{proof}

Assuming that $\vbl(c) \neq \vbl(c^*)$ holds for all $c \in \+C$,
our next strategy is to eliminate clauses that have a ``large'' intersection with $c^*$,
which are ``rare,''
by pinning certain variables outside $\vbl(c^*)$ to satisfy them.
The intuition is that when we sequentially pin the variables $c_i^*$ to the values $\sigma_i^*$ for $i = 1, \dots, k$,
all remaining clauses share only ``few'' variables with $c^*$,
allowing us to control the clause lengths during the simplification process.
We formalize this idea as the following condition.
\begin{condition}\label{cond:pin-sequence}
  Let $t_1$ and $t_2$ be two positive integers. Assume the following conditions hold for CNF formula $\Phi=(V, \+C)$ and a clause $c^*\notin \+C$.
  Let $\widetilde{\+C}\triangleq \set{c \in \+C \mid \abs{\vbl(c) \cap \vbl(c^*)} \ge t_1 }$. Then
  \begin{itemize}
    \item the size of $\widetilde{\+C}$ is at most $t_2$;
    \item there exists a sequence of variables
          $u_1, u_2, \ldots, u_\ell\notin \vbl(c^*)$ together with a sequence of values $\tau_1, \tau_2, \ldots, \tau_\ell \in \{\text{True}, \text{False}\}$, where $\ell \le |\widetilde{\+C}|$,
          such that pinning $u_i$ with $\tau_i$ for all $1 \leq i \leq \ell$ satisfies all the clauses in $\widetilde{\+C}$.
  \end{itemize}
\end{condition}

\begin{lemma}\label{lem:resilience-bounded-condition}
  Assume that \Cref{cond:pin-sequence} holds with parameters $t_1$ and $t_2$.
  For any integer $k,d,s\ge 1$,
  satisfying that
  $k \geq \log d + \log k + \log (2\mathrm{e}) + t_1+t_2$,
  the $(k, d, s)$-CNF formula is $O_{k,t_2}(1)$-resilient.
\end{lemma}

\begin{proof}
  First, we pin $u_i$ with the value $\tau_i$ from $i=1$ to $\ell$ one by one.
  After each pinning, the CNF formula can be simplified,
  and the size of each clause is always at least
  $k - t_2$ (because $|\widetilde{\+C}| \le t_2$) and at most $k$
  during the simplification process.
  By \Cref{lem:local-uniformity} with $2^{k-t_2} \ge 2\e dk$,
  \begin{equation*}
    \Pr[X\sim \mu_{\Phi}]{\bigwedge_{i=1}^\ell \tuple{X_{u_i} = \tau_i}}
    \ge \tuple{1 - \frac{1}{2}\exp\tuple{\frac{1}{k}}}^\ell
    \ge \tuple{1 - \frac{1}{2}\exp\tuple{\frac{1}{k}}}^{t_2}.
  \end{equation*}
  Conditioning on the this pinning, all clauses in $\widetilde{\+C}$ are satisfied (hence, removed) and
  the remaining clauses shares at most $t_1$ variables with $c^*$ due to the definition of $\widetilde{\+C}$.
  Therefore, while pinning $v^*_i$ with $\sigma^*_i$ from $i=1$ to $k$,
  the size of each clause is always at least $k - t_1-t_2$ and at most $k$.
  By \Cref{lem:local-uniformity} with $2^{k - t_1-t_2} \ge 2\e dk$,
  we have
  \begin{equation*}
    \Pr[X\sim \mu_{\Phi}]{X_{\vbl\tuple{c^*}} = \sigma^*
      \mid \bigwedge_{i=1}^\ell \tuple{X_{u_i} = \tau_i}}
    \ge \tuple{1 - \frac{1}{2}\exp\tuple{\frac{1}{k}}}^{k}.
  \end{equation*}
  To satisfy the condition of \Cref{lem:local-uniformity}, it suffices to assume $k \geq \log d + \log k + \log (2\mathrm{e}) + t_1+t_2$.
  Combining the two lower bounds, $\Phi$ is $\tuple{1 - \frac{1}{2}\exp\tuple{\frac{1}{k}}}^{k+t_2}$-resilient, completing the proof.
\end{proof}

\Cref{lem:resilience-bounded-simple} and \Cref{lem:resilience-bounded-condition} provide a general approach to establish resilience property.
To use \Cref{lem:resilience-bounded-condition}, we need to verify \Cref{cond:pin-sequence}.
The following lemma provides useful structural properties to verify the condition.
Intuitively, the lemma says that for a set family, if every set in the family is large and the intersection of any two sets in the family is small, then the number of sets in the family is small.

\begin{lemma}\label{lem:subset-intersect-new}
  Let $p \ge 1$ and $q \leq k$
  such that $\frac{k}{p} + \frac{pq}{2} \le k$.
  For any set family $\+S \subseteq 2^{[k]}$ with ground set $[k]$ satisfying that
  $\abs{S} \ge \frac{k}{p} + \frac{pq}{2}$ for any $S \in \+S$ and
  $\abs{S \cap S'} \le q$ for any $S, S' \in \+S$,
  it holds that $\abs{\+S} \le p$.
\end{lemma}
\begin{proof}
  We first assume that $p$ is an integer. Suppose by contradiction that $\abs{\+S} \ge p+1$.
  Consider a sub-set-family $\+S' \subseteq \+S$ with $\abs{\+S'} = p+1$.
  On one hand, we have $\abs{\bigcup_{S\in \+S'} S} \le k$,
  since $\bigcup_{S\in \+S'} S \subseteq [k]$.
  On the other hand,
  \begin{equation*}
    \sum_{S_i\in \+S'} \abs{S_i} - \sum_{S,S'\in \+S'} \abs{S \cap S'}
    \ge (p + 1) \tuple{\frac{k}{p} + \frac{pq}{2}} - \binom{p + 1}{2} q = \frac{(p+1)k}{p} > k.
  \end{equation*}
  However, by the inclusion-exclusion principle,
  $k \geq \abs{\bigcup_{S\in \+S'} S} \ge \sum_{S\in \+S'} \abs{S} - \sum_{S,S'\in \+S'} \abs{S \cap S'}$,
  which yields a contradiction.
  The case when $p$ is not an integer can be proved similarly by considering a sub-set-family $\+S'$ of size $|\+S'| = \ceil{p}$.
  The contradiction follows since
  $$\sum_{S_i\in \+S'} \abs{S_i} - \sum_{S,S'\in \+S'} \abs{S \cap S'}
    \ge \ceil{p} \tuple{\frac{k}{p} + \frac{pq}{2}} - \binom{\ceil{p}}{2} q > \frac{\ceil{p}k}{p} > k.$$
  Combining the two cases completes the proof.
\end{proof}

\begin{corollary}\label{cor:bound-large-intersection}
  Given a $(k,d,s)$-CNF formula $\Phi=(V, \+C)$ and a clause $c^*$ with $|\vbl(c^*)| = k$,
  let $p \ge 1$ such that $\frac{k}{p} + \frac{ps}{2} \le k$ and
  $\widetilde{\+C}\triangleq \{c \in \+C \vert \abs{\vbl(c) \cap \vbl(c^*)} \ge \frac{k}{p} + \frac{ps}{2} \}$,
  it holds that $|\widetilde{\+C}| \le p$.
\end{corollary}
\begin{proof}
  Define the set
  \begin{equation*}
    \+S = \set{ \vbl(c)\cap \vbl(c^*) \mid c \in \+C \land \abs{\vbl(c)\cap \vbl(c^*)} \ge \frac{k}{p} + \frac{ps}{2} }.
  \end{equation*}
  To prove the lemma, we need to show that $|\+S| = |\widetilde{\+C}|$. We claim that for any $c \in \widetilde{\+C}$, the set $\vbl(c)\cap \vbl(c^*)$ are distinct. Since any two clauses in $\+C$ share at most $s$ variables, it suffices to show that $\frac{k}{p} + \frac{ps}{2} > s$. The inequality holds trivially when $p \geq 2$. Assume $p \in [1,2)$. The inequality is equivalent to $s < \frac{k}{p(1-p/2)}$, which holds because $s \leq k$ and $0 < p(1-p/2) \leq \frac{1}{2}$ when $p \in [1,2)$.
\end{proof}

In the following, we give a detailed analysis
for sublinear and linear intersection, respectively.

\subsection{Sublinear intersection with the local lemma condition} \label{sec:sublinear-intersection}

In this subsection,
we establish the resilience property of $(k, d, s)$-CNF formulas
with sublinear intersection, i.e., $s=o(k)$,
under the optimal local lemma condition.

\begin{proof}[Proof of \Cref{lem:resilience-bounded-sublinear}]
  If there exists a clause $c' \in \+C$ such that $\vbl(c') = \vbl(c^*)$ but $c' \neq c^*$,
  applying \Cref{lem:resilience-bounded-simple} with $s=k^{1-\eta}$,
  we have the $(k, d, s)$-CNF formula is $O_{k}(1)$-resilient
  since $k\ge \log d+ \log k + \log (2\mathrm{e}) +k^{1-\eta}$ holds.
  In the following, we only need to check~\Cref{cond:pin-sequence}.

  Applying \Cref{cor:bound-large-intersection} with $p=k^{\frac{\eta}{2}}$ and $s=k^{1-\eta}$,
  it holds that there are at most $k^{\frac{\eta}{2}}$ clauses in $\+C$ share
  at least $\frac{3}{2}k^{1-\frac{\eta}{2}}$ variables with $c^*$,
  for any $k$ satisfying that $\frac{3}{2}k^{1-\frac{\eta}{2}} < k$.
  Note that the condition $\frac{3}{2}k^{1-\frac{\eta}{2}} < k$ holds since $k > 2^{\frac{2}{\eta}}=O_\eta(1)$.
  Thus, \Cref{cond:pin-sequence} holds with $t_1=\frac{3}{2}k^{1-\frac{\eta}{2}}$ and $t_2=k^{\frac{\eta}{2}}$ if we assume the pinning sequence $(u_i,\tau_i)_{i=1}^\ell$ exists.
  By \Cref{lem:resilience-bounded-condition},
  since $k\ge \log d+\log k + \log (2\mathrm{e})+\frac{3}{2}k^{1-\frac{\eta}{2}}+k^{\frac{\eta}{2}}$ holds,
  the $(k, d, s)$-CNF formula is $O_{k,\eta}(1)$-resilient.

  The remaining task is to prove the existence of the pinning sequence.
  The process to find the pinning sequence is as follows:
  Sort the clauses $c \in \widetilde{\+C} = \{c \in \+C: \abs{\vbl(c) \cap \vbl(c^*)} \ge \frac{3}{2}k^{1-\frac{\eta}{2}} \}$ in the increasing order
  by the number of variables $|\vbl(c) \setminus \vbl(c^*)|$ that are not in $\vbl(c^*)$ (break ties arbitrarily).
  Say the ordering is $c_1, c_2, \ldots, c_t$, where $t = |\widetilde{\+C}|$.
  Since $\vbl(c) \neq \vbl(c^*)$ holds for all $c \in \+C$,
  the first clause $c_1$ must contain a variable $u_1\notin \vbl(c^*)$
  (pick an arbitrary variable if there are multiple)
  and we pin it with value $\tau_1$ that can satisfy the clause $c_1$.
  Suppose we have processed the clause $c_i$.
  We find the smallest $j > i$ such that the clause $c_j$ is not satisfied by the previous pinned variables.
  We claim that there must exists an unpinned variable $u_j$ such that $u_j \in \vbl(c_j) \setminus \vbl(c^*)$
  and we pin $u_j$ with value $\tau_j$ that can satisfy the clause $c_j$.
  Repeating this process until all clauses in $\widetilde{\+C}$ are satisfied.
  It is easy to see that the number of pinned variables is at most
  $ |\widetilde{\+C}|$.

  We now prove the existence of $u_j$ by contradiction.
  Suppose after pinning $1 \leq r < \lfloor k^{\frac{\eta}{2}} \rfloor$ variables,
  we need to process a clause $c_j$ but its unpinned variables are all in $\vbl(c^*)$, where $j  > r$.
  Thus, $c_j$ has at least $k - r$ variables in $\vbl(c^*)$
  and so does each of the previous clause $\{c_1, c_2, \ldots, c_{j-1}\}$ in the sequence due to the sorting.
  Note that $r \leq |\widetilde{\+C}| \leq k^{\frac{\eta}{2}} < k$ holds. Hence, $k-r \geq 0$.
  Let $S_i = \vbl(c_i) \cap \vbl(c^*)$.
  Consider a subset $\+C' \subseteq \set{c_1, \ldots, c_j}$ of size $r+1$.
  On the one hand, we have $\abs{\bigcup_{i:c_i\in\+C'} S_i} \le k$ by the definition of $S_i$.
  On the other hand, since $|S_i \cap S_{i'}| \leq s =  k^{1 - \eta}$ holds for any $i, i' \in [j]$, we have
  \begin{align*}
    \sum_{i:c_i\in\+C'} \abs{S_i} - \sum_{i<i':c_i,c_i'\in\+C'} \abs{S_i \cap S_{i'}}
     & \ge (r+1)(k - r) - \binom{r+1}{2} \cdot k^{1 - \eta} \\
     & = k + r \tuple{k - \frac{r + 1}{2} k^{1 - \eta} - (r + 1)} >k,
  \end{align*}
  where the last inequality holds because $r+1\le k^{\frac{\eta}{2}}$ and $k > 2^{\frac{2}{\eta}}$.
  However, by the inclusion-exclusion principle,
  $\abs{\bigcup_{i:c_i\in\+C'} S_i}
    \ge \sum_{i:c_i\in\+C'} \abs{S_i} - \sum_{i<i':c_i,c_i'\in\+C'} \abs{S_i \cap S_{i'}} > k$,
  which yields a contradiction.

  Finally, we put all the conditions together to obtain $k \geq 2^{\frac{2}{\eta}}$ and
  \begin{align*}
    k\ge \log d+\log k + \log (2\mathrm{e})+\frac{3}{2}k^{1-\frac{\eta}{2}}+k^{\frac{\eta}{2}} = \log d + O(k^{1 - \frac{\eta}{2}}) . & \qedhere
  \end{align*}
\end{proof}

\subsection{Linear intersection with relaxed local lemma conditions} \label{sec:linear-intersection}

In this subsection, we show how to relax the local lemma condition
so that the resilience property of $(k, d, s)$-CNF formulas holds
even if the size of intersection $s$ between clauses is linear in $k$.

\begin{proof}[Proof of \Cref{lem:resilience-bounded-linear}]
  If there exists a clause $c' \in \+C$ such that $\vbl(c') = \vbl(c^*)$ but $c' \neq c^*$,
  applying \Cref{lem:resilience-bounded-simple} with $s=\zeta k$,
  we have the $(k, d, s)$-CNF formula is $O_{k}(1)$-resilient
  when $k\ge \log d+ \log k + \log (2\mathrm{e}) +\zeta k$ holds,
  which is guaranteed by the condition $k \geq C \log d + o(k)$ and $C\ge \frac{1}{1-\zeta}$ in \Cref{lem:resilience-bounded-linear}.
  In the following, we assume that $\vbl(c) \neq \vbl(c^*)$ holds for any $c \in \+C$.

  We start with a simple analysis which works for all $\zeta \in (0, 1)$. Then, we give an improved analysis for $\zeta \in (0, \frac{1}{2})$. Assume $\zeta \in (0, 1)$.
  Observe that there is at most one clause that shares more than $\frac{1+\zeta}{2}k$ variables with $c^*$,
  since otherwise there exist two clauses sharing more than $\zeta k$ variables with each other.
  If such a clause does exist, denote it as $c_0$ and
  we pin a variable $u$ in $\vbl(c_0)\setminus \vbl(c^*)$ with value $\tau$ that can satisfy $c_0$ (such $u$ exists due to $\vbl(c_0) \neq \vbl(c^*)$).
  By \Cref{lem:local-uniformity} with $2^{k} \ge 2\e dk$,
  we have
  \begin{equation*}
    \Pr[X\sim \mu_{\Phi}]{X_u=\tau}
    \ge \tuple{1 - \frac{1}{2}\exp\tuple{\frac{1}{k}}}.
  \end{equation*}
  Conditioning on this pinning,
  the remaining clauses share at most $\frac{1+\zeta}{2}k$ variables with $c^*$.
  Therefore, while pinning $v^*_i$ with $\sigma^*_i$ from $i=1$ to $k$,
  the size of each clause is always at least $k - \frac{1+\zeta}{2}k -1$ and at most $k$.
  By \Cref{lem:local-uniformity} with $2^{k - \frac{1+\zeta}{2}k -1} \ge 2\e dk$,
  we have
  \begin{equation*}
    \Pr[X\sim \mu_{\Phi}]{X_{\vbl\tuple{c^*}} = \sigma^*
      \mid \bigwedge_{i=1}^\ell \tuple{X_{u_i} = \tau_i}}
    \ge \tuple{1 - \frac{1}{2}\exp\tuple{\frac{1}{k}}}^{k}.
  \end{equation*}
  Combining these two lower bounds,
  we have $\Phi$ is $\tuple{1 - \frac{1}{2}\exp\tuple{\frac{1}{k}}}^{k+1}$-resilient.
  To make all the conditions hold during the application of \Cref{lem:local-uniformity},
  we need $k \ge \frac{2}{1 - \zeta} \log d + \frac{2}{1-\zeta}\log (4\mathrm{e}k)$.

  Now, we show how to improve the local lemma condition for small $\zeta\in (0, \frac{1}{2})$
  using \Cref{cond:pin-sequence} and \Cref{lem:resilience-bounded-condition}.
  Applying \Cref{cor:bound-large-intersection}
  with $p=\sqrt{2 / \zeta}$ and $s=\zeta k$,
  it holds that there are at most $\sqrt{2 / \zeta}$ clauses in $\+C$ sharing at least
  $\sqrt{2\zeta}k$ variables with $c^*$.
  Thus, \Cref{cor:bound-large-intersection} holds with $t_1=\sqrt{2\zeta}k$ and $t_2=\sqrt{2 / \zeta}$ if we assume the pinning sequence $(u_i,\tau_i)_{i=1}^{\ell}$ exists.
  By \Cref{lem:resilience-bounded-condition},
  the $(k, d, s)$-CNF formula is $O_{k,\zeta}(1)$-resilient if $k\ge \log d+\log(2\mathrm{e}k)+\sqrt{2\zeta}k+\sqrt{2 / \zeta}$ holds.

  The remaining task is to prove the existence of the pinning sequence.
  The argument is the same as the proof of \Cref{lem:resilience-bounded-sublinear}.
  The only difference is how to show the contradiction.
  Now, we have
  \begin{align*}
    \sum_{i:c_i\in\+C'} \abs{S_i} - \sum_{i<i':c_i,c_i'\in\+C'} \abs{S_i \cap S_{i'}}
     & \ge (r+1)(k - r) - \binom{r+1}{2} \cdot \zeta k               \\
     & = k + r \tuple{\tuple{1-\frac{\zeta(r+1)}{2}}k - (r + 1)} >k,
  \end{align*}
  where the last inequality holds since $r + 1 \le \sqrt{2 / \zeta}$, $\zeta\cdot \sqrt{2 / \zeta} < 2$ and $k > \frac{2}{\sqrt{2 \zeta} - \zeta}$.

  Finally, we put all the conditions together. Note that $ \frac{1}{1-\zeta} \leq \frac{1}{1-\sqrt{2\zeta}}\le \frac{2}{1-\zeta}$ holds for $\zeta\in (0, 3-2\sqrt{2})$. Recall $C$ is defined in \eqref{eq:C-bounded-linear}. Hence, the final condition is $k > \frac{2}{\sqrt{2\zeta} - \zeta}$ and
  \begin{align*}
    k\ge C \log d + C \log(4\mathrm{e}k) + \frac{2 + \sqrt{2}}{\sqrt{\zeta}}.
  \end{align*}
  Note that both $\frac{2}{\sqrt{2\zeta} - \zeta}$ and $\frac{2 + \sqrt{2}}{\sqrt{\zeta}}$ are at most $O(\frac{1}{\sqrt{\zeta}})$. We have
  \begin{align*}
    k \geq C \log d + O\left(C \log k +\frac{1}{\sqrt{\zeta}}\right). & \qedhere
  \end{align*}
\end{proof}
\section{Resilience of random CNF formulas}

In this section, we establish the resilience property for random CNF formulas and prove \Cref{theorem:learning-random-cnf}.
We first give a formal definition of random CNF formulas.
\begin{definition}[Random $k$-CNF formulas]
  \label{def:random-cnf}
  A random $k$-CNF formula $\Phi(k, n, m)$ with $n$ variables and $m \triangleq \floor{\alpha n}$ clauses is generated by selecting $m$ clauses independently and uniformly at random from all possible clauses over $n$ variables, where $\alpha > 0$ is the density of the formula.
  \begin{itemize}
    \item The variable set is defined as $V = \set{v_1, v_2, \ldots, v_n}$.
    \item Each clause is generated independently as a disjunction of $k$ literals, where each literal is sampled uniformly at random with replacement from the set of all $2n$ possible literals
    We denote by $\+C_{\Phi}$ the set of clauses in the formula $\Phi$.
    \item We use $H_{\Phi}$ to denote the hypergraph associated with formula $\Phi$, where each variable is a vertex and each clause is a hyperedge connecting at most $k$ vertices.
    \item We use $G_{\Phi}$ to denote the line graph of $\Phi$, i.e., each vertex in $G_{\Phi}$ corresponds to a clause in $\Phi$ and there is an edge between two vertices in $G_{\Phi}$ if and only if the corresponding two clauses share at least one variable.
    For a clause $c \in \+C$, let $N(c) \triangleq \set{c^\prime \in \+C \mid \vbl(c) \cap \vbl(c^\prime) \neq \emptyset}$
    be the set of neighbors of $c$ in $G_{\Phi}$.
    For a subset of clauses $\+C' \subseteq \+C$, 
    let $N(\+C') \triangleq \bigcup_{c \in \+C'} N(c)\setminus \+C'$ 
    denote all the one-step neighbors of clauses in $\+C'$ excluding clauses in itself.
  \end{itemize}
\end{definition}
\begin{remark}
  Every clause is a disjunction of literals, where every literal is $x$ or $\neg x$ for some variable $x$. For example, the clause $x_1 \lor \neg x_2 \lor x_3$ has variables $\set{x_1, x_2, x_3}$ and literals set $\set{x_1, \neg x_2, x_3}$.
  We use $\vbl(c)$ to denote the set of variables that appear in clause $c$ and $\deg(v)$ to denote the degree of variable $v$ in formula $\Phi$, i.e., the number of clauses that contain variable $v$.

  Note that repetitions of variables in clauses are allowed. Accordingly, we extend the standard definition of a $k$-CNF formula and continue to refer to each generated random formula as a $k$-CNF formula even if some clauses contain fewer than $k$ distinct variables.
\end{remark}

Our proof is represented in the following roadmap.
\begin{itemize}
    \item We first list some properties of CNF formulas. A CNF formula is said to be well-behaved (\Cref{def:well-behaved-random-cnf}) if it satisfies these properties. We show that with high probability, random CNF formulas are well-behaved~ in \Cref{lemma:well-behaved-random-cnf}.
    \item Next, for any \emph{fixed} well-behaved CNF formula $\Phi$, we show that it satisfies the resilience property with a large enough $\theta$ (\Cref{lemma:valiant-algorithm-marginal-lower-bound}), and Valiant's algorithm can learn $\Phi$ exactly with desired sample and computational complexities.
\end{itemize}

\subsection{Good properties and well-behaved CNF formulas}

\subsubsection{Good properties of CNF formulas}

We first list some properties of random CNF formulas, 
for which a random CNF formula with high probability satisfies.
The first two properties bound the minimum size of each clause and the 
maximum intersection size between any two clauses.
\begin{property}[Bounded clause size]
  \label{property:clause-size}
  Let $\Phi = (V, \+C)$ be a $k$-CNF formula.
  For each clause $c \in \+C$, $|\vbl(c)| \ge k - 2$.
\end{property}
\begin{property}[Bounded intersection]
  \label{property:clause-intersection}
  Let $\Phi = (V, \+C)$ be a $k$-CNF formula.
  For every two distinct clauses $c, c' \in \+C$, $\abs{\vbl(c) \cap \vbl(c')} \le 3$.
\end{property}

By \Cref{property:clause-intersection}, every two clauses share at most $3$ variables. 
So, at first glance, 
it appears that one could apply an argument similar to that 
for bounded intersection CNF formulas in \Cref{sec:resilience-bounded-intersection} 
to establish the resilience property for random CNF formulas. 
However, a more careful analysis is required. 
This is because the analysis in \Cref{sec:resilience-bounded-intersection} also requires a local lemma condition 
but in a random CNF formula $\Phi$, 
some variables may have a very large degree depending on $n$. 
In fact, with high probability, the maximum degree of variables in $\Phi$ is $\Theta(\log n)$.
\begin{fact}[{\cite[Lemma A.1]{chenCountingRandomSAT2025}}]
  \label{prop:random-cnf-variable-degree}
  With probability $1 - o(1/n)$ over the random $k$-CNF formula $\Phi = \Phi(k, n, m)$, 
  with density $\alpha$, the maximum degree of variables in $\Phi$ is at most $6 \log n + 4 k \alpha$.
\end{fact}

Analyzing the effect of high-degree variables is a challenging problem. 
We need to understand how the high degree vertices are distributed in the hypergraph $H_\Phi$ and 
how they affect the distribution of other variables.
A similar challenge arose in the previous works
\cite{heImprovedBoundsSampling2023,  chenCountingRandomSAT2025} to design sampling algorithms 
for random CNF formulas.
To tackle this challenge, previous works introduced a procedure 
\textsf{IdentifyBad}$(\Phi, \degreeConstant, \badclauseConstant)$
to find a set $\sbadvars$ of bad variables and a set $\sbadclauses$ of bad clauses.
Intuitively, 
$\sbadvars$ and $\sbadclauses$ are the set of variables and clauses 
that are significantly affected by the high degree variables.
To introduce this procedure, we define two thresholds.
Define a threshold $ \degreeConstant \alpha$ for high-degree variables, 
where $\degreeConstant$ is a constant to be determined.
Define another threshold $\badclauseConstant k$ to identify clauses 
that are significantly affected by high-degree variables, 
where $\badclauseConstant$ is another constant to be determined.
The procedure is given in \Cref{alg:identify-bad}. 
We remark that this procedure is only for analysis and will not be implemented in the learning algorithm.

\vspace{0.2cm}

\begin{algorithm}[H]
  \caption{\textsf{IdentifyBad}$(\Phi, \degreeConstant, \badclauseConstant)$ 
  \cite{heImprovedBoundsSampling2023, chenCountingRandomSAT2025}}\label{alg:identify-bad}
  \SetKwInOut{Input}{Input}
  \SetKwInOut{Output}{Output}
  \Input{a CNF $\Phi=(V, \+C)$, thresholds $\degreeConstant$ and $\badclauseConstant$;}
  \Output{a set of bad vertices $\sbadvars \subseteq V$ and a set of bad clauses $\sbadclauses \subseteq \+C$;}
  Initialize $\sbadvars \gets \set{v \in V \mid \deg(v) > \degreeConstant \alpha}$ 
  and $\sbadclauses = \emptyset$\;
  \While{$\exists c\in\+C \setminus\sbadclauses$ such that 
  $\abs{\vbl(c) \cap \sbadvars} > \badclauseConstant k$\label{line:identify-bad-while}}{
    Update $\sbadvars \gets\sbadvars \cup \vbl(c)$ and $\sbadclauses \gets \sbadclauses\cup\{c\}$;
  }
  \KwRet{$\sbadvars$ and $\sbadclauses$;}
\end{algorithm}

\vspace{0.2cm}

\Cref{alg:identify-bad} is a \emph{deterministic} procedure. In Line~\ref{line:identify-bad-while}, 
if there are multiple choices of $c$, 
we choose an arbitrary one (say, the smallest $c$ according to some ordering).
Define the set of good variables $\sgoodvars$ and the set of good clauses 
$\sgoodclauses$ as $\sgoodvars = V \setminus \sbadvars$ and $\sgoodclauses = \+C \setminus \sbadclauses$, 
respectively. The following observations are direct consequences of the above procedure.
\begin{observation}
  \label{observationL:good_variable_degree}
  For every good variable $v \in \sgoodvars$, it holds that $\deg(v) \le \degreeConstant \alpha$.
\end{observation}
\begin{observation}
  \label{observationL:good_clause_variable_count}
  For every good clause $c \in \sgoodclauses$, 
  it holds that $(1 - \badclauseConstant) k \le \abs{\vbl(c) \cap \sgoodvars} \le k$.
\end{observation}

The following property shows that for any connected clause sets with size at least $\log n$, 
it contains at most a linear fraction of bad clauses. 
Intuitively, it says that the high degree variables cannot make too many clauses bad for a fixed set of clauses. 
\begin{property}[Bounded bad clauses]
  \label{property:bounded-bad-clauses}
  Let $\Phi = (V, \+C)$ be a $k$-CNF formula and $\degreeConstant, \badclauseConstant, \eta$ be parameters.
  For any $\+C' \subseteq \+C$ of size $\abs{\+C'} \ge \log n$ connected in $G_\Phi$, it holds that
  $$\abs{\+C' \cap \sbadclauses} 
  \le \frac{12k^5}{(1-\eta)(\badclauseConstant - \eta)\degreeConstant} \abs{\+C'}.$$
\end{property}

We need a few other properties of random CNF formulas.
The following property is a standard bound on the ``growth rate'' of connected sets of clauses.
\begin{property}[Bounded growth rate]
  \label{property:bounded-growth-rate}
  Let $\Phi = (V, \+C)$ be a $k$-CNF formula.
  For every clause $c \in \+C$ and $\ell \ge 1$, 
  there are at most $n^3(\mathrm{e}k^2\alpha)^\ell$ many connected sets of clauses in $G_\Phi$ that 
  contain $c$ and have size $\ell$.
\end{property}
The next ``edge expansion'' property is also standard for random CNF formulas. 
We remark that this is a slightly stronger property than the one in 
{\cite[Property 3.5]{chenCountingRandomSAT2025}}. 
Roughly speaking, the property says a large subset of clauses should contain many distinct variables. 
\begin{property}[Edge expansion]
  \label{property:edge-expansion}
  Let $\Phi = (V, \+C)$ be a $k$-CNF formula and $\rho \in (0,1), \eta \in (0,1), B \geq 1$ be parameters.
  We say the CNF formula $\Phi$ satisfies the $(\rho, \eta, B)$-edge expansion property 
  if for any $\ell \le \rho \abs{\+C}$, any $\ell$ clauses $c_1, c_2, \ldots, c_\ell \in \+C$, 
  and any variable sets $S_1, S_2, \ldots, S_\ell$ satisfying that 
  $\forall i\in [\ell]$, $S_i\subseteq \vbl(c_i)$ and
  $\abs{S_i} \ge B$,
  it holds that
  \[\abs{\bigcup_{i\in [\ell]} S_i} > (1 - \eta)\cdot B\cdot \ell. \]
\end{property}

Finally, we introduce a novel structural property that 
characterizes the presence of clauses with degree-one variables in any small derived subformula.
Intuitively, this property ensures that within every small collection of clauses $\+C' \subseteq \+C$, 
there exist some clauses $c \in \+C'$ such that $c$ contains many variables $v$ 
such that $v$ only appears in $c$ but not in other clauses in $\+C' \setminus \set{c}$.
\begin{property}[Degree-one variable property]
  \label{property:degree-one-variable}
  Let $\Phi = (V, \+C)$ be a $k$-CNF formula and $\BetaConstant$ be a parameter in $(0, 1)$.
  We say the CNF formula $\Phi$ satisfies the 
  degree-one variable property with parameter $\beta$ 
  if for any subset of clauses $\+C' \subseteq \+C$ with size at least $2$ and at most $2 \log n$, 
  there exist at least two different clauses $c_1 ,c_2 \in \+C'$ such that
  $$\forall i \in \set{1, 2}, \quad  
  \abs{\vbl(c_i) \setminus \bigcup_{c' \in \+C' \setminus \set{c_i}} \vbl(c')} \ge \BetaConstant k.$$
\end{property}

\subsubsection{Well-behaved CNF formulas}

We now define well-behaved CNF formulas and then show that with high probability,
random CNF formulas we are interested in are well-behaved.
In \Cref{def:well-behaved-random-cnf}, 
$\degreeConstant, \badclauseConstant, \rho,\eta,\beta$ are parameters that appear  
in the definitions of good properties and $\zeta \in (0,1 - \badclauseConstant)$ is a new parameter.
\begin{definition}[Well-behaved random CNF formulas]
  \label{def:well-behaved-random-cnf}
  A $k$-CNF formula $\Phi = (V, \+C)$ is said to be $(k, \alpha, \degreeConstant, \badclauseConstant, \eta, \rho, \zeta, \beta)$-well-behaved if
  \begin{itemize}
    \item $\Phi$ satisfies \Cref{property:clause-size} (Bounded clause size);
    \item $\Phi$ satisfies \Cref{property:clause-intersection} (Bounded intersection);
    \item $\Phi$ satisfies \Cref{property:bounded-bad-clauses} (Bounded bad clauses)
    with parameters $\degreeConstant, \badclauseConstant, \eta$ when $\alpha > 1 / k^3$ (this property is not required when $\alpha \le 1 / k^3$);
    \item $\Phi$ satisifes \Cref{property:bounded-growth-rate} (Bounded growth rate);
    \item $\Phi$ satisifes \Cref{property:edge-expansion} (Edge expansion) 
    with two sets of parameters $(\rho, \eta, B_1)$ and $(\rho, \eta, B_2)$ 
    where $B_1 = k - 2$, $B_2 = (1 - \badclauseConstant - \zeta)k - 5 k^{4 / 5}$;
    \item $\Phi$ satisifes \Cref{property:degree-one-variable} (Degree-one variable property) with parameter $\BetaConstant$.
  \end{itemize}
\end{definition}

We set these parameters as follows:
\begin{equation}\label{eq:parameters}
\begin{aligned}
  \degreeConstant &= 12 k^{7}, \quad
  \badclauseConstant = k^{-1/5}, \quad
  \frozenEta = k^{-2/5}, \\[2pt]
  \frozenRho &= 2^{-k}, \quad
  \frozenConstant = 2 k^{-1/5}, \quad
  \BetaConstant = 1 - k^{-1/5}.
\end{aligned}
\end{equation}
The following lemma shows that with high probability,
random CNF formulas are well-behaved under the above parameter settings.
\begin{lemma}
  \label{lemma:well-behaved-random-cnf}
  Let $k \ge 10^5$ and $\alpha \le \frac{2^{k - 30 k^{4 / 5} \log k}}{2^{10} \e^2 k^8}$ be two constants.
  For any $n \geq n_0(\alpha, k)$, with probability $1 - o(1 / n)$, 
  the random formula $\Phi = \Phi(k, n, m)$ 
  is $(k, \alpha, \degreeConstant, \badclauseConstant, \frozenEta, 
  \frozenRho, \frozenConstant, \BetaConstant)$-well-behaved for parameters defined in \eqref{eq:parameters}.
\end{lemma}

\Cref{lemma:well-behaved-random-cnf} is proved by verifying the 
all properties hold with probability at least $1 - o(1/n)$. 
Most properties can be verified by either straightforward union bounds 
or the techniques in previous works. 
The only non-trivial property is the new degree-one variable property, 
which is proved in \Cref{lemma:random-cnf-degree-one-variable}. 
The verification of other properties together 
with the proof of \Cref{lemma:well-behaved-random-cnf} 
is deferred to \Cref{sec:proof-well-behavedness-random-cnf}.

\begin{lemma}
  \label{lemma:random-cnf-degree-one-variable}
  For any fixed $k$ and $\alpha$, let $\beta \in (0, 1)$.
  If $k > 8 / (1 - \beta)$ and $n$ are sufficiently large,
  for a random $k$-CNF formula $\Phi = \Phi(k, n, m)$, with high probability, 
  the following holds: every subset of clauses $\+C' \subseteq \+C_{\Phi}$ with size 
  at least $2$ and at most $2\log n$ satisfies that 
  there exist at least two different clauses $c_1 ,c_2 \in \+C'$ such that
  $$\forall i \in \set{1, 2}, \quad  \abs{\vbl(c_i) \setminus 
  \bigcup_{c' \in \+C' \setminus \set{c_i}} \vbl(c')} \ge \beta k.$$
\end{lemma}
\begin{proof}
  For any subset of clauses $\+C' \subseteq \+C$ of size $2\le r\le 2 \log n$, 
  we define the bad event $\+B_{\+C'}$ as the event that 
  at least $r - 1$ clauses $c \in \+C'$ satisfy 
  $\abs{\vbl(c) \setminus \bigcup_{c' \in \+C' \setminus \set{c}} \vbl(c')} < \beta k$.
  Denote the variables in $\bigcup_{c \in \+C'} \vbl(c)$ by $x_1, \ldots, x_N$  (repetitions allowed),
  where $N = rk$.
  For $i \in [n]$, let $R_i$ be the number of occurrences of $v_i$ in $x_1, \ldots, x_N$.
    
  Construct a simple graph $H$ on the vertex set $[N]$ by connecting vertices $i$ and $j$ 
  iff $x_i = x_j$. 
  Each connected component of $H$ is either an isolated vertex or a clique of size $R_v \ge 2$ for some $v$. 
  Fix a total order on $[N]$, and perform the following procedure to construct a new graph $F_H$: 
  visit each clique according to the total order. 
  Within each clique, remove all edges and deterministically select a spanning tree 
  by choosing the smallest vertex in the component as the root
  and connecting every other vertex in the clique to it. 
  As a result, $F_H$ is a forest on $[N]$ with 
  \begin{equation*}
    \abs{E(F_H)} = \sum_{i = 1}^n (R_{v_i}- 1) \cdot \^1[R_{v_i} \ge 2] 
    \ge \frac{1}{2}\sum_{i = 1}^n R_{v_i} \cdot \^1[R_{v_i} \ge 2]
  \end{equation*}
  Note that if the event $\+B_{\+C'}$ occurs, then
  $\sum_{i = 1}^n R_{v_i} \cdot \^1[R_{v_i} \ge 2] \ge (r - 1) \cdot (1 - \beta)k$,
  since each clause in $c \in \+C'$ satisfying 
  $\abs{\vbl(c) \setminus \bigcup_{c' \in \+C' \setminus \set{c}} \vbl(c')} < \beta k$ 
  contains at least $(1-\beta)k$ variables 
  that appear at least twice in $\set{x_1, \ldots, x_N}$.
  Let $t = \ceil{(1 - \beta)(r - 1)k / 2}$.
  Then,
  \begin{equation*}
    \Pr{\+B_{\+C'}} \le \Pr{\abs{E(F_H)} \ge t}=\sum_{i=t}^{N-1}\Pr{\abs{E(F_H)} = i}.
  \end{equation*}
    
  Recall that each variable is drawn independently and uniformly from $\{v_1, \ldots, v_n\}$, 
  ignoring the sign of the variable. 
  Hence, there are $n^N$ possible assignments of variables in $\{v_1, \ldots, v_n\}$ to the $x_i$'s. 
  Since the number of connected components in $F_H$ is $N - \abs{E(F_H)}$, 
  there are at most $n^{N - \abs{E(F_H)}}$ distinct ways to assign variables to these components.
  Moreover, the number of forests with $N$ vertices with $\abs{E(F_H)}$ edges is at most 
  $\binom{\binom{N}{2}}{\abs{E(F_H)}}$.
  
  Therefore, for sufficiently large $n$ such that $\binom{N}{2} / [(t+1)n] \le 1/2$,
  we have
  \begin{align*}
    \Pr{\+B_{\+C'}} 
    &\le \sum_{i=t}^{N-1} \binom{\binom{N}{2}}{i} \cdot\frac{ n^{N - i}}{n^N}
    \le \binom{\binom{N}{2}}{t} \cdot n^{-t}
    \sum_{i=t}^{N-1} \tuple{\frac{\binom{N}{2}}{(t+1)n}}^{i-t}\\
    &\le 2\cdot \binom{\binom{N}{2}}{t}\cdot n^{-t}
    \le 2\cdot \tuple{\frac{\mathrm{e} N^2}{2tn}}^t \\
    &\le 2\tuple{\frac{\mathrm{e}}{1-\beta}\cdot \frac{2k r}{n}}^{k' r/2}
  \end{align*}
  where $k' = (1-\beta)k/2$.
  Taking the union bound over all subset $\+C' \subseteq \+C$ with size at most $2 \log n$, 
  we have
  \begin{align*}
    \Pr{\bigcup_{\substack{\+C' \subseteq \+C, \\ \abs{\+C'} \le 2 \log n}} \+B_{\+C'}} 
    &\le 2\sum_{r = 2}^{\floor{2 \log n}} \binom{m}{r} 
    \tuple{\frac{\mathrm{e}}{1 - \beta} \cdot \frac{2kr}{n}}^{k' r/2}
    \le 2\sum_{r = 2}^{\floor{2 \log n}}
    \tuple{\frac{\e\alpha n}{r}\cdot \tuple{\frac{\mathrm{e}}{1 - \beta} \cdot \frac{2kr}{n}}^{k'/2}}^{r}\\
    &=2\sum_{r = 2}^{\floor{2 \log n}} \tuple{C_0\cdot \tuple{\frac{r}{n}}^{k'/2 - 1}}^{r}
    =o\tuple{\frac{1}{n}},
  \end{align*}
  where $C_0 = \mathrm{e} \alpha (2 \mathrm{e} k / (1 - \beta))^{k' / 2}$ 
  is a constant depending only on $k, \alpha, \beta$,
  and the last equation follows from $k'/2-1 > 1$.
  This completes the proof.
\end{proof}

\subsection{Resilience property for well-behaved CNF formulas}

\subsubsection{Learning well-behaved CNF formulas (Proof of \Cref{theorem:learning-random-cnf})}

With \Cref{lemma:well-behaved-random-cnf}, 
we now focus on a \emph{fixed} well-behaved CNF formula $\Phi$.
We show that given $T$ independent samples from $\mu_{\Phi}$, 
Valiant's Algorithm outputs a random CNF formula $\widehat{\Phi}$ 
such that $\mu_{\Phi} = \mu_{\widehat{\Phi}}$ with high probability when $T = O_k(n^{\exp(-\sqrt{k})})$, 
where the probability is taken over the randomness of the $T$ samples. 
Note that the output formula $\widehat{\Phi}$ 
must be a $k$-CNF formula but the random CNF formula $\Phi$ 
may have clauses of size less than $k$. 
However, we can still show that two CNF formulas have the same set of satisfying assignments 
because every clause of size $k - i$ can be simulated by a set of clauses of size $k$.

To prove the sample complexity of Valiant's algorithm, the key is to verify the following resilience property for CNF formulas that are well-behaved.

\begin{lemma}\label{lemma:valiant-algorithm-marginal-lower-bound}
  Let $k \ge 10^5$, $\alpha \le \frac{2^{k - 30 k^{4 / 5} \log k}}{2^{10} \e^2 k^8}$ be two constants. For any $n \geq n_0(k,\alpha)$,
  any fixed CNF formula $\Phi$ with $n$ variables and $m = \floor{\alpha n}$ clauses that is 
  $(k, \alpha, \degreeConstant, \badclauseConstant, \frozenEta, \frozenRho, \frozenConstant, \BetaConstant)$-well-behaved, 
  where parameters are set in~\eqref{eq:parameters}, 
  $\Phi$ is $\theta$-resilient with $\theta = \frac{1}{8^k} \cdot n^{-\exp(-\sqrt{k})}$.
\end{lemma}

With \Cref{lemma:valiant-algorithm-marginal-lower-bound}, we can prove \Cref{theorem:learning-random-cnf}.

\learningRandom*

\begin{proof}
  By \Cref{lemma:well-behaved-random-cnf}, with probability $1 - o(1 / n)$ over the random formula $\Phi = \Phi(k, n, m = \floor{\alpha n})$, $\Phi$ is $(k, \alpha, \degreeConstant, \badclauseConstant, \frozenEta, \frozenRho, \frozenConstant, \BetaConstant)$-well-behaved.
  Assume that $\Phi = (V, \+C)$ is a well-behaved random CNF formula. 
  By \Cref{lemma:valiant-algorithm-marginal-lower-bound}, 
  $\Phi$ is $\theta$-resilient with $\theta = \frac{1}{8^k} \cdot n^{\exp(-\sqrt{k})}$.
  For each clause $c \in \+C$ with size less than $k$ (specifically, size $k-1$ or $k-2$), we can add all clauses of size $k$ that contain all literals in $c$ into $\+C$ without changing the set of satisfying assignments, which results in a $k$-CNF formula $\Phi'$ with at most $O_k(mn^2)$ clauses. Note that $\Phi$ and $\Phi'$ are equivalent in the sense that they have the same set of satisfying assignments.
  We can imagine that Valiant's algorithm learns the CNF formula $\Phi'$. 
  Combining the fact that $\Phi'$ is $\theta$-resilient and \Cref{pps:resilience}, the theorem follows.
\end{proof}

Thus, we only need to verify the resilience property for well-behaved CNF formulas.

\subsubsection{Verifying resilience property (Proof of \Cref{lemma:valiant-algorithm-marginal-lower-bound})}

Suppose we want to verify the resilience property of a fixed well-behaved CNF formula $\Phi=(V, \+C)$ 
with respect to a clause $c^*$ with $\vbl(c^*) = \set{v_1^*, \ldots, v_k^*}$ 
and forbidden assignment $\sigma^* = (\sigma_1^*, \ldots, \sigma_k^*) \in \set{\true, \false}^k$.
If there exists a clause $c \in \+C$ such that 
$c$ forbids a partial assignment on $\vbl(c) \subseteq \vbl(c^*)$ 
that is consistent with $\sigma^*$, 
then it follows directly that 
$\Pr[X\sim\mu_{\Phi}]{X_{\vbl(c^*)} = \sigma^*} = 0$.
Therefore, we assume that the fixed clause $c^*$ satisfies the following condition:
\begin{condition}
  \label{condition:clause-no-forbid}
  There does not exist a clause $c \in \+C$ such that $c$ forbids a partial assignment 
  on $\vbl(c) \subseteq \vbl(c^*)$ that is consistent with $\sigma^*$.
\end{condition}

Basically, the proof of the resilience property in \Cref{lemma:valiant-algorithm-marginal-lower-bound} 
follows from the chain rule of conditional probabilities.
For simplicity of notation, 
we use $\sigma^*_{\le i}$ to denote the pinning $\set[1 \leq j \leq i]{X_{v_j^*} = \sigma_j^*}$. 
By the chain rule of conditional probabilities, we have
\begin{equation*}
  \Pr[X\sim\mu_{\Phi}]{X_{\vbl(c^*)} = \sigma^*} 
  = \prod_{i = 1}^k \Pr[X\sim\mu_{\Phi}]{X_{v_i^*} = \sigma_i^* \mid \sigma^*_{\le i - 1}}.
\end{equation*}
If we can establish a lower bound for each conditional probability in the product, 
then we can consequently derive a lower bound for the marginal probability. 
We now state the following lemma.
\begin{lemma}
  \label{lemma:term-lower-bound}
  For any $i \in [k]$, it holds that
  \begin{equation}
    \label{eq:marginal-lower-bound}
    \Pr[X\sim\mu_{\Phi}]{X_{v_i^*} = \sigma_i^* \mid \sigma^*_{\le i - 1}} 
    \ge \frac{1}{8} \cdot \tp{1 - \tp{\frac{1}{2}}^{\frozenConstant k / 2 - 2}}^{\log n}.
  \end{equation}
\end{lemma}
\noindent 
Recall that we set $\frozenConstant = 2 k^{-1/5}$ in \eqref{eq:parameters}.
Therefore,
assuming the correctness of \Cref{lemma:term-lower-bound} immediately implies the following lower bound
for sufficiently large $n$:
\begin{align*}
  &\Pr[X\sim\mu_{\Phi}]{X_{\vbl(c^*)} = \sigma^*} \ge \tp{\frac{1}{8} \cdot \tp{1 - \tp{\frac{1}{2}}^{\frozenConstant k / 2 - 2}}^{\log n}}^k \\
  =& \tp{\frac{1}{8} \cdot \tp{1 - \tp{\frac{1}{2}}^{k^{4 / 5} - 2}}^{\log n}}^k \ge \frac{1}{8^k} \cdot n^{-8k \cdot 2^{-k^{4 / 5}}} \ge \frac{1}{8^k} \cdot n^{-\exp(-\sqrt{k})}.
\end{align*}
This proves \Cref{lemma:valiant-algorithm-marginal-lower-bound}.
In the following, we focus on proving \Cref{lemma:term-lower-bound} for a fixed $i \in [k]$.

\subsubsection{Lower bound of conditional probability (Proof of \Cref{lemma:term-lower-bound})}

Consider a \emph{deterministic} function
\begin{align*}
  \Pinning: \set{(\tau,v) \in \set{\true, \false}^V \times V} 
  \to \set{(\Lambda,\tau_\Lambda) \mid \Lambda \subseteq V}
\end{align*} 
such that given an assignment $\tau$ on $V$ and a variable $v$ in $V$, 
it outputs a subset $\Lambda \subseteq V$ and the pinning $\tau_\Lambda$ on $\Lambda$ projected from $\tau$. 
We call such a $(\Lambda, \tau_\Lambda)$ a \emph{revealing result}.

Based on the deterministic revealing process $\Pinning$, 
we define a \emph{random} process where we first sample a random $Y$ and then apply $\Pinning(Y, v_i^*)$ to 
get a random revealing result $(S, Y_S)$.
\begin{enumerate}
  \item Draw a \emph{random} solution $Y\sim\mu^{\sigma^*_{\le i - 1}}_{\Phi}$, 
  where $Y \in \{\true, \false\}^V$ is a random assignment from $\mu_{\Phi}$ 
  conditioned on the partial assignment on $\set{v_1^*, \ldots, v_{i-1}^*}$ is fixed as $\sigma^*_{\le i - 1}$.
  \item Output the \emph{random} revealing result $(S, Y_S)= \Pinning(Y, v_i^*)$.
\end{enumerate}
\noindent
Let $\+P$ be the collection of all possible revealing results generated by the above random process, i.e.,
\begin{equation}\label{eq:revealing-results}
  \+P 
  \triangleq \set{\, \Pinning(\tau, v_i^*) \mid \tau\in \supp\tp{\mu^{\sigma^*_{\le i - 1}}_{\Phi}}\,}.
\end{equation}

\begin{definition}[Conditional Gibbs revealing process]
  \label{def:conditional-gibbs-revealing-process}
  The function $\Pinning$ is said to be a \emph{conditional Gibbs revealing process} with respect to $ \pi = \mu^{\sigma^*_{\le i - 1}}_{\Phi}$ if it satisfies the following 
  properties.
  \begin{itemize}
    \item With probability $1$, $v_i^* \notin S$, $\{v_1^*, \ldots, v_{i-1}^*\} \subseteq S$ and $\tau_{v_j^*} = \sigma_j^*$ for all $1 \leq j \leq i-1$.
    \item Let $(S,Y_S) = \Pinning(Y, v_i^*)$, where $Y \sim \pi$. It holds that for any $(\Lambda,\tau_\Lambda) \in \+P$, conditional on $(S,Y_S) = (\Lambda, \tau_\Lambda)$, $Y_{V \setminus \Lambda}$ follows the law of $\pi$ conditional on the configuration on $\Lambda$ being fixed as $\tau_\Lambda$. Formally, for any $x \in \set{\true, \false}^{V \setminus \Lambda}$ on $V \setminus \Lambda$,
    \begin{equation}
      \label{eq:conditional-gibbs-property}
     \Pr[Y\sim\pi]{Y_{V \setminus \Lambda} = x \mid \Pinning(Y, v_i^*) 
      = (\Lambda, \tau_{\Lambda})} = \Pr[Y\sim\pi]{Y_{V \setminus \Lambda} = x \mid Y_{\Lambda} = \tau_{\Lambda}} = \mu_{V \setminus \Lambda, \Phi}^{\tau_\Lambda}(x),
    \end{equation}
    where the last equation holds due to the first property.
  \end{itemize} 
\end{definition}

The above property says that for $Y \sim \mu^{\sigma^*_{\le i - 1}}_{\Phi}$, 
suppose the revealing process reveals a pinning on a subset $S = \Lambda$ 
where $\{v_1^*, \ldots, v_{i-1}^*\} \subseteq \Lambda$, 
then the unrevealed random assignment on $V \setminus \Lambda$ follows 
the law of $\mu^{\sigma^*_{\le i - 1}}_{\Phi}$ conditional on the revealed pinning on $\Lambda$, 
in other words, $Y_{V \setminus \Lambda} \sim \mu^{\tau_{\Lambda}}_{V \setminus \Lambda, \Phi}$. 
The property is satisfied by many natural revealing processes. 
For instance, we can reveal the values of variables in $Y$ one by one and always put any revealed variables 
into the subset $\Lambda$. 
The specific construction of $\Pinning$ that we use in the proof will be given later.

We define \emph{nice} revealing results as follows.
The nice revealing results enable us to establish a lower bound for the conditional probability in \Cref{lemma:term-lower-bound}.

\begin{definition}[Nice revealing result]
  \label{def:nice-pinning}
  A revealing result $(\Lambda, \tau_{\Lambda}) \in \+P$ is said to be \emph{nice} if the following conditions are satisfied.

  \begin{itemize}
    \item The variable $v_i^*$ has not been revealed; that is, $v_i^* \notin \Lambda$.
    \item All previously revealed variables are included, i.e., $\{v_1^*, \ldots, v_{i-1}^*\} \subseteq \Lambda$, and for each $1 \le j \le i-1$, it holds that $\tau_{v_j^*} = \sigma_j^*$.
    \item \label{item:nice-pinning-main} Let $\Phi'$ denote the CNF formula obtained from $\Phi$ by simplifying with respect to $\tau_{\Lambda}$,
    that is, by removing all variables in $\Lambda$ and all clauses satisfied by $\tau_{\Lambda}$.
    Then one of the following holds:
    \begin{enumerate}
      \item \label{item:nice-pinning-isolated} $v_i^*$ is an \textbf{isolated variable} in $\Phi'$, meaning that it does not appear in any clause of $\Phi'$; or
      \item $v_i^*$ is contained in some clause of $\Phi'$. 
      In this case, let $C'$ be the maximal connected component of the dependency graph $G_{\Phi'}$ such that $v_i^* \in \vbl(C')$.
      The following conditions must all be satisfied:
      \begin{enumerate}
        \item \label{item:nice-pinning-1} For all clauses $c \in C'$, except at most one clause, 
        it holds that $\abs{\vbl_{\Phi'}(c)} \ge \frozenConstant k - 1$;
        \item \label{item:nice-pinning-2} For the (at most one) exceptional clause $c' \in C'$ with $\abs{\vbl_{\Phi'}(c')} < \frozenConstant k - 1$;
        if $\vbl_{\Phi'}(c') = \{v_i^*\}$, the clause $c'$ is satisfied by $\sigma_i^*$;
        \item \label{item:nice-pinning-3} The size of $C'$ is bounded by $\abs{C'} \le \log n$.
      \end{enumerate}
    \end{enumerate}
  \end{itemize} 
\end{definition}

To prove the marginal lower bound for $v^*_i$, we establish the following two lemmas.
The first lemma says that, conditional on any nice revealing result, 
the conditional probability of $v_i^*$ being assigned to $\sigma_i^*$ is at least a constant.
The second lemma says there exists a revealing process such that with at least constant probability over $Y \sim \mu^{\sigma^*_{\le i - 1}}_{\Phi}$, the revealing result is nice.
Formally, assume that the CNF formula $\Phi$ is well-behaved and satisfies the conditions in \Cref{lemma:valiant-algorithm-marginal-lower-bound} and the fixed clause $c^*$ satisfies \Cref{condition:clause-no-forbid}. Then we have the following two lemmas.
\begin{lemma}
  \label{lemma:random-cnf-marginal-lower-bound-nice}
  For any nice revealing result $(\Lambda, \tau_{\Lambda})$, it holds that
  \begin{equation}
    \label{eq:nice-pinning-lower-bound}
    \Pr[Y\sim\mu^{\sigma^*_{\le i - 1}}_{\Phi}]{Y_{v_i^*} = \sigma_i^* \mid Y_{\Lambda} = \tau_{\Lambda}} 
    \ge \frac{1}{4} \cdot \tp{1 - \tp{\frac{1}{2}}^{\frozenConstant k / 2 - 2}}^{\log n}.
  \end{equation}
\end{lemma}
\begin{lemma}
  \label{lemma:random-cnf-marginal-lower-bound-nice-probability}
  There exists a conditional Gibbs revealing process \emph{$\Pinning$} such that
  \begin{equation}
    \label{eq:nice-pinning-probability}
    \Pr[Y\sim\mu^{\sigma^*_{\le i - 1}}_{\Phi}]{
      \emph{\Pinning}(Y, v_i^*) \textnormal{ is nice}
    } \ge \frac{1}{2}.
  \end{equation}
\end{lemma}
\noindent
Assuming the correctness of these two lemmas, 
we can prove the lower bound \eqref{eq:marginal-lower-bound} 
in \Cref{lemma:term-lower-bound}.

\begin{proof}[Proof of \Cref{lemma:term-lower-bound}]
  Let $\Pinning$ be a conditional Gibbs revealing process in \Cref{lemma:random-cnf-marginal-lower-bound-nice-probability}.
  Let $\+P$ be the collection of all possible revealing results generated by $\Pinning$.
  Let $\+P_{\text{nice}} \subseteq \+P$ be the collection of nice revealing results.
  Let $Y$ be a random solution drawn from $\mu^{\sigma^*_{\le i - 1}}_{\Phi}$.
  By the law of total probability over the randomness of $Y$, we have
  \begin{align*}
    \Pr{Y_{v_i^*} = \sigma_i^* }
    &= \sum_{(\Lambda, \tau_{\Lambda}) \in \+P} 
    \Pr{Y_{v_i^*} = \sigma_i^* \mid \Pinning(Y, v_i^*) = (\Lambda, \tau_{\Lambda})} 
    \cdot \Pr{\Pinning(Y, v_i^*) = (\Lambda, \tau_{\Lambda})} \\
    &\overset{\eqref{eq:conditional-gibbs-property}}{=} \sum_{(\Lambda, \tau_{\Lambda}) \in \+P} 
    \Pr{Y_{v_i^*} = \sigma_i^* \mid Y_{\Lambda} = \tau_{\Lambda}} 
    \cdot \Pr{\Pinning(Y, v_i^*) = (\Lambda, \tau_{\Lambda})} \\
    &\ge \sum_{(\Lambda, \tau_{\Lambda}) \in \+P_{\text{nice}}} 
    \Pr{Y_{v_i^*} = \sigma_i^* \mid Y_{\Lambda} = \tau_{\Lambda}} 
    \cdot \Pr{\Pinning(Y, v_i^*) = (\Lambda, \tau_{\Lambda})} \\
    &\overset{\eqref{eq:nice-pinning-lower-bound}}{\ge} \frac{1}{4} \cdot \tp{1 - \tp{\frac{1}{2}}^{\frozenConstant k / 2 - 2}}^{\log n} \cdot \tp{\sum_{(\Lambda, \tau_{\Lambda}) \in \+P_{\text{nice}}} \Pr{\Pinning(Y, v_i^*) = (\Lambda, \tau_{\Lambda})}} \\
    &\overset{\eqref{eq:nice-pinning-probability}}{\ge} \frac{1}{8} \cdot \tp{1 - \tp{\frac{1}{2}}^{\frozenConstant k / 2 - 2}}^{\log n}.
  \end{align*}
  Since $\Pr[X\sim\mu_{\Phi}]{X_{v_i^*} = \sigma_i^* \mid \sigma^*_{\le i - 1}}=\Pr[Y\sim\mu^{\sigma^*_{\le i - 1}}_{\Phi}]{Y_{v_i^*} = \sigma_i^*}$,
  this proves the lemma.
\end{proof}

Our task is reduced to proving the two lemmas.
We now prove \Cref{lemma:random-cnf-marginal-lower-bound-nice}.
The main idea of our proof is to satisfy clauses in $C'$ one by one via pinning all degree-one variables 
(other than $v^*_i$) in these clauses to their satisfying values. 
Observe that all marginal probabilities involved in this process can be lower bounded by $1 / 2$ 
since the variable only appears in one clause, and we choose the satisfying value.
The reason that we can always find a clause with sufficient degree-one variables 
to proceed is due to the degree-one variable property of well-behaved CNF formulas 
and the definition of nice pinnings. 
Once all clauses in $C'$ are satisfied, 
the variable $v_i^*$ becomes a degree-one variable 
and we can pin it to $\sigma_i^*$ with probability $1 / 2$.
\begin{proof}[Proof of \Cref{lemma:random-cnf-marginal-lower-bound-nice}]
  Recall that $\Phi$ is a 
  $(k, \alpha, \degreeConstant, \badclauseConstant, \frozenEta, 
  \frozenRho, \frozenConstant, \BetaConstant)$-well-behaved CNF formula 
  and the fixed clause $c^*$ satisfies \Cref{condition:clause-no-forbid}.
  
  If $v_i^*$ is an isolated variable in $\Phi'$, then the marginal probability can be directly lower bounded by $1/2$ and the lemma follows.
  Thus, we focus on the case where $v_i^*$ is contained in some clauses of $\Phi'$.
  Let $(\Lambda, \tau_{\Lambda})$ be any nice revealing result and $\Phi'$ be the CNF formula simplified by $\tau_{\Lambda}$. 
  Since $v_i^* \notin \Lambda$, the variable $v_i^*$ must be in the simplified formula $\Phi'$.
  Let $C'$ be the maximal connected component in the dependency graph $G_{\Phi'}$ with $v_i^* \in \vbl_{\Phi'}(C')$.
  Let $c'$ be the (at most one) exceptional clause in $C'$ with $\abs{\vbl_{\Phi'}(c')} < \frozenConstant k - 1$. If no such clause exists, we simply ignore $c'$ in the following analysis.

  We first introduce a procedure for finding a sequence of variables to pin. \Cref{alg:iterative-elimination} takes as input the CNF formula $\Phi'$ and the variable $v_i^*$, and iteratively finds clauses in $C'$ with sufficiently many degree-one variables (other than $v_i^*$) to pin. In each iteration, it selects a clause $c^\circ$ with at least $\frozenConstant k / 2 - 1$ degree-one variables (other than $v_i^*$), collects these degree-one variables into a set $S_t$, and defines an assignment $\tau_{S_t}$ on $S_t$ that satisfies $c^\circ$. The clause $c^\circ$ is then removed from $\Phi'$. This process continues until all clauses in $C'$ are removed. The output of the algorithm is the sequences $(S_1, S_2, \ldots, S_T)$ and $(\tau_1, \tau_2, \ldots, \tau_T)$.
  We use $\deg_{\Phi'}(v)$ to denote the degree (number of clauses containing $v$) of variable $v$ in the formula $\Phi'$. 

  \vspace{0.2cm}

  \begin{algorithm}[H]
    \caption{\textsf{IterativeElimination}$(\Phi', v^*_i)$}
    \label{alg:iterative-elimination}

    \SetKwInOut{Input}{Input}
    \SetKwInOut{Output}{Output}

    \Input{a CNF formula $\Phi'$ with clause set $C'$, a variable $v^*_i$;}
    \Output{the sequences $(S_1, S_2, \ldots, S_T)$ and $(\tau_{S_1}, \tau_{S_2}, \ldots, \tau_{S_T})$;}

    $T \gets 0$;

    \While{$|C'| > 1$}{
        $T \gets T + 1$;
        
        Let $\mathcal{A} = \{\, c \in C': \left|\{\,x \in \vbl(c) : 
        \deg_{\Phi'}(x) = 1 \wedge x \neq v^*_i \,\}\right| \ge \frozenConstant k / 2 - 2\} \setminus \set{c'}$;
        \tcc{$c'$ is the exceptional small clause. Ignore if not exist.}

        Choose $c^\circ \in \mathcal{A}$ with the smallest index;

        Let $S_T \gets \{\, x \in \vbl(c^\circ) : \deg_{\Phi'}(x) = 1 \wedge x \neq v^*_i \,\}$;

        Let $\tau_T$ be the projection of the forbidden assignment of $c^\circ$ on $S_T$;

        Remove $c^\circ$ from $\Phi'$ and from $C'$;
    }

    \KwRet{$(S_1, S_2, \ldots, S_T)$ and $(\tau_1, \tau_2, \ldots, \tau_T)$;}
  \end{algorithm}
  \vspace{0.2cm}

  We first show that in each iteration of \Cref{alg:iterative-elimination}, 
  the set $\mathcal{A}$ is nonempty.
  By the definition of nice revealing results, 
  we have $\abs{C'} \le \log n$ (\Cref{item:nice-pinning-3} of \Cref{def:nice-pinning}).
  Therefore, during each iteration of the ``while'' loop, it holds that $2 \le \abs{C'} \le \log n$.
  Since \Cref{alg:iterative-elimination} never removes any variable from $\Phi'$,
  by \Cref{item:nice-pinning-1} of \Cref{def:nice-pinning},
  it holds that $\abs{\vbl_{\Phi'}(c)} \ge \frozenConstant k - 1$ 
  for any $c \in C' \setminus \set{c'}$ 
  throughout the entire process of \textsf{IterativeElimination}.
  Since $\Phi$ is well-behaved, 
  and letting $C'$ denote the subset of clauses specified in \Cref{property:degree-one-variable}, 
  there must exist a clause other than $c'$ that contains at least 
  $\BetaConstant k = k - k^{4/5}$ degree-one variables,
  and hence at most $k^{4/5}$ variables whose degree is larger than one.
  Consequently, this clause contains at least $\frozenConstant k - k^{4/5} - 1$ degree-one variables 
  with respect to $\Phi'$.
  Recalling that $\frozenConstant = 2k^{-1/5}$ in \eqref{eq:parameters}, 
  it contains at least $\frozenConstant k / 2 - 2$ degree-one variables other than $v_i^*$. 
  This confirms that $\mathcal{A} \neq \emptyset$.

  We next observe that once a variable is included in some $S_t$, 
  it no longer appears in any clause of the remaining formula, 
  since its degree in $\Phi'$ is one and the only clause containing it has already been removed. 
  Therefore, all subsets $S_t$ are disjoint.
  Furthermore, $v_i^*$ does not belong to any $S_t$ for all $t \in [T]$. 
  When only one clause remains in $C'$, the ``while'' loop terminates, 
  and we denote this remaining clause by $c^\sharp$.

  For clarity, we introduce some notation.  
  Let $\pi$ denote the distribution $\mu^{\sigma^*_{\le i-1}}_{\Phi}$.  
  For each $t \in [T]$, let $c_t^\circ$ denote the clause corresponding to the subset $S_t$.  
  Define $\+E_t$ as the event that there exists $w \in S_t$ such that $Y(w) \neq \tau_t(w)$,  
  which means that $Y$ satisfies the clause $c_t^\circ$ through some variables in $S_t$.
  By the chain rule of conditional probabilities,
  we have
  \begin{align*}
    &\Pr[Y\sim\pi]{Y_{v_i^*} = \sigma_i^* \mid Y_{\Lambda} = \tau_{\Lambda}} 
    \geq \Pr[Y\sim\pi]{\tp{Y_{v_i^*} = \sigma_i^*} \land \bigwedge_{t \in [T]} \+E_t \mid Y_{\Lambda} = \tau_{\Lambda}} \\
    =\,& \Pr[Y\sim\pi]{Y_{v_i^*} = \sigma_i^* \mid Y_{\Lambda} = \tau_{\Lambda}, \bigwedge_{t \in [T]} \+E_t}
    \cdot \prod_{t=1}^T \Pr[Y\sim\pi]{\+E_t \mid Y_{\Lambda} = \tau_{\Lambda}, \bigwedge_{j < t} \+E_j}.
  \end{align*}
  Note that since $\{v_1^*, \ldots, v_{i-1}^*\} \subseteq \Lambda$,  
  these variables are removed from $\Phi$. Moreover, for all $1 \le j \le i-1$, it holds that $\tau_{v_j^*} = \sigma_j^*$.  
  Therefore, all variables in the clause set $C'$ are free variables whose values remain unfixed under the distribution $\pi$.

  We lower bound the conditional probability of $\+E_t$'s as follows.
  For any $t \in [T]$, by the definition of the subset $S_t$ and the forbidden assignment $\tau_t$, 
  we claim that
  $$\Pr[Y \sim \pi]{Y_{S_t} = \tau_t \mid Y_{\Lambda} = \tau_{\Lambda}, \bigwedge_{j < t} \+E_j} 
  \le \tp{\frac{1}{2}}^{|S_t|}
  \le \tp{\frac{1}{2}}^{\frozenConstant k / 2 - 2}.$$
  The inequality holds because all variables in $S_t$ are degree-one variables in the formula obtained by 
  simplifying $\Phi$ with $\tau_{\Lambda}$ and removing all previous clauses $c_j^\circ$ for $j < t$.  
  Note that the condition $\bigwedge_{j < t} \+E_j$ ensures that all clauses $c_j^\circ$ for $j < t$ are satisfied.  
  For each degree-one variable $w \in S_t$, 
  the marginal probability that $w$ takes the forbidden value $\tau_t(w)$ is always at most $1/2$ under any conditioning.
  Combining this observation and the chain rule of conditional probabilities
  gives the first inequality.
  The second inequality follows from the fact that $|S_t| \ge \frozenConstant k / 2 - 2$. Hence
  \begin{align*}
    \Pr[Y\sim\pi]{\+E_t \mid Y_{\Lambda} = \tau_{\Lambda}, \bigwedge_{j < t} \+E_j} 
    \geq 1 - \tp{\frac{1}{2}}^{\frozenConstant k / 2 - 2}.
  \end{align*}

  Finally, we lower bound the conditional probability of $Y_{v_i^*} = \sigma_i^*$. 
  Given the condition $\bigwedge_{t \in [T]} \+E_t$, 
  all clauses except the last clause $c^\sharp$ are satisfied. 
  Furthermore, we can remove all variables in $\cup_{t \in [T]} S_t$,
  since they are no longer involved in any clause.
  The remaining formula $\Phi'$ consists of a single clause $c^\sharp$ and possibly several isolated variables.  
  We now analyze the marginal distribution of $v_i^*$ in this reduced formula, 
  considering the following three cases:
  \begin{itemize}
    \item If $v_i^* \notin \vbl_{\Phi'}(c^\sharp)$,  
    then $v_i^*$ is an isolated variable and assigned the value $\sigma_i^*$ with probability $\frac{1}{2}$.
    \item If $v_i^* \in \vbl_{\Phi'}(c^\sharp)$ and  
    $c^\sharp = c'$ is the exceptional clause in the definition of nice revealing results,  
    then there are two subcases:
    \begin{itemize}
      \item If $|\vbl_{\Phi'}(c^\sharp)| = 1$,  
      by \Cref{item:nice-pinning-2} of \Cref{def:nice-pinning}, $c'$ is satisfied by $\sigma_i^*$,  
      and thus $v_i^*$ is assigned the value $\sigma_i^*$ with probability $1$.
      \item If $|\vbl_{\Phi'}(c^\sharp)| > 1$,  
      we can pin another variable $w \neq v_i^*$ in $\vbl_{\Phi'}(c^\sharp)$ to satisfy $c^\sharp$,  
      which occurs with probability at least $\frac{1}{2}$.  
      Condition on this pinning, $v_i^*$ is assigned the value $\sigma_i^*$ with probability $\frac{1}{2}$.
      Thus, the probability that $v_i^*$ is assigned the value $\sigma_i^*$ is at least $\frac{1}{4}$.
    \end{itemize}
    \item If $v_i^* \in \vbl_{\Phi'}(c^\sharp)$ and $c^\sharp \neq c'$,  
    then $c^\sharp$ must contain at least $\frozenConstant k$ variables,  
    since no variable in $\bigcup_{t \in [T]} S_t$ belongs to $c^\sharp$.
    In particular, $|\vbl_{\Phi'}(c^\sharp)| > 1$.
    By the same argument as the previous case, 
    $v_i^*$ is assigned the value $\sigma_i^*$ with probability at least $\frac{1}{4}$.
  \end{itemize}
  Combining the above three cases, we have
  \begin{equation*}
    \Pr[Y\sim\pi]{Y_{v_i^*} = \sigma_i^* \mid Y_{\Lambda} = \tau_{\Lambda}} 
    \geq \frac{1}{4} \cdot \tp{1 - \tp{\frac{1}{2}}^{\frozenConstant k / 2 - 2}}^{\log n}. \qedhere
  \end{equation*}
\end{proof}

Now, the only thing left is to explicitly construct a conditional Gibbs revealing process $\Pinning$ 
that satisfies the desired property in \Cref{lemma:random-cnf-marginal-lower-bound-nice-probability}.

\subsection{Construction of revealing process}

In this subsection, we describe the revealing process $\Pinning$ (\Cref{alg:reveal}), 
which is used to prove \Cref{lemma:random-cnf-marginal-lower-bound-nice-probability}.  
We begin by considering two simple cases:  
(1) $\alpha < 1 / k^3$, or  
(2) there is no clause containing $v_i^*$ that remains unsatisfied under $\sigma^*_{\le i - 1}$.  
In either case, the revealing process simply returns the pinning $\sigma^*_{\le i - 1}$.

We then assume that $\alpha \ge 1 / k^3$ and there exists at least one clause 
containing $v_i^*$ that has not been satisfied by $\sigma^*_{\le i - 1}$ 
in the following analysis. 
Let $c_0$ denote the clause with the smallest index that contains $v_i^*$ 
and remains unsatisfied by $\sigma^*_{\le i - 1}$.  
Note that $c_0$ is fixed and does not depend on the randomness of the random process built upon $\Pinning$.  
We will consistently use $c_0$ to refer to this clause throughout the following analysis.

\paragraph{Modify bad variables}
Let $\Phi = (V, \+C)$ be a well-behaved CNF formula satisfying the condition in \Cref{lemma:well-behaved-random-cnf}.
Recall that $\sbadvars$ and $\sbadclauses$ are the output of 
$\textsf{IdentifyBad}(\Phi, \degreeConstant, \badclauseConstant)$, which is created by first adding high-degree variables and then recursively adding clauses that are significantly affected by high-degree variables.
Here, our goal is to analyze the conditional marginal probability $\oPr_{X \sim \mu_{\Phi}}[X_{v^*_i} = \sigma^*_i \vert \sigma^*_{\le i-1}]$, where the values of $v_j^*$ are fixed for $j \leq i - 1$. We need to take the effect of these variables into account.
Now, we slightly modify $\sbadclauses$ and $\sbadvars$ to obtain 
the final sets of bad variables and clauses that will be used in the analysis.
Define $\intersectionclause \subseteq \+C$ be the set of clauses that contain at least $2k^{4/5}$ variables in $\vbl\tp{c^*}$. 
We also regard these clauses as ``bad clauses'', helping us ensure that each clause contains a sufficient number of unrevealed variables after the revealing process.
We show that there are a small number of such clauses by providing an upper bound of $|\intersectionclause|$ using \Cref{cor:bound-large-intersection}.
Recall that any two clauses share at most $3$ variables, since $\Phi$ is well-behaved (by \Cref{property:clause-intersection}).  
We set the parameters in \Cref{cor:bound-large-intersection} as $q = 3$ and $p = (2k^{4/5} - \sqrt{4k^{8/5} - 6k})/3$, which is the smaller root of the equation $\frac{k}{p} + \frac{3p}{2} = 2k^{4/5}$.
By \Cref{cor:bound-large-intersection} and the fact that $p \ge 1$, $2k^{4 / 5} \le k$, let $\tilde{\+C} = \intersectionclause = \set{c \in \+C : \abs{\vbl(c) \cap \vbl(c^*)} \ge 2k^{4/5}}$ and we have the following bound on the size of $\intersectionclause$:
\begin{equation}\label{eq:bounded-intersection-clause}
  \abs{\intersectionclause} 
  \le \frac{2k^{4/5} - \sqrt{4k^{8/5} - 6k}}{3}
  \le k^{4/5} - 2,
\end{equation}
where the last inequality holds when $k \ge 5$.
Recall that $c_0$ denotes the clause with the smallest index that contains $v_i^*$ 
and remains unsatisfied under $\sigma^*_{\le i - 1}$.
Define
\begin{align}
  \label{eq:modify-bad}
  \begin{aligned}
    &\badvars \triangleq \sbadvars \cup \vbl(\intersectionclause) \cup \set{v^*_1, \ldots, v^*_{i - 1}}\cup \vbl(c_0),\quad
    &&\badclauses \triangleq \sbadclauses \cup \intersectionclause \cup \set{c_0},\\
    &\goodvars \triangleq V \setminus \badvars,\quad
    &&\goodclauses \triangleq \+C \setminus \badclauses,
  \end{aligned}
\end{align}
where $\vbl(\intersectionclause) = \bigcup_{c \in \intersectionclause} \vbl(c)$.
Compared to the original $\sbadvars$ and $\sbadclauses$ defined in \Cref{alg:identify-bad}, we further add $\intersectionclause$ and $c_0$ as the bad sets. Their variables are all treated as bad variables. Finally, all fixed variables $v^*_1, \ldots, v^*_{i - 1}$ are also treated as bad variables.

We introduce two notations to distinguish the good and bad variables appearing in a clause.  
For any clause $c \in \+C$, define
\begin{equation*}
    \vblg(c) \triangleq \vbl(c) \cap \goodvars 
    \quad \text{and} \quad
    \vblb(c) \triangleq \vbl(c) \cap \badvars.
\end{equation*}
We then have the following observation.
\begin{observation}\label{obs:goodclauses-bounds}
    For any $c \in \goodclauses$, it holds that
    $$\abs{\vblb(c)} \le \badclauseConstant k + 5k^{4 / 5} - 3, 
    \quad \abs{\vblg(c)} \ge (1 - \badclauseConstant) k - 5k^{4 / 5}.$$
    We denote $\goodlowerbound \triangleq (1 - \badclauseConstant) k - 5k^{4 / 5}$ as the lower bound of the number of good variables in any good clause.
  \end{observation}
  \begin{proof}
    To verify the upper bound of $\abs{\vblb(c)}$, 
    note that $c$ contains at most $\badclauseConstant k$ variables from $\sbadvars$; 
    there are at most $k^{4/5} - 1$ clauses in $\intersectionclause \cup \set{c_0}$ 
    and each of them shares at most $3$ variables with $c$; 
    and $c \notin \intersectionclause$ contains at most $2k^{4/5}$ variables from 
    $\set{v^*_1, \ldots, v^*_{i-1}}$. The upper bound is
    \begin{equation*}
      \badclauseConstant k + 3(k^{4/5} - 1) + 2k^{4/5} = \badclauseConstant k + 5k^{4/5} - 3.
    \end{equation*}
    The lower bound of $\abs{\vblg(c)}$ can be verified using $\goodvars = V \setminus \badvars$.
\end{proof}

\paragraph{Associated component}

The revealing results $(\Lambda,\tau_\Lambda)$ can be viewed as a partial pinning on $\Lambda$. 
To define the process $\Pinning$, we need to classify different types of clauses given a partial pinning.

Let $\sigma$ be an arbitrary partial pinning. 
We use $\Gamma (\sigma)$ to denote the set of variables that $\sigma$ is \emph{not} defined on. 
In other words, $\sigma \in \{\true, \false\}^{V \setminus \Gamma(\sigma)}$.
We say $c(\sigma) = \true$ iff clause $c$ is satisfied by the pinning $\sigma$.
Given a pinning $\sigma$, we are mainly interested in the unpinned variables in $\Gamma(\sigma)$.
For a clause $c \in \+C$,
let $\vbls(c) \triangleq \vbl(c) \cap \Gamma(\sigma), \vblsg(c) \triangleq \vbls(c) \cap \goodvars, \vblsb(c) \triangleq \vbls(c) \cap \badvars$ 
be the set of unpinned variables, good variables, and bad variables in $c$ under $\sigma$, respectively.
Define
$$N^{\sigma}(c) \triangleq \set{c^\prime \in \+C \mid c \neq c'  \land  \vbls(c^\prime) \cap \vbls(c) \neq \emptyset}$$
be the set of $c$'s neighbors through unpinned variables under $\sigma$.
For a subset of clauses $\+C' \subseteq \+C$, 
\begin{align*}
  N^{\sigma}(\+C') \defeq \set{c' \in \+C \setminus \+C' \mid \exists c \in \+C', \vbls(c) \cap \vbls(c') \neq \emptyset}.
\end{align*}
By definition, two clauses are viewed as connected if they share unpinned variables.

Next, we classify the clauses under the pinning $\sigma$.
For any clause $c$, we write $c(\sigma) = \true$ iff $c$ is satisfied by the pinning $\sigma$.
We are mainly interested in clauses with $c(\sigma) \ne \true$ because all satisfied clauses can be viewed as \emph{removed} under the pinning $\sigma$.
We first define the \emph{frozen} and \emph{blocked} clauses.
Intuitively, a clause is frozen if
it is a good clause but currently has only a small number of unpinned good variables.  
A clause is blocked means that 
although it has many unpinned good variables, 
all of them are ``frozen'' by some frozen clauses.  
Hence, this clause is said to be blocked by the frozen clauses.
The formal definitions are as follows.

\begin{definition}[Frozen and blocked clauses]
  \label{definition:frozen-and-blocked-clauses}
  For the parameter $\frozenConstant \in (0,1)$ in \eqref{eq:parameters} and a pinning $\sigma$, 
  we say a clause $c \in \goodclauses$ is \emph{frozen} if it satisfies that
  $c(\sigma) \ne \true$ and $|\vblsg(c)| \le \frozenConstant k$.
  Formally, let
  \begin{equation*}
    \frozenclauses \triangleq \set{c \in \goodclauses \mid \tp{c(\sigma) \ne \true} \land \tp{\abs{\vblsg(c)} \le \frozenConstant k}}.
  \end{equation*}
  A clause $c \in \goodclauses\setminus \frozenclauses$ is \emph{blocked} if it satisfies that
  $c(\sigma) \ne \true$ and
  for every $v \in \vblsg(c)$, there exists $c' \in \frozenclauses$ such that $v \in \vbl(c')$.
  Formally, let
  \begin{equation*}
    \blockedclause \triangleq 
    \set{c \in \goodclauses \setminus \frozenclauses \mid \tp{c(\sigma) \ne \true} \land
     \tp{\forall v \in \vblsg(c), \exists c' \in \frozenclauses \text{ s.t. } v \in \vbl(c')}}.
  \end{equation*}
\end{definition}
\noindent

The following quick observation follows from the definition.
\begin{observation}
  The three sets $\blockedclause, \frozenclauses$, and $\badclauses$ are pairwise disjoint.
\end{observation}

Now, for a bad clause $c$,
we use the following procedure to construct a connected component $\badInterior(c)$
of clauses that consist of all frozen, blocked, and bad clauses  
that are connected to $c$ through unpinned variables in $\Gamma(\sigma)$.
Moreover, let $\currentComponent(c)\triangleq \badInterior(c) \cup N^{\sigma}(\badInterior(c))$ 
be the set of clauses that
contains $\badInterior(c)$ together with all clauses that are one-step neighbors of $\badInterior(c)$  
through unpinned variables in $\Gamma(\sigma)$.

\begin{definition}[Associated component and its exterior]
    Given a pinning $\sigma$ and a bad clause $c$,
    its associated component $\badInterior(c)$ is constructed iteratively as follows:
    \begin{enumerate}
        \item Initialize $\badInterior(c) = \{c\}$.
        \item If a clauses $c \in \frozenclauses\cup \blockedclause\cup \badclauses$ satisfying $c \in N^{\sigma}(\badInterior(c))$, add $c$ into $\badInterior(c)$.
        \item Repeat this process until there is no such $c$.
    \end{enumerate}
    Let $\currentComponent(c)\triangleq \badInterior(c)\cup N^{\sigma}(\badInterior(c))$.
\end{definition}

Finally, we define a set of \emph{alive} variables.
Intuitively, a variable $v$ is alive means that after pinning $v$, 
each unsatisfied good clause still contains many unpinned good variables in $\goodvars$. 

\begin{definition}[Alive variables]
  \label{definition:alive-variables}
  For a pinning $\sigma$, we say a variable $v \in \Gamma(\sigma)$ is \emph{alive} if it satisfies that
  $v \in \goodvars$ and
  for every clause $c \in \goodclauses$ with $v \in \vbl(c)$,
  either $c(\sigma) = \true$, or
  $|\vblsg(c) \setminus \set{v}| > \frozenConstant k - 1$.
  Denote the set of alive variables by $\aliveVars$.
\end{definition}

We are now ready to present our specific revealing process $\Pinning$. Recall that $\Pinning$ is a deterministic process such that given any full assignment $\tau \in \set{\true, \false}^V$ on $V$, it outputs a subset $S \subseteq V$ and the partial assignment $\tau_S$ on $S$. In the following algorithm, we further assume $\tau$ is consistent with $\sigma^*_{\le i - 1}$, i.e., $\tau_{v_j^*} = \sigma_j^*$ for all $1 \leq j \leq i-1$.
The process is given in \Cref{alg:reveal}.

\vspace{0.2cm}

\begin{algorithm}[ht]
  \caption{\Pinning$(\tau, v_i^*)$}
  \label{alg:reveal}
  \SetKwInOut{Input}{Input}
  \SetKwInOut{Output}{Output}
  \Input{an assignment $\tau\in\set{\true, \false}^V$ consistent with $\sigma^*_{\le i - 1}$, a variable $v_i^*\in V$;}
  \Output{a set $S$ of variables, the partial assignment $\tau_S$ on $S$;}

  Initialize $S = \set{v_1^*,\ldots, v_{i-1}^*}$;

  \If{$\alpha < 1 / k^3$ \textnormal{ or there is no clause containing} $v_i^*$ \textnormal{that has not been satisfied by} $\tau_{S}$}{
      \KwRet{$\tp{S, \tau_S}$;}
  }

  Let $c_0$ be the minimum-index clause containing $v_i^*$ that has not been satisfied by $\tau_{S}$;

  Define bad (good) variables $\badvars (\goodvars)$ and clauses $\badclauses (\goodclauses)$ with $c_0$ as in \eqref{eq:modify-bad};

  Let $v = \nextVar$, which is defined as
  \begin{equation*}
      \nextVar \triangleq \begin{cases}
        v \in V_{\text{alive}}^{\tau_S} \cap \vbl\tp{\mathcal{C}_{\mathrm{ext}}^{\tau_S}(c_0)}
        & \text{if } V_{\text{alive}}^{\tau_S}\cap \vbl\tp{\mathcal{C}_{\mathrm{ext}}^{\tau_S}(c_0)}\neq \emptyset,\\
        \bot &\text{otherwise};
      \end{cases}
  \end{equation*}
  \tcc{Pick the vertex with the smallest index to break the tie.}

  \While{$v \neq \bot$}{
      $S \gets S \cup \set{v}$;

      $v\gets \nextVar$;
  }
  \KwRet{$\tp{S, \tau_S}$;}
\end{algorithm}

\vspace{0.2cm}
We first prove that \Cref{alg:reveal} is indeed a conditional Gibbs revealing process.

\begin{lemma}
  \label{lemma:random-cnf-reveal-is-gibbs}
  The revealing process \emph{$\Pinning$} in \Cref{alg:reveal} is a conditional Gibbs revealing process with respect to $ \pi = \mu^{\sigma^*_{\le i - 1}}_{\Phi}$.
\end{lemma}
\begin{proof}

Let $Y \sim \pi$. Let $(S, Y_S) = \Pinning(Y, v_i^*)$. 
We need to show that, conditional on $(S, Y_S) = (\Lambda, \tau_{\Lambda})$, 
$Y_{V \setminus \Lambda}$ follows the law of $\pi$ 
conditioned on the assignment of $\Lambda$ being fixed as $\tau_{\Lambda}$. 
Let $(\Lambda, \tau_{\Lambda})$ be a possible output of the algorithm.
To this end,
we only need to show that $\Pinning(Y, v_i^*)$ outputs $(\Lambda, \tau_{\Lambda})$ if and only if $Y_{\Lambda} = \tau_{\Lambda}$. 

Note that once \Cref{alg:reveal} needs to reveal the value of $\tau_w$ for some vertex $w$, 
it must hold that $w \in S$. 
Therefore, if $Y_{\Lambda} = \tau_{\Lambda}$, 
then although $Y_{V \setminus \Lambda}$ remains random, 
the entire execution of the algorithm becomes deterministic 
and outputs $(\Lambda, \tau_{\Lambda})$.
Conversely, if the algorithm outputs $(\Lambda, \tau_{\Lambda})$, 
it is straightforward to verify that $Y_{\Lambda} = \tau_{\Lambda}$.
\end{proof}

Recall that our goal is to construct a specific revealing process  
$\Pinning$ such that the revealing result is nice with high probability  
(\Cref{lemma:random-cnf-marginal-lower-bound-nice-probability}).  
We can now prove this lemma for the easy case where  
$\alpha < 1 / k^3$ or there is no clause containing $v_i^*$ that has not been satisfied by $\sigma^*_{\le i - 1}$.

\begin{proof}[Proof of \Cref{lemma:random-cnf-marginal-lower-bound-nice-probability} for easy case]
  \label{proof:random-cnf-marginal-lower-bound-nice-probability-easy-case}
  Note that if $\alpha < 1 / k^3$ or there is no clause containing $v_i^*$,
  the returned revealing result $(\Lambda, \tau_{\Lambda})$ in the above procedure 
  is simply $(\{v_1^*, \ldots, v_{i-1}^*\},\sigma^*_{\le i - 1})$. 
  We show that this revealing result is nice. 
  For both cases, 
  the conditions that $v_i^* \notin \Lambda$, $\{v_1^*, \ldots, v_{i-1}^*\} \subseteq \Lambda$, 
  $\tau_{v_j^*} = \sigma_j^*$ for all $1 \leq j \leq i-1$ hold directly. 
  It suffices to verify \Cref{item:nice-pinning-main} in \Cref{def:nice-pinning}. 
  Let $\Phi' = (V', \+C')$ be the CNF formula simplified by $\tau_{\Lambda}$, i.e., removing all clauses satisfied by $\tau_{\Lambda}$ and removing all variables in $\Lambda$ from the remaining clauses.

  On the one hand,  
  if $v^*_i$ is not contained in any clause that has not been satisfied by $\sigma^*_{\le i - 1}$, 
  then $v_i^*$ is an isolated variable in $\Phi'$.  
  Hence, the returned revealing result is nice.  

  On the other hand, if $\alpha < 1 / k^3$,  
  we assume that there exists a clause $c' \in \+C'$ containing $v_i^*$  
  that is not satisfied by $\sigma^*_{\le i - 1}$  
  (otherwise, $v_i^*$ would also be an isolated variable in $\Phi'$, 
  and the returned revealing result would again be nice).
  Since $\Phi$ is well-behaved, any two clauses share at most three variables by \Cref{property:clause-intersection}.  
  This implies that there is at most one clause in $\+C$ containing more than $\frac{2}{3}k$ variables from the set $\set{v_1^*, \ldots, v_{i-1}^*}$.  
  Consequently, \Cref{item:nice-pinning-1} in \Cref{def:nice-pinning} holds,  
  as after simplification, every other clause in $\Phi'$ contains 
  at least $k - \frac{2}{3}k - 2 = \frac{k}{3} - 2 > \frozenConstant k - 1$ variables.

  For the (at most one) exceptional clause $c' \in \+C'$, if $\vbl(c') = \set{v_i^*}$,  
  we claim that $c'$ is satisfied by $\sigma_i^*$.  
  Suppose otherwise. Since $\vbl(c') \subseteq \set{v_1^*, \ldots, v_i^*}$  
  and $c'$ is not removed during the simplification process,  
  it follows that $c'$ is not satisfied by $\sigma^*_{\le i-1}$.  
  Moreover, under the assumption that $c'$ is not satisfied by $\sigma_i^*$,  
  the assignment $\sigma^*_{\le i}$ fixes all variables in $\vbl(c')$ but still fails to satisfy $c'$.  
  This contradicts the assumption on $\sigma^*$ in \Cref{condition:clause-no-forbid}.  
  Hence, \Cref{item:nice-pinning-2} in \Cref{def:nice-pinning} holds.  
  Finally, by \Cref{property:bounded-growth-rate} and the fact that $\alpha < 1/k^3$,  
  no connected component in $G_{\Phi'}$ has size larger than $\log n$.  
  Therefore, \Cref{item:nice-pinning-3} in \Cref{def:nice-pinning} also holds.
\end{proof}

In the following,  
we assume that $\alpha \ge 1 / k^3$ and that  
there exists a clause containing $v_i^*$ that is not satisfied by $\sigma^*_{\le i - 1}$.  
Let $c_0$ denote the clause with the smallest index among such clauses.  
Moreover, since $\Phi$ is well-behaved,  
\Cref{property:bounded-bad-clauses} holds with appropriate parameters.

\subsubsection{Lower bound of good variables}\label{sec:lower-bound-good-variables}

We now state an observation about the invariant property of $\Pinning$,  
namely, that it preserves the number of unpinned good variables in every good clause.
This property is used to verify the \Cref{item:nice-pinning-1} of \Cref{def:nice-pinning}.
First, we have the following observation.
\begin{observation}\label{obs:initial-goodclauses-bound}
  $|\vblsg(c)| > \frozenConstant k$
  holds for any $c\in \goodclauses$ 
  under the initial pinning $\sigma = \sigma^*_{\le i-1}$.
\end{observation}

\begin{proof}
  For any $c \in \goodclauses$, by \Cref{obs:goodclauses-bounds}, we have  
  $\abs{\vblg(c)} \ge (1 - \badclauseConstant)k - 5k^{4/5}$.  
  Furthermore, since $c \notin \badclauses$, in particular $c \notin \intersectionclause$, we have  
  \begin{equation*}
    \abs{\vbl(c) \cap \set{v_1^*, \ldots, v_{i-1}^*}}
    \le \abs{\vbl(c) \cap \vbl(c^*)}
    \le 2k^{4/5}.
  \end{equation*}
  Therefore,
  \begin{equation*}
    \abs{\vblsg(c)}
    \ge \abs{\vblg(c)} - \abs{\vbl(c) \cap \set{v_1^*, \ldots, v_{i-1}^*}}
    \ge (1 - \badclauseConstant)k - 5k^{4/5} - 2k^{4/5}
    \overset{\eqref{eq:parameters}}{=} k - 8k^{4/5},
  \end{equation*}
  where the last equality follows from the parameter setting $\badclauseConstant = k^{-1/5}$ in \eqref{eq:parameters}.
  Meanwhile, since $\frozenConstant = 2k^{-1/5}$ in \eqref{eq:parameters}, we also have $\frozenConstant k = 2k^{4/5}$.
  Therefore, when $k \ge 10^5$ (as assumed in \Cref{theorem:learning-random-cnf}),  
  it follows that $|\vblsg(c)| > \frozenConstant k$.
\end{proof}

The procedure $\Pinning(\tau,v_i^*)$ maintains a pinning $\tau_S$ on a subset $S$. 
For simplicity, we denote the pinning $\tau_S$ as $\sigma$. 
According the procedure, the initial set $S = \set{v_1^*, \ldots, v_{i-1}^*}$ 
and the initial pinning $\sigma =\tau_S = \sigma^*_{\le i-1}$. 
Then, the procedure expands the set $S$ 
by adding one variable at a time and 
the pinning $\sigma$ maintained by the procedure is updated to $\tau_S$ on new $S$ accordingly.

\begin{observation}\label{obs:random-ksat-frozen-clause}
    $|\vblsg(c)| > \frozenConstant k - 1$
    always holds for any $c\in \goodclauses$ during the whole procedure $\emph{\Pinning}(\tau,v_i^*)$, 
    where $\sigma = \tau_S$ is the pinning maintained by the procedure.
\end{observation}

\begin{proof}
    We prove this observation by induction.  
    Initially, the observation holds directly by \Cref{obs:initial-goodclauses-bound}.
    For the induction step, assume that after revealing $t$ variables,
    $|\vblsg(c)| > \frozenConstant k - 1$ holds for any $c\in \goodclauses$.
    We now reveal the $(t + 1)$-th variable $v$,  
    and denote the updated pinning by $\sigma'$.
    On the one hand,
    for any $c\in \goodclauses$ with $v\in \vbl(c)$, 
    \begin{equation*}
      \abs{\vbl^{\sigma'}_{\-g}(c)} 
      = \abs{\vblsg(c) \setminus \set{v}}  
      > \frozenConstant k - 1,
    \end{equation*}
    where the last inequality follows from the definition of $\nextVar$ and $\aliveVars$.
    On the other hand,
    for any $c\in \goodclauses$ with $v\notin \vbl(c)$, 
    the update of the pinning does not affect the clause,  
    and thus the condition continues to hold.
\end{proof}

\subsubsection{Conditional independence}\label{sec:conditional-independence}

In the following, let $\sigma = \tau_S$ denote the output of  
$\Pinning(\tau,v_i^*)$.
We establish the following property, which states that conditioned on the pinning $\sigma$,  
the marginal distribution of $v_i^*$ depends only on the sub-CNF formula induced by the clauses in $\badInterior(c_0)$. Recall that notations: for any $\+C' \subseteq \+C$,
\begin{align*}
    N^{\sigma}(\+C') &\defeq \set{c' \in \+C \setminus \+C' \mid \exists c \in \+C', \vbls(c) \cap \vbls(c') \neq \emptyset}, \\
    N^{}(\+C') &\defeq \set{c' \in \+C \setminus \+C' \mid \exists c \in \+C', \vbl(c) \cap \vbl(c') \neq \emptyset}.
\end{align*}

\begin{lemma}\label{lem:random-kcnf-neighbor-cut}
  For any $c\in N(\badInterior(c_0))$,
  either $c(\sigma) = \true$, 
  or $c\notin N^{\sigma}(\badInterior(c_0))$.
\end{lemma}

\begin{proof}
  If $c(\sigma) = \true$, then the lemma follows immediately.  
  Hence, we assume that $c(\sigma) \ne \true$.  
  Moreover, if $c \in \frozenclauses \cup \blockedclause \cup \badclauses$,  
  then, since $c \notin \badInterior(c_0)$,  
  the construction of $\badInterior(c_0)$ ensures that  
  $c\notin N^{\sigma}(\badInterior(c_0))$,
  which also proves the lemma.
  In the following, 
  we further assume that $c \notin \frozenclauses \cup \blockedclause \cup \badclauses$.
  
  Suppose, for the sake of contradiction, that  
  $c \in N^{\sigma}(\badInterior(c_0))$ (and thus $c \in \currentComponent(c_0)$).
  Since $c(\sigma) = \false$ and  
  $c \notin \frozenclauses \cup \blockedclause \cup \badclauses$,  
  the definition of $\blockedclause$ implies that there exists a variable
  $v\in \vblsg(c)$ 
  such that $v\notin \vbl(\frozenclauses)$.
  Because $v \in \vbl(c) \subseteq \vbl(\currentComponent(c_0))$,
  if $v\in \aliveVars$, then $v\in \aliveVars\cap \vbl(\currentComponent(c_0))$,
  contradicting the termination condition $\aliveVars\cap \vbl(\currentComponent(c_0)) = \emptyset$. 
    
  We now show that $v \in \aliveVars$ indeed holds.  
  By the definition of $\aliveVars$ and the fact that $v \in \Gamma(\sigma) \setminus \badvars$,  
  it suffices to verify that for every good clause $c' \in \goodclauses$ containing $v$,  
  either ${c'(\sigma) = \true}$ or  
  $|\vblsg(c') \setminus \set{v}| > \frozenConstant k - 1$.
  We argue this by contradiction.  
  Suppose there exists a good clause $c' \in \goodclauses$ containing $v$ such that  
  $c'(\sigma) \neq \true$ and  
  $|\vblsg(c') \setminus \set{v}| \le \frozenConstant k - 1$.  
  By \Cref{obs:random-ksat-frozen-clause} and the fact that $c' \in \goodclauses$,  
  we have  
  $|\vblsg(c')| > \frozenConstant k - 1$.  
  This implies that  
  $v \in \vblsg(c')$  
  and hence  
  $|\vblsg(c')| \le \frozenConstant k$.  
  Therefore, $c' \in \frozenclauses$,  
  which contradicts with $v \notin \vbl(\frozenclauses)$.  
  This completes the proof.
\end{proof}

Intuitively, we explain why this lemma implies the conditional independence under the pinning $\sigma$. 
Note that $v_i^* \in c_0$, and hence $v_i^* \in \vbl(\badInterior(c_0))$.  
For any clause $c \in N(\badInterior(c_0))$,  
the lemma ensures that one of the following two conditions must hold:
(i) $c(\sigma) = \true$, which means that conditional on $\sigma$,  
the clause $c$ is already satisfied and can therefore be removed;
(ii) $c \notin N^{\sigma}(\badInterior(c_0))$,  
which means that all remaining un-revealed variables $\vbls(c)$ are outside $\vbl(\badInterior(c_0))$. Hence, these clauses are disconnected from $v_i^*$ after pinning $\sigma$.

\begin{remark}
  \label{remark:random-cnf-associated-component-size}
  As a remark, the above can be rephrased as follows. 
  Consider the CNF formula $\Phi'$ simplified by the pinning $\sigma$ 
  (i.e., remove variables in $\vbl(\sigma)$ and all clauses that are satisfied by $\sigma$). 
  Recall that $c_0$ is the smallest-index clause containing $v_i^*$ that has not been satisfied by $\sigma$. 
  Observe that all other variables in $\vbl(c_0) \setminus \set{v_1^*, \ldots, v_{i-1}^*}$ are bad and are not pinned by $\sigma$ during the revealing process. 
  Therefore, $v_i^*$ is contained in some clause in $\Phi'$. Let $C'$ be the maximal connected component in the dependency graph $G_{\Phi'}$ such that $v_i^* \in \vbl(C')$. On the other hand, consider a new simplified CNF formula $\Phi''$ that only contains clauses in $\badInterior(c_0)$ and variables in $\vbl(\badInterior(c_0))$, where all variables in $\vbl(\sigma)$ are removed and all clauses that are satisfied by $\sigma$ are also removed. Let $C''$ be the maximal connected component in the dependency graph $G_{\Phi''}$ such that $v_i^* \in \vbl(C'')$. By \Cref{lem:random-kcnf-neighbor-cut}, it holds that $C' = C''$ and thus $\abs{C'} \le \abs{\badInterior(c_0)}$.
\end{remark}
\color{black}

\subsubsection{Size of associated component}\label{sec:size-of-associated-component}

Recall that $c_0$ is the minimum-index clause containing $v_i^*$ that has not been satisfied by $\sigma^*_{\le i - 1}$.
Let $(S,Y_S)$ be the output of $\Pinning(Y,v_i^*)$, where $Y \sim \mu^{\sigma^*_{\le i - 1}}_{\Phi}$.
In this subsection, we show that with moderate probability,  
the size of $\badInterior[Y_s](c_0)$ is small by establishing a tail bound.

\begin{restatable}{lemma}{TailBound}
  \label{lemma:random-cnf-bad-interior-tail-bound}
  Assume that the conditions in \Cref{lemma:valiant-algorithm-marginal-lower-bound} are satisfied for the CNF formula $\Phi=(V,\+C)$. Let $(S,Y_S)$ be the output of $\emph{$\Pinning$}(Y,v_i^*)$, where $Y \sim \mu^{\sigma^*_{\le i - 1}}_{\Phi}$. We have the following upper bound on the probability that the size of $\badInterior[Y_S](c_0)$ is at least $\log n$:
  $$\sum_{\ell = \ceil{\log n}}^{\ceil{\frozenRho\cdot \alpha n}} \alpha n \cdot n^3 \tp{\mathrm{e}k^2\alpha}^\ell \cdot \tp{20 k \cdot 2^{10 k^{4 / 5} \log k}}^{\ell} \cdot \tp{\frac{1}{2}\exp\tp{\frac{1}{k}}}^{(1 - \frozenEta)\cdot \revealedLowerbound\cdot \varrho \cdot \ell}.$$
  where
  $$\revealedLowerbound = (1 - \badclauseConstant - \frozenConstant) k - 5 k^{4 / 5}, \quad \varrho = \frac{1 - \frac{24k^5}{(1 - \frozenEta)(\badclauseConstant - \frozenEta)\degreeConstant}}{1 + \frac{\frozenEta + 2/k - 2\frozenEta/k}{\frozenConstant - \frozenEta - 2/k + 2\frozenEta/k}}.$$
\end{restatable}

The rest of the proof is organized as follows. We first prove \Cref{lemma:random-cnf-bad-interior-tail-bound} by the standard witness argument. 
Then, in \Cref{sec:proof-random-cnf-marginal-lower-bound-nice-probability}, we use the tail bound in \Cref{lemma:random-cnf-bad-interior-tail-bound} to prove \Cref{lemma:random-cnf-marginal-lower-bound-nice-probability}.

We first prove this tail bound in \Cref{lemma:random-cnf-bad-interior-tail-bound} by the standard witness argument. To apply the standard properties of random CNF formulas, which are only applicable for not so large clause sets, we include the pruning method that originates from~\cite[Lemma 7.8]{heImprovedBoundsSampling2023}. 
{Recall the definition of $G_{\Phi}$. The vertex set of $G_{\Phi}$ is $\+C$ and two vertex $c_1, c_2$ are adjacent iff $c_1\neq c_2$ and $\vbl(c_1)\cap \vbl(c_2)\neq \emptyset$.}
Given a set of clauses $K \subseteq \+C$, which is a set of vertices in $G_{\Phi}$, we use $G[K]$ to denote the induced subgraph of $G_{\Phi}$ on $K$.

\begin{lemma}
  \label{lemma:random-cnf-pruning-lemma}
  For any $(S, \sigma)$ generated by \emph{$\Pinning$},
  there exists $\prunedBadInterior(c_0) \subseteq \badInterior(c_0)$ such that
  \begin{enumerate}
    \item $G_{\Phi}[{\prunedBadInterior(c_0)}]$ is a connected subgraph of $G_{\Phi}$.
    \item If $\abs{\badInterior(c_0)}\le \frozenRho\cdot \alpha n$, then $\prunedBadInterior(c_0) = \badInterior(c_0)$. Otherwise we have $\frac{\frozenRho\cdot \alpha n}{k}\le |\prunedBadInterior(c_0)|\le \frozenRho\cdot \alpha n$.
    \item For any clause $c\in \prunedBadInterior(c_0)\cap \blockedclause$ and $v\in \tp{\vbl(c)\cap \Gamma(\sigma)} \setminus \badvars$, there exists some $c^\sharp\in \prunedBadInterior(c_0) \cap \frozenclauses$ such that $v\in \vbl(c^\sharp)$.
  \end{enumerate}
\end{lemma}
\begin{proof}
  We introduce the following pruning process to construct $\prunedBadInterior(c_0)$ from $\badInterior(c_0)$.
  Initialize $\prunedBadInterior(c_0) \gets \badInterior(c_0)$, we prune $\prunedBadInterior(c_0)$ by the following process until $|\prunedBadInterior(c_0)| \le \frozenRho\cdot \alpha n$.
  \begin{itemize}
    \item If there exists $c\in \prunedBadInterior(c_0)\cap \blockedclause$, then let $\+S_1, \+S_2,\dots, \+S_t$ be the maximal connected components of $\prunedBadInterior(c_0)$ in $G_{\Phi}$ after removing $c$, i.e., $\+S_i$'s are maximal connected components in $G_{\Phi}[{\prunedBadInterior(c_0) \setminus \set{c}}]$. Assume that $\+S_1$ has the maximal size. We update $\prunedBadInterior(c_0) \gets \+S_1$.
    \item Otherwise, $\prunedBadInterior(c_0) \cap \blockedclause = \emptyset$. Then let $c\in \prunedBadInterior(c_0)$ be an arbitrary clause such that removing $c$ does not disconnect $\prunedBadInterior(c_0)$ in $G_\Phi$, i.e., $G_{\Phi}[{\prunedBadInterior(c_0)\setminus \set{c}}]$ is a connected component. We update $\prunedBadInterior(c_0)\gets \prunedBadInterior(c_0)\setminus \set{c}$.
  \end{itemize}

  We begin to verify the properties of $\prunedBadInterior(c_0)$.

  The first item holds directly by the construction of $\prunedBadInterior(c_0)$.

  For the second item, if $\abs{\badInterior(c_0)}\le \frozenRho\cdot \alpha n$, the first item holds trivially.
  So we assume that $|\badInterior(c_0)| > \frozenRho\cdot \alpha n$.
  We first show that if $|\prunedBadInterior(c_0)|> \frozenRho\cdot \alpha n$, then after one-step pruning, we have $|\prunedBadInterior(c_0)| \ge \frac{\frozenRho\cdot \alpha n}{k}$. For the case that $\prunedBadInterior(c_0) \cap \blockedclause = \emptyset$, it holds that $|\prunedBadInterior(c_0)| \ge \frozenRho\cdot \alpha n - 1 \ge \frac{\frozenRho\cdot \alpha n}{k}$. Next, we consider the case that $\prunedBadInterior(c_0) \cap \blockedclause \neq \emptyset$. Note that after one-step pruning, there are at most $k$ maximal connected components, so by the averaging argument, $|\+S_1|\ge \frac{|\prunedBadInterior(c_0)|}{k} \ge \frac{\frozenRho\cdot \alpha n}{k}$. To verify that the components are at most $k$, since each clause $c$ contains at most $k$ variables, each variable in $\vbl(c)$ belongs to at most one $\cup_{e \in \+S_i} \vbl(e)$ for some $i \in [t]$. Hence, the number of components $t \leq k$.

  Finally, we verify that for any clause $c\in \prunedBadInterior(c_0)\cap \blockedclause$ and $v\in \vbl(c)\cap \Gamma(\sigma)$, there exists some $c^\sharp \in \prunedBadInterior(c_0)$ such that $v\in \vbl(c^\sharp)$ and $c^\sharp\in \frozenclauses\cup \badclauses$.

  To begin with, we prove that this condition holds initially. To see this, fix any clause $c\in \badInterior(c_0)\cap \blockedclause$ and $v\in \tp{\vbl(c)\cap \Gamma(\sigma)} \setminus \badvars$ (the lemma holds trivially if $c$ does not exist or $v$ does not exist for $c$). By the definition of blocked clauses, there exists some $c^\sharp\in \frozenclauses$ such that $v\in \vbl(c^\sharp)$. We show that $c^\sharp\in \badInterior(c_0)$ through contradiction. Suppose $c^\sharp\notin \badInterior(c_0)$. By the definition of frozen clauses $\frozenclauses$, $c^\sharp(\sigma)\neq \true$. Combining with \Cref{lem:random-kcnf-neighbor-cut} and the fact that $\vbl(c^\sharp) \cap \vbl(c) \neq \emptyset$, it holds that $\vbl(c^\sharp)\cap \vbl(c)\cap \Gamma(\sigma) = \emptyset$ which reaches a contradiction with the assumption $v\in \vbl(c^\sharp)\cap \vbl(c)\cap \Gamma(\sigma)$.

  Next, we show that after one-step pruning, this condition still holds. Note that by definition, $\blockedclause$ and $\frozenclauses$ are two disjoint sets. For the case that $\prunedBadInterior(c_0) \cap \blockedclause = \emptyset$ before pruning, this condition holds trivially after pruning.
  For the case that $\prunedBadInterior(c_0) \cap \blockedclause \neq \emptyset$ before pruning, it holds that $\+S_1,\dots, \+S_t$ are disconnected in $G_\Phi$ after removing the chosen blocked clause. So this condition still holds; otherwise, they are not disconnected.
\end{proof}

Fix an arbitrary $(S, \sigma)$ generated by $\Pinning$.
We include the following lemma showing that $|\prunedBadInterior(c_0)\cap \blockedclause|$ can be upper bounded using $|\prunedBadInterior(c_0)\cap \frozenclauses|$. This property is useful in later proofs. We remark that this lemma is implicit in the proof of~\cite[Lemma 7.9]{heImprovedBoundsSampling2023}.

{
  \begin{lemma}
    \label{lemma:random-cnf-bounded-blocked-clauses}
    Assume that the conditions in \Cref{lemma:valiant-algorithm-marginal-lower-bound} are satisfied.
    For any $(S, \sigma)$ generated by \Pinning, we have
    \[
    \abs{\prunedBadInterior(c_0)\cap \blockedclause}\le \frac{\frozenEta + 2/k - 2\frozenEta/k}{\frozenConstant - \frozenEta - 2/k + 2\frozenEta/k}\cdot \abs{\prunedBadInterior(c_0)\cap \frozenclauses}.\]
  \end{lemma}
  \begin{proof}
    Let $V_1 = \vbl(\prunedBadInterior(c_0)\cap \frozenclauses)$ be the variables in all frozen clauses in the pruned correlated component, and let $V_2 = \vbl(\prunedBadInterior(c_0)\cap\blockedclause)$ be the variables in all blocked clauses in the pruned correlated component. We remark that $V_1$ and $V_2$ may contain bad variables.

    Next, we give an upper bound of $\abs{V_1\cup V_2}$. We claim that
    \begin{align}\label{eq:random-cnf-bounded-blocked-clauses-upper-bound}
      \abs{V_1 \cup V_2} \le k\abs{\prunedBadInterior(c_0)\cap \frozenclauses} + (1 - \frozenConstant) k \cdot \abs{\prunedBadInterior(c_0)\cap \blockedclause}.
    \end{align}
    To see this, we first count all variables in $\prunedBadInterior(c_0)\cap \frozenclauses$ and include other missing variables in $\prunedBadInterior(c_0)\cap \blockedclause$. It holds that for any clause, there are at most $k$ variables, and this gives the first term. On the other hand, and for any blocked clause $c\in \prunedBadInterior(c_0)\cap \blockedclause$, we have $\abs{\vbl(c)\cap \Gamma(\sigma)\setminus \badvars} > \frozenConstant\cdot k$ and each of these variables is contained in some frozen clause in $\prunedBadInterior(c_0)\cap \frozenclauses$ by \Cref{lemma:random-cnf-pruning-lemma}. So there are at most $(1 - \frozenConstant) k$ variables in $c$ that are not counted yet. This gives the second term.

    Then, we give a lower bound of $\abs{V_1\cup V_2}$. We claim that
    \begin{align}\label{eq:random-cnf-bounded-blocked-clauses-lower-bound}
      \abs{V_1 \cup V_2}\ge (1-\frozenEta)\cdot (k-2)\cdot \tp{\abs{\prunedBadInterior(c_0)\cap \frozenclauses} + \abs{\prunedBadInterior(c_0)\cap \blockedclause}}.
    \end{align}
    To see this, $\Phi$ satisfies \Cref{property:clause-size} and \Cref{property:edge-expansion} with parameters $\rho = \frozenRho$, $\eta = \frozenEta$ and $B_1 = k - 2$ by \Cref{def:well-behaved-random-cnf}.
    Since every clause has size at least $k - 2$, let $c_1, \ldots, c_{\ell}$ be all clauses in $\prunedBadInterior(c_0)\cap (\frozenclauses\cup \blockedclause)$ and $S_i$ be all variables in $c_i$, note that $\ell \le \frozenRho \cdot m$ by the second item of \Cref{lemma:random-cnf-pruning-lemma}, we have
    \begin{align*}
      \abs{V_1 \cup V_2} \ge \abs{\bigcup_{i = 1}^{\ell} S_i} & \ge (1 - \frozenEta)\cdot (k-2)\cdot \ell \\
      & = (1-\frozenEta)\cdot (k-2)\cdot \tp{\abs{\prunedBadInterior(c_0)\cap \frozenclauses} + \abs{\prunedBadInterior(c_0)\cap \blockedclause}},
    \end{align*}
    where the last equation holds because $\frozenclauses$ and $\blockedclause$ are disjoint sets.
    This lemma follows by by combining \eqref{eq:random-cnf-bounded-blocked-clauses-upper-bound}, \eqref{eq:random-cnf-bounded-blocked-clauses-lower-bound}, the fact that $(1-\frozenEta)(k-2) > (1-\frozenConstant)k$ (due to the definitions of parameters in~\eqref{eq:parameters}) and rearranging the terms.
  \end{proof}
    
}

The following result is a direct consequence of \Cref{lemma:random-cnf-pruning-lemma}.
\begin{proposition}
  \label{prop:pruned-bad-interior-size}
  For any $(S, \sigma)$ generated by \emph{$\Pinning$},
  if $\abs{\badInterior(c_0)} \ge \log n$, then $\log n \le |\prunedBadInterior(c_0)| \le \frozenRho \cdot \alpha n$.
\end{proposition}
\begin{proof}
  By \Cref{lemma:random-cnf-pruning-lemma}, it holds that $|\prunedBadInterior(c_0)| \le \frozenRho \cdot \alpha n$ and it suffices to show that $|\prunedBadInterior(c_0)| \ge \log n$.
  If $\abs{\badInterior(c_0)} \le \frozenRho \cdot \alpha n$, then by the definition of $\prunedBadInterior(c_0)$, it holds that $\prunedBadInterior(c_0) = \badInterior(c_0)$ and this proposition holds directly.
  So we assume that $\abs{\badInterior(c_0)} > \frozenRho \cdot \alpha n$. By \Cref{lemma:random-cnf-pruning-lemma}, it holds that $\frac{\frozenRho\cdot \alpha n}{k}\le |\prunedBadInterior(c_0)|\le \frozenRho\cdot \alpha n$ and it holds that $\frac{\frozenRho\cdot \alpha n}{k} \ge \log n$ for any $n$ sufficiently large. The proposition then follows.
\end{proof}

\color{black}

We then show that $|\prunedBadInterior(c_0) \cap \frozenclauses|$ has a lower bound in terms of $|\prunedBadInterior(c_0)|$.
Note that by the definition of frozen clauses, for any clause $c\in \prunedBadInterior(c_0)\cap \frozenclauses$, it holds that $c$ is not satisfied by the partial assignment $\sigma$, i.e. $c(\sigma)\neq \true$. There are at most $\frozenConstant k$ good variables that are not revealed. Meanwhile, by \Cref{obs:goodclauses-bounds}, note that there are at least $\goodlowerbound \defeq (1 - \badclauseConstant) k - 5k^{4 / 5}$ good variables in total, so there are at least $\revealedLowerbound = \goodlowerbound - \frozenConstant\cdot k = (1 - \badclauseConstant - \frozenConstant) k - 5 k^{4 / 5}$ good variables that have been revealed. Note that this matches the setting of $B_2$ in \Cref{def:well-behaved-random-cnf}.

We first lower bound the number of frozen clauses in the pruned associated component.

{
  \begin{lemma}
    Assume that the conditions in \Cref{lemma:valiant-algorithm-marginal-lower-bound} are satisfied.
    \label{lemma:random-cnf-lowerbound-frozen-clauses}
    For any $(S, \sigma)$ generated by \emph{$\Pinning$},
    if $\log n\le |\prunedBadInterior(c_0)| \le \frozenRho\cdot \alpha n$, it holds that
    \begin{align*}
      \abs{\prunedBadInterior(c_0)\cap \frozenclauses} \ge \varrho \cdot \abs{\prunedBadInterior(c_0)}, \quad \text{where } \varrho = \frac{1 - \frac{24k^5}{(1 - \frozenEta)(\badclauseConstant - \frozenEta)\degreeConstant}}{1 + \frac{\frozenEta + 2/k - 2\frozenEta/k}{\frozenConstant - \frozenEta - 2/k + 2\frozenEta/k}}.
    \end{align*}
  \end{lemma}
  \begin{proof}
    By \Cref{property:bounded-bad-clauses}, the assumption that $|\prunedBadInterior(c_0)| \ge \log n$ and the fact that $\badclauses$ is a union of $\sbadclauses$ and at most $k^{4/5}$ clauses, we have
    \begin{align*}
      \abs{\prunedBadInterior(c_0)\cap \badclauses} \le & ~ \frac{12k^5}{(1 - \frozenEta)(\badclauseConstant - \frozenEta)\degreeConstant}\abs{\prunedBadInterior(c_0)} + k^{4 / 5} \\
      \text{(by $n$ is sufficiently large)}\quad\le & ~ \frac{24k^5}{(1 - \frozenEta)(\badclauseConstant - \frozenEta)\degreeConstant}\abs{\prunedBadInterior(c_0)}.
    \end{align*}
    By the definition of $\badInterior(c_0)$, it only contains clauses in $\frozenclauses\uplus \blockedclause \uplus \badclauses$ and $\prunedBadInterior(c_0)$ is a subset of $\badInterior(c_0)$. Hence
    \[
      \abs{\prunedBadInterior(c_0)\cap \tp{\frozenclauses\cup \blockedclause}} \ge \tp{1 - \frac{24k^5}{(1 - \frozenEta)(\badclauseConstant - \frozenEta)\degreeConstant}}\abs{\prunedBadInterior(c_0)}.
    \]

    By \Cref{lemma:random-cnf-bounded-blocked-clauses}, we have that
    \[
      \abs{\prunedBadInterior(c_0)\cap \blockedclause}\le \frac{\frozenEta + 2/k - 2\frozenEta/k}{\frozenConstant - \frozenEta - 2/k + 2\frozenEta/k}\cdot \abs{\prunedBadInterior(c_0)\cap \frozenclauses}.\]
    Finally, by combining the above two inequalities and rearranging the terms, we have
    \begin{align*}
      \abs{\prunedBadInterior(c_0)\cap \frozenclauses} \ge \varrho \cdot \abs{\prunedBadInterior(c_0)}, \quad \text{where } \varrho = \frac{1 - \frac{24k^5}{(1 - \frozenEta)(\badclauseConstant - \frozenEta)\degreeConstant}}{1 + \frac{\frozenEta + 2/k - 2\frozenEta/k}{\frozenConstant - \frozenEta - 2/k + 2\frozenEta/k}}. & \qedhere
    \end{align*}
  \end{proof}
}

{
The following lemma originates from {\cite[Lemma 4.11]{chenCountingRandomSAT2025}}. We slightly modify its statement to fit our setting, which helps us show the diminishing of large associated components.
Recall that by \Cref{obs:goodclauses-bounds}, there are at least $\goodlowerbound \defeq (1 - \badclauseConstant) k - 5k^{4 / 5}$ good variables in total.
\begin{lemma}
  \label{lemma:local-uniformity-new}
  Assume that $\goodlowerbound \ge 10$ and $2^{\goodlowerbound} \ge 2\mathrm{e} k\cdot \degreeConstant\alpha$.
  Let $\varsigma \in \set{\true, \false}^S$ be a feasible partial assignment over $S$, where $S \subseteq \badvars$ is a subset of bad variables. 
  For any subset of good variables $T \subseteq \tp{V \setminus S} \cap \goodvars$, the following holds:
  $$\forall \tau \in \set{\true, \false}^T, \quad \Pr[X \sim \mu_{\Phi}]{X_T = \tau \mid X_S = \varsigma} \le \tp{\frac{1}{2} \exp\tp{\frac{1}{k}}}^{\abs{T}}.$$
\end{lemma}
\begin{proof}
  By the law of total probability, it suffices to show that for any $\omega \in \set{\true, \false}^{\tp{V \setminus S} \cap \badvars}$ with $\Pr[X \sim \mu_{\Phi}]{X_{\badvars \setminus S} = \omega \mid X_S = \varsigma} > 0$, it holds that
  $$\Pr[X \sim \mu_{\Phi}]{X_T = \tau \mid X_S = \varsigma, X_{{\badvars \setminus S}} = \omega} \le \tp{\frac{1}{2} \exp\tp{\frac{1}{k}}}^{\abs{T}}.$$
  To see the above, note that conditioned on $X_S = \varsigma$ and $X_{{\badvars \setminus S}} = \omega$, all bad clauses are satisfied and the simplified CNF formula only contains good clauses. Each remaining good clause has at least $\goodlowerbound$ good variables that are not fixed, and each variable has degree at most $\degreeConstant \alpha$. Then we can apply \Cref{thm:lovasz-local-lemma} by setting the parameter $x(c) = \e \cdot 2^{-\goodlowerbound}$. Note that the condition holds by verifying that
  $$2^{-\goodlowerbound} \le x(c) \prod_{\substack{c' \in \goodclauses \\ \vbl(c) \cap \vbl(c') \neq \emptyset}} (1 - x(c')),$$
  which holds since $\goodlowerbound \ge 10$ and $2^{\goodlowerbound} \ge 2\mathrm{e} k\cdot \degreeConstant\alpha$.
  Thus, by \Cref{thm:lovasz-local-lemma}, let $A$ be the event that $X_T = \tau$ and $\vbl(A)$ be the set of variables that $A$ is defined on, we have
  \begin{align*}
    &\Pr[X \sim \mu_{\Phi}]{X_T = \tau \mid X_S = \varsigma, X_{\badvars \setminus S} = \omega} \\
    \le& 2^{-\abs{T}} \cdot \prod_{\substack{c' \in \goodclauses \\ \vbl(c') \cap \vbl(A) \neq \emptyset}} (1 - x(c'))^{-1}
    \le 2^{-\abs{T}} \cdot \tp{1 - \e \cdot 2^{-\goodlowerbound}}^{-\abs{T} \cdot \degreeConstant \alpha} \\
    \le& \tp{2 \exp\tp{-2\e \cdot 2^{-\goodlowerbound} \cdot \degreeConstant \alpha}}^{-\abs{T}}
    \le \tp{\frac{1}{2} \exp\tp{\frac{1}{k}}}^{\abs{T}}. \qedhere
  \end{align*}
\end{proof}
}

To show the diminishing of large associated components, we are going to apply the local uniformity on these revealed variables.
Fix an arbitrary $(S, \sigma)$ generated by $\Pinning$.
We give a lower bound of $|\tp{\vbl(\sigma) \setminus \badvars} \cap \vbl({\prunedBadInterior(c_0) \cap \frozenclauses})|$ which is a subset of variables in $\sigma$.
We remark here that $\vbl(\sigma) \setminus \badvars$ is the set of revealed variables during the execution of \Pinning, excluding the initial pinning $\sigma^*_{\le i - 1}$, and we actually lower bound the number of revealed good variables in frozen clauses of the pruned associated component.
Recall that $\revealedLowerbound = \goodlowerbound - \frozenConstant\cdot k = (1 - \badclauseConstant - \frozenConstant) k - 5 k^{4 / 5}$ is the minimal number of variables that are revealed in each clause in $\prunedBadInterior(c_0)\cap \frozenclauses$. This matches the setting of $B_2$ of \Cref{property:edge-expansion} in each clause of \Cref{def:well-behaved-random-cnf}.

\begin{lemma}
  \label{lemma:random-cnf-lowerbound-sampled-vars}
  Assume that the conditions in \Cref{lemma:valiant-algorithm-marginal-lower-bound} are satisfied.
  For any $(S, \sigma)$ generated by \emph{$\Pinning$}, it holds that
  $$\abs{\tp{\vbl(\sigma) \setminus \badvars} \cap \vbl\tp{\prunedBadInterior(c_0) \cap \frozenclauses}} \ge (1 - \frozenEta)\cdot \revealedLowerbound\cdot \varrho \cdot \abs{\prunedBadInterior(c_0)},$$
  where
  $$\revealedLowerbound = (1 - \badclauseConstant - \frozenConstant) k - 5 k^{4 / 5}, \quad \varrho = \frac{1 - \frac{24k^5}{(1 - \frozenEta)(\badclauseConstant - \frozenEta)\degreeConstant}}{1 + \frac{\frozenEta + 2/k - 2\frozenEta/k}{\frozenConstant - \frozenEta - 2/k + 2\frozenEta/k}}.$$
\end{lemma}
\begin{proof}
  Due to \Cref{def:well-behaved-random-cnf}, \Cref{property:edge-expansion} holds with parameters $\rho = \frozenRho$, $\eta = \frozenEta$ and $B = \revealedLowerbound$.
  Let $c_1, \ldots, c_{\ell}$ be all clauses in $\prunedBadInterior(c_0)\cap \frozenclauses$ and $S_i = \tp{\vbl(\sigma) \setminus \badvars} \cap \vbl(c_i)$ be all revealed variables in $c_i$, we have
  \begin{align*}
    \abs{\tp{\vbl(\sigma) \setminus \badvars} \cap \vbl\tp{\prunedBadInterior(c_0) \cap \frozenclauses}} & = \abs{\bigcup_{i = 1}^{\ell} S_i} \ge (1 - \frozenEta)\cdot \revealedLowerbound\cdot \ell \\
    & = (1 - \frozenEta)\cdot \revealedLowerbound\cdot \abs{\prunedBadInterior(c_0)\cap \frozenclauses} \\
    & \ge (1 - \frozenEta)\cdot \revealedLowerbound\cdot \varrho \cdot \abs{\prunedBadInterior(c_0)},
  \end{align*}
  where the last inequality follows from \Cref{lemma:random-cnf-lowerbound-frozen-clauses}.
\end{proof}

Now, we are ready to prove \Cref{lemma:random-cnf-bad-interior-tail-bound}. Recall that $(S,Y_S)$ is the output of $\Pinning(Y,v_i^*)$, where $Y \sim \mu^{\sigma^*_{\le i - 1}}_{\Phi}$.
\TailBound*
\begin{proof}
  By \Cref{prop:pruned-bad-interior-size}, if $|\badInterior[Y_S](c_0)|\ge \log n$, then $\log n \le |\prunedBadInterior[Y_s](c_0)| \le  \frozenRho\cdot \alpha n$. So we have 
  \[
    \Pr{\abs{\badInterior[Y_S](c_0)}\ge \log n} = \Pr{\log n \le \abs{\prunedBadInterior[Y_S](c_0)} \le \frozenRho\cdot \alpha n}.
  \]
Fix an arbitrary subset of clause $C^\sharp$ and an arbitrary subset of variable $V^\sharp$, the probability of the event satisfying that $\frozenclauses[Y_S] \cap \prunedBadInterior[Y_S](c_0) = C^\sharp$ and $\vbl(C^\sharp) \cap S = V^\sharp$ can be upper bounded by $\tp{\frac{1}{2} \exp\tp{\frac{1}{k}}}^{|V^\sharp|}$. To see this, note that variables in $V^\sharp$ are all revealed variables during the execution of \Pinning, excluding the initial pinning $\sigma^*_{\le i - 1}$. 
Hence, the event happens only if all revealed variables in $V^\sharp$ take the values that forbid the clauses in $C^\sharp$. The upper bound follows from \Cref{lemma:local-uniformity-new}.

Next, we consider the number of possible $C^\sharp$ and $V^\sharp$. Then, this lemma follows from a union bound over all possible $C^\sharp$ and $V^\sharp$.
Recall that $\prunedBadInterior[Y_S](c_0)$ has the following properties:
\begin{enumerate}
\item $\prunedBadInterior[Y_S](c_0)$ is a connected component in $G_\Phi$;
\item For any frozen clause in $\prunedBadInterior[Y_S](c_0)$, at least $\revealedLowerbound$ variables have been revealed in $S$;
\item \Cref{lemma:random-cnf-lowerbound-sampled-vars} holds: the total number of revealed variables in the frozen clauses of the pruned associated component has a lower bound.
\end{enumerate}

Fix an arbitrary size $\ell$ with $\log n \le \ell \le \frozenRho\cdot \alpha n$.
There are at most $\alpha n\cdot n^3(\mathrm{e}k^2\alpha)^\ell$ choices of possible connected components of size $\ell$. For each connected component, we enumerate all possible choices of frozen clauses and revealed variables in these frozen clauses. We have the following upper bound on the number of choices for possible frozen clauses and revealed variables:
$$\alpha n \cdot n^3(\mathrm{e}k^2\alpha)^\ell \cdot 2^{\ell} \cdot \tp{\sum_{i = 0}^{\ceil{k - \revealedLowerbound}} \binom{k}{i}}^{\ell} \le \alpha n \cdot n^3(\mathrm{e}k^2\alpha)^\ell \cdot \tp{20 k \cdot 2^{10 k^{4 / 5} \log k}}^{\ell}.$$
Note that the number of revealed variables is at least $(1 - \frozenEta)\cdot \revealedLowerbound\cdot \varrho \cdot \ell$ for a fixed $\ell$.

Finally, by the union bound, the probability that $\abs{\badInterior(c_0)}\ge \log n$ is upper bounded by
\begin{equation*}
    \sum_{\ell = \ceil{\log n}}^{\ceil{\frozenRho\cdot \alpha n}} \alpha n \cdot n^3 \tp{\mathrm{e}k^2\alpha}^\ell \cdot \tp{20 k \cdot 2^{10 k^{4 / 5} \log k}}^{\ell} \cdot \tp{\frac{1}{2}\exp\tp{\frac{1}{k}}}^{(1 - \frozenEta)\cdot \revealedLowerbound\cdot \varrho \cdot \ell}. \qedhere
\end{equation*}
\end{proof}

\subsubsection{Putting everything together}\label{sec:proof-random-cnf-marginal-lower-bound-nice-probability}

\begin{proof}[Proof of \Cref{lemma:random-cnf-marginal-lower-bound-nice-probability}]
  As discussed in the \hyperref[proof:random-cnf-marginal-lower-bound-nice-probability-easy-case]{proof for the easy case}, the lemma holds when $\alpha < 1 / k^3$ or there is no clause containing $v^*_i$. It then suffices to prove the lemma when $\alpha \ge 1 / k^3$ and there is at least one clause containing $v^*_i$.
  We then show that the random process given in \Cref{alg:reveal} outputs a nice pinning as defined in \Cref{def:nice-pinning} with probability at least $1/2$.

  By \Cref{lemma:random-cnf-reveal-is-gibbs}, the process is indeed a conditional Gibbs revealing process, and the first two properties always hold. 
  It suffices to show that the returned pinning satisfies \Cref{item:nice-pinning-main} with probability at least $1 / 2$.
  Let $(S, Y_S)$ be the output of $\Pinning(Y, v_i^*)$ where $Y \sim \mu_\Phi^{\sigma^*_{\le i - 1}}$. As discussed in \Cref{remark:random-cnf-associated-component-size}, with probability $1$, the following definitions yield well-defined objects.
  Let $\Phi'$ be the simplified formula after applying the simplification process on $\Phi$ given $Y_S$. Let $C'$ be the maximal connected component of $G_{\Phi'}$ such that $v_i^* \in \vbl(C')$. 
  By going through the proof for the easy case and taking \Cref{obs:random-ksat-frozen-clause} into consideration, $C'$ satisfies \Cref{item:nice-pinning-1} and \Cref{item:nice-pinning-2} with probability $1$.

  To verify \Cref{item:nice-pinning-3}, it suffices to show that the probability that the size of the associated component with respect to the pinning $(S, Y_S)$ is at least $\log n$ is at most $1/2$.
  By \Cref{remark:random-cnf-associated-component-size} and \Cref{lemma:random-cnf-bad-interior-tail-bound}, we have
  \begin{align*}
    \Pr{C' \ge \log n} &\le \Pr{\abs{\badInterior[Y_S](c')}\ge \log n} \\
    &\le \sum_{\ell = \ceil{\log n}}^{\ceil{\frozenRho\cdot \alpha n}} \alpha n \cdot n^3 \tp{\mathrm{e}k^2\alpha}^\ell \cdot \tp{20 k \cdot 2^{10 k^{4 / 5} \log k}}^{\ell} \cdot \tp{\frac{1}{2}\exp\tp{\frac{1}{k}}}^{(1 - \frozenEta)\cdot \revealedLowerbound\cdot \varrho \cdot \ell} \\
    &\le \alpha n^4 \sum_{\ell = \ceil{\log n}}^{\ceil{\frozenRho\cdot \alpha n}} \sqb{\mathrm{e}k^2\alpha \cdot 20 k \cdot 2^{10 k^{4 / 5} \log k} \cdot \tp{\frac{1}{2}\exp\tp{\frac{1}{k}}}^{(1 - k^{-2 / 5}) \cdot (1 - 2k^{-1/5}) \cdot (k - 8 k^{4 / 5})}}^{\ell} \\
    &\le \alpha n^4 \sum_{\ell = \ceil{\log n}}^{\ceil{\frozenRho\cdot \alpha n}} \sqb{\mathrm{e}k^2\alpha \cdot 20 k \cdot 2^{10 k^{4 / 5} \log k} \cdot 2^{-(k - 16 k^{4 / 5})}}^{\ell} \le \alpha n^4 \sum_{\ell = \ceil{\log n}}^{+\infty} 2^{-8 \ell} \le 1 / 2,
  \end{align*}
  where the third inequality follows by plugging the parameters for $\varrho$:
  $$\varrho = \frac{1 - \frac{24k^5}{(1 - k^{-2/5})(k^{-1/5} - k^{-2/5})\cdot 12 k^7}}{1 + \frac{k^{-2/5} + 2k^{-1} - 2k^{-7/5}}{k^{-1/5} - k^{-2/5} - 2k^{-1} + 2k^{-7/5}}} \ge \frac{1 - 8k^{-9/5}}{1 + 2k^{-1/5}} \ge 1 - 2k^{-1/5}.$$
  Combining the above, the lemma holds.
\end{proof}

\subsection{Proof of well-behavedness of random CNF formulas}\label{sec:proof-well-behavedness-random-cnf}

\begin{fact}
    \label{lemma:wel-behaved-distinct-vbl}
    Let $k$ and $\alpha$ be two constants.
    For $n$ large enough, with probability $1 - o(1 / n)$ over the random formula $\Phi = \Phi(k, n, m = \floor{\alpha n})$, $\abs{\vbl(c)} \ge k - 2$ holds for every $c \in \+C$.
  \end{fact}
  \begin{proof}
    For a fixed clause $c \in \+C$, the probability that $\abs{\vbl(c)} < k - 2$ is at most
    $$\sum_{j = 1}^{k - 3} \binom{n}{j} \tp{\frac{j}{n}}^k \le k \tp{\frac{\mathrm{e}n}{k - 3}}^{k - 3} \tp{\frac{k - 3}{n}}^k \le \frac{k^4 \mathrm{e}^{k - 3}}{n^3}.$$
    The lemma follows from a union bound over all $m \le \alpha n$ clauses.
  \end{proof}

  \begin{fact}
    \label{lemma:well-behaved-shared-vbl}
    Let $k$ and $\alpha$ be two constants.
    For $n$ large enough, with probability $1 - o(1 / n)$ over the random formula $\Phi = \Phi(k, n, m = \floor{\alpha n})$, $\abs{\vbl(c) \cap \vbl(c')} \le 3$ holds for every two distinct clauses $c, c' \in \+C$.
  \end{fact}
  \begin{proof}
    For a pair of distinct clauses $c, c' \in \+C$, the probability that $\abs{\vbl(c) \cap \vbl(c')} > 3$ is at most
    $$\frac{\binom{n}{4} \cdot (4k)^4 \cdot (4k)^4 \cdot n^{k - 4} \cdot n^{k - 4}}{n^k \cdot n^k} \le \frac{4^8 \cdot k^8}{n^4}.$$
    The lemma follows from a union bound over all $\binom{m}{2} = O_\alpha(n^2)$ pairs of clauses.
  \end{proof}
  
  \begin{lemma}[{\cite[Lemma A.6]{chenCountingRandomSAT2025}}]
    \label{lemma:random-cnf-growth-rate-graph}
    Let $k$ and $\alpha$ be two constants.
    Suppose $\alpha \le 2^k$. With probability $1 - o(1/n)$ over the random formula $\Phi = \Phi(k, n, m = \floor{\alpha n})$ with fixed density $\alpha$, $H_\Phi$ satisfies that for every clause $c$ in $\Phi$ and $\ell \ge 1$, there are at most $n^3(\mathrm{e}k^2\alpha)^\ell$ connected sets of clauses in $G_\Phi$ that contain $c$ and have size $\ell$.
  \end{lemma}
  
  \begin{lemma}[{\cite[Lemma A.14]{chenCountingRandomSAT2025}}]
      \label{lemma:random-cnf-bounded-fraction-bad-clauses-graph}
      For any fixed $k$ and $\alpha$, assume $\eta, \rho, \degreeConstant, \badclauseConstant$ are parameters satisfying that\footnote{The statement here is slightly different from that in the original paper, where the condition $\tp{\mathrm{e}(\rho k \alpha)^\eta}^k \le \rho^2$ assumed here is stronger than $\mathrm{e}(\rho k \alpha)^\eta \le 1$ in the original paper, since $\rho<1$. Hence, we can use the same result because we assume a stronger condition.}
      \begin{enumerate}
        \item $\eta k \ge 4$, $\rho < 1$, $\badclauseConstant \ge \eta + 1/k$;
        \item $6k^5 \le \degreeConstant \le \-e^{k-2}\alpha$;
        \item $\tp{\mathrm{e}\tp{\rho k \alpha}^\eta}^k\le \rho^2$.
      \end{enumerate}
      Then, with probability $1 - o(1/n)$ over the random formula $\Phi = \Phi(k, n, m = \floor{\alpha n})$, for any $\+C' \subseteq \+C$ of size $\abs{\+C'} \ge \log n$ connected in the line graph of $H_\Phi = (V, \+C)$ (namely, connected in $G_\Phi$), it holds that 
      $$\abs{\+C' \cap \sbadclauses} \le \frac{12k^5}{(1-\eta)(\badclauseConstant - \eta)\degreeConstant} \abs{\+C'}.$$
  \end{lemma}

\begin{lemma}\label{lemma:random-cnf-expansion}
    For any fixed $k$ and $\alpha$, assume $\eta, \rho, B$ are parameters satisfying that
    \begin{enumerate}
      \item $\eta B \ge 4$, $\rho < 1$;
      \item $2^k \cdot \mathrm{e}^{2 \cdot B} \cdot \tp{\rho \cdot B \cdot \alpha}^{\eta \cdot B} \le \rho^2$.
    \end{enumerate}
    Then, for any $n$ sufficiently large, with probability $1 - o(1/n)$ over the random formula $\Phi = \Phi(k, n, m = \floor{\alpha n})$, for any $\ell \le \rho \abs{\+C_{\Phi}}$, any $\ell$ clauses $c_1, c_2, \ldots, c_\ell \in \+C_{\Phi}$, and any variable sets $S_1, S_2, \ldots, S_\ell$ where $\forall i\in [\ell]$, $S_i\subseteq \vbl(c_i)$ and
    $\abs{S_i} \ge B$,
    it holds that
    \[\abs{\bigcup_{i\in [\ell]} S_i} > (1 - \eta)\cdot B\cdot \ell. \]
  \end{lemma}
  \begin{proof}
    For $\ell \le \rho \abs{\+C_{\Phi}}$, let $r = \floor{(1 - \eta)\cdot B\cdot \ell}$.
    Define the bad event $\+B_\ell$ as follows:
    there exists a subset $U \subseteq V_{\Phi}$ of size $r$, $\ell$ clauses $c_1, c_2, \ldots, c_\ell \in \+C_{\Phi}$ and $\ell$ subsets of variables $S_1, S_2, \ldots, S_\ell$ where $\forall i\in [\ell]$, $S_i\subseteq \vbl(c_i)$ and $\abs{S_i} \ge B$, satisfying that $S_i \subseteq U$ for all $i\in [\ell]$.
    We then bound the probability of $\+B_\ell$.
    \begin{align*}
      \Pr{\+B_\ell} & \le \binom{n}{r} \cdot \binom{m}{\ell} \cdot 2^{k \ell} \cdot \tp{\frac{r}{n}}^{B\cdot \ell} \le 2^{k \ell} \tp{\frac{\mathrm{e} n}{r}}^r \tp{\frac{\mathrm{e} m}{\ell}}^\ell \tp{\frac{r}{n}}^{B\cdot \ell} \\
      & \le 2^{k \ell} \tp{\frac{\mathrm{e} n}{(1 - \eta)\cdot B\cdot \ell}}^{(1 - \eta)\cdot B\cdot \ell} \tp{\frac{\mathrm{e} \alpha n}{\ell}}^{\ell} \tp{\frac{(1 - \eta)\cdot B\cdot  \ell}{n}}^{B\cdot \ell}    \\
      & = \tp{2^k \cdot \alpha \cdot \mathrm{e}^{(1 - \eta) \cdot B + 1} \cdot \tp{(1 - \eta) \cdot B}^{\eta \cdot B} \cdot \tp{\frac{\ell}{n}}^{\eta \cdot B - 1}}^{\ell}                                           \\
      & \le \tp{\alpha \cdot \tp{2^k \cdot \mathrm{e}^{2\cdot B} \cdot B^{\eta \cdot B}} \cdot \tp{\frac{\ell}{n}}^{\eta \cdot B - 1}}^{\ell}.
    \end{align*}
    On one hand, if $\ell < n^{1 / 3}$,
    $$\Pr{\+B_\ell} \le \alpha \cdot \tp{2^k \cdot \mathrm{e}^{2 \cdot B} \cdot B^{\eta \cdot B}} \cdot n^{-\frac{2}{3} (\eta \cdot B - 1)} \le \alpha \cdot \tp{2^k \cdot \mathrm{e}^{2 \cdot B} \cdot B^{\eta \cdot B}} \cdot n^{-2}.$$
    where the last inequality holds since $\eta \cdot B \ge 4$.
    On the other hand, if $n^{1 / 3} \le \ell \le \rho m$,
    \begin{align*}
      \Pr{\+B_\ell} & \le \tp{\alpha \cdot \tp{2^k \cdot \mathrm{e}^{2 \cdot B} \cdot B^{\eta \cdot B}} \cdot \tp{\alpha\cdot \rho}^{\eta \cdot B - 1}}^{\ell} \\
                    & = \tp{\rho^{-1} \cdot \tp{2^k \cdot \mathrm{e}^{2 \cdot B} \cdot \tp{\rho \cdot B \cdot \alpha}^{\eta \cdot B}}}^{\ell} \le \rho^{n^{1 / 3}} \le n^{-3},
    \end{align*}
    where we apply the assumption that $2^k \cdot \mathrm{e}^{2 \cdot B} \cdot \tp{\rho \cdot B \cdot \alpha}^{\eta \cdot B} \le \rho^2$ and the last inequality $\rho^{n^{1 / 3}} \le n^{-3}$ holds because $\rho < 1$ is a constant.
  
    By a union bound over all $\ell \le \rho m$, we have $\sum_{\ell = 1}^{\floor{\rho m}} \Pr{\+B_\ell} \le o(1 / n)$ and the lemma follows.
  \end{proof}
  
  Finally, we prove \Cref{lemma:well-behaved-random-cnf}, which is a direct consequence of the above lemmas.
\begin{proof}[Proof of \Cref{lemma:well-behaved-random-cnf}]
We first consider the case $\alpha>1/k^3$.
By \Cref{lemma:wel-behaved-distinct-vbl}, \Cref{lemma:well-behaved-shared-vbl} and \Cref{lemma:random-cnf-growth-rate-graph}, with probability $1 - o(1 / n)$, $\Phi$ satisfies \Cref{property:clause-size} (Bounded clause size), \Cref{property:clause-intersection} (Bounded intersection) and \Cref{property:bounded-growth-rate} (Bounded growth rate).
Plugging in the parameters in \eqref{eq:parameters} into \Cref{lemma:random-cnf-bounded-fraction-bad-clauses-graph}, with probability $1 - o(1/n)$, $\Phi$ satisfies \Cref{property:bounded-bad-clauses} (Bounded bad clauses) with the desired parameters.
Plugging in the parameters in \eqref{eq:parameters} into \Cref{lemma:random-cnf-expansion} and \Cref{lemma:random-cnf-degree-one-variable}, we conclude that with probability $1 - o(1 / n)$, $\Phi$ satisfies \Cref{property:edge-expansion} (Edge expansion) and \Cref{property:degree-one-variable} (Degree-one variable property) with the desired parameters.
The lemma follows by a union bound. 
For the case $\alpha \le 1/k^3$, the proof is the same, except that we do not need to show \Cref{property:bounded-bad-clauses} (Bounded bad clauses).
\end{proof}
\section{Information-theoretic lower bounds of sample complexity}

\subsection{Preliminaries of information theory}

Let $X\in\+X$ be a discrete random variable over a finite set $\+X$ with $|\+X|\ge 2$.
Define the \emph{entropy} of $X$ as $H(X) \triangleq -\sum_{x \in \+X} \Pr{X = x} \ln \Pr{X = x}$. 
Let $(X,Y) \in \+X \times \+Y$ be a joint random variable. 
Defined the \emph{conditional entropy} of $X$ given $Y$ as 
$H(X|Y) \triangleq -\sum_{x \in \+X, y \in \+Y} \Pr{X = x, Y = y} \ln \Pr{X = x | Y = y}$.
The \emph{mutual information} of $X$ and $Y$ is defined as $I(X; Y) \triangleq H(X) - H(X|Y)$.
Define the \emph{binary entropy} function $H_b: [0, 1] \to \mathbb{R}$ as $H_b(p) \triangleq -p\ln p - (1-p)\ln(1-p)$. 

\begin{lemma}[Fano's inequality~\cite{cover2006elements}]
  \label{lemma:fano-inequality}
  For any Markov chain $X\to Y\to \widehat{X}$,
  \begin{equation*}
    H_b\tp{\Pr{\widehat{X}\neq X}} + \Pr{\widehat{X}\neq X}\ln(|\+X|  - 1) \ge H(X| \widehat{X}).
  \end{equation*}
  In particular, if $X$ is uniformly distributed over the set $\+X$ and hence $H(X) = \ln |\+X|$, then
  \begin{equation}
    \label{eq:classic-fano}
    \Pr{\widehat{X}\neq X}\ge 1 - \frac{I(X; Y) + \ln2}{\ln |\+X|}.
  \end{equation}
\end{lemma}

Let $\rho: \+X \times \+X \to \=R$ be a symmetric function. For any scalar $t\ge 0$, 
define the maximum and minimum neighborhood sizes around a point at radius $t$ as follows:
\begin{equation*}
  N^{\max}_t \triangleq \max_{x\in \+X} | \{x'\in \+X \,|\, \rho(x, x')\le t\} |,
  \quad N^{\min}_t \triangleq \min_{x\in \+X} | \{x'\in \+X \,|\, \rho(x, x')\le t\} |.
\end{equation*}

\begin{lemma}[{Distance-based Fano's inequality~\cite{duchi2013distance}}]
  \label{lemma:distance-fano-inequality}
  For any Markov chain $X\to Y \to \widehat{X}$, let $P_t = \Pr{\rho(\widehat{X}, X)\ge t}$, it holds that
  \[
    H_b(P_t) + P_t \ln\left( \frac{|\+X| - N_t^{\min}}{N_t^{\max}}\right) + \ln (N_t^{\max}) \ge H(X | \widehat{X}).
  \]
  In particular, if $X$ is uniformly distributed over the set $\+X$ and $|\+X| - N_t^{\min} > N_t^{\max}$, then
  \begin{equation}
    \label{eq:distance-based-fano}
    \Pr{\rho(\widehat{X}, X) > t} \ge 1 - \frac{I(X; Y) + \ln2}{\ln\left( {(|\+X| - N_t^{\min})}/{N_t^{\max}} \right)}.
  \end{equation}
\end{lemma}

\begin{remark}
  \cite[Corollary 1]{duchi2013distance} claimed a slightly stronger lower bound of $\mathbb{P}[\rho(\widehat{X}, X) > t]$. The weaker version stated in~\eqref{eq:distance-based-fano} suffices for our purposes.
\end{remark}

\subsection{Sample complexity of exact learning CNF formulas with disjoint clauses}
We prove the lower bound in \Cref{theorem:lower-bound-simple-thm}.
The proof is a simple application of \Cref{lemma:fano-inequality}.

\lowerBoundSublinearIntersection*

\begin{proof}
  We construct a simple CNF with $n$ variables, 
where all variables are partitioned into $k$ groups and each group has exactly $n/k$ variables. 
Say the $i$-th group contains all variables with label between $(i-1)n/k + 1$ and $in/k$.
We construct $k$ disjoint clauses, where each clause picks one variable from each group.
Formally, let $(p_i)_{i\in [k-1]}$ be $k-1$ permutations, 
where each $p_i$ is a permutation over the set $[n/k]$. 
We construct $n/k$ clauses, 
for each $i\in [n/k]$, let variables $\{i, n/k+p_1(i), 2n/k+p_2(i), \dots, n - n/k + p_{k-1}(i)\}$ be 
a clause which forbids all-False assignments. 
Note that these $n/k$ clauses are disjoint, which implies $d = 1$ and $s = 0$. 
Note that the above construction is uniquely determined by the set of permutations $(p_i)_{i\in [k-1]}$. 
To prove the lower bound, we consider the following random simple CNF formulas $\Phi$:
\begin{itemize}
  \item independently sample $k-1$ permutations $(p_i)_{i\in [k-1]}$ uniformly at random;
  \item construct the CNF formula $\Phi$ as described above using the permutations $(p_i)_{i\in [k-1]}$.
\end{itemize}

Hence, the random variable $\Phi$ is drawn from a uniform distribution. Let $\+X$ denote the support of $\Phi$. 
It holds that $|\+X| = \tp{\frac{n}{k}!}^{k-1}$ and
\begin{align}\label{eq:X-size}
  \ln|\+X|\ge (k-1)\frac{n}{k}\ln\tp{\frac{n}{\mathrm{e}k}}.
\end{align}
Let $\mu_{\Phi}$ denote the uniform distribution over all satisfying assignments of $\Phi$.

Let $X_1, X_2, \dots, X_T \sim \mu_{\Phi}$ be $T$ samples from $\mu_{\Phi}$.
Let $\Phi'$ be the CNF formula returned by a learning algorithm given samples $X_1, X_2, \dots, X_T$. 
The following process forms a Markov chain:
\begin{equation*}
  \Phi \to (X_1, X_2, \dots, X_T) \to \Phi'
\end{equation*}
We use \Cref{lemma:fano-inequality} to show that 
$\Pr{\Phi\neq \Phi'} \geq \frac{9}{10}$ if $T\le \frac{k-1}{10k\ln2}\ln(\frac{N}{\mathrm{e}k}) - \frac{1}{N}$, 
which implies that any algorithm that exactly learns the product CNF formulas with probability 
at least $\frac{1}{3}$ requires at least $\Omega(\log n)$ samples.
This proves the theorem.

By \Cref{lemma:fano-inequality}, 
it suffices to show that $\frac{ I(X_1, X_2,\dots,X_T; \Phi) + \ln2}{\ln|\+X|}\le \frac{1}{10}$.
Using the chain rule,
\begin{align}\label{eq:I-X-Phi}
  I(X_1, X_2,\dots,X_T; \Phi) = \sum_{i=1}^T I(X_i; \Phi~|~X_1, X_2,\dots,X_{i-1})\le \sum_{i=1}^T I(X_i; \Phi) \leq T\cdot n\ln2,
\end{align}
where the last inequality is due to $I(X_i; \Phi)\le H(X_i)\le n\ln2$ because $X_i$ is an $n$-bit string. 
Combining~\eqref{eq:X-size} and \eqref{eq:I-X-Phi}, 
it holds that if $T\le \frac{k-1}{10k\ln2}\ln(\frac{n}{\mathrm{e}k}) - \frac{1}{n}$, 
then $\frac{ I(X_1, X_2,\dots,X_T; \Phi) + \ln2}{\ln|\+X|}\le \frac{1}{10}$.
\end{proof}

\subsection{Sample complexity of approx. learning CNF formulas in the local lemma regime}\label{sec:sample-complexity-of-approx-learning-cnf-formulas-in-the-local-lemma-regime}

We prove the lower bound in \Cref{theorem:lower-bound-general-thm}. We need to use the following gadgets.
\begin{definition}[Unrestricted gadgets and restricted gadgets]
  \label{definition:lower-bound-gadgets}
  Let $k \geq 2,\ell \geq 1$ be two integers. 
  Given a variable set $U\triangleq \set{v_{i,j}: i \in [\ell], j \in [k]}$ of size $k\ell$, 
  we construct two types of $(k,k,k-1)$-CNF formulas: 
  unrestricted gadgets $\Phi_{un}^U=(U,\+C_{un}^U)$ and restricted gadgets $\Phi_{res}^U=(U,\+C_{res}^U)$.
  \begin{itemize}
    \item Arbitrarily arrange $k\ell$ variables into $\ell$ layers and each layer contains exactly $k$ variables. 
    Let $v_{i,j}$ denote the $j$-th variable in the $i$-th layer. Let $\+C$ be an empty set at the beginning.
    \item For each $i$ from $1$ to $\ell - 1$, 
    we construct $k$ clauses $c_{ij}$ for $j$ from $1$ to $k$ and add them into the set $\+C$. 
    The clause $c_{ij}$ is constructed as follows: 
    it contains all variables in the $i$-th layer except for the $j$-th variable and 
    it contains the $j$-th variable in the $(i+1)$-th layer. 
    Formally, $\vbl(c_{ij}) = \set{v_{i,r} : r \neq j} \cup \set{v_{i+1,j}}$. 
    If $i$ is an odd number, $c_{ij}$ forbids all-True assignments of $\vbl(c_{ij})$. 
    Otherwise, $c_{ij}$ forbids all-False assignments of $\vbl(c_{ij})$.
    \item We create an additional clause $c$. 
    It contains all variables in the first layer, and it forbids all-True assignments of $\vbl(c)=\{v_{1,j}: j \in [k]\}$. We remark that $c \notin \+C$.
  \end{itemize}

  So far, 
  we have constructed a clause set $\+C$ with $k(\ell-1)$ clauses and an additional clause $c$. 
  The \emph{unrestricted depth-$\ell$ gadget} is defined by $\Phi_{un}^U= (U,\+C_{un}^U)$, 
  where $\+C_{un}^U = \+C$. 
  The \emph{restricted depth-$\ell$ gadget} is defined by $\Phi_{res}^U= (U,\+C_{res}^U)$, 
  where $\+C_{res}^U = \+C \cup \set{c}$.
  See \Cref{fig:gadgets} for an illustration.

  In both $\Phi_{un}^U$ and $\Phi_{res}^U$, 
  the degree of each variable is most $d = k,$ and two clauses share at most $s = k - 1$ variables. 
  Hence, both of them are $(k, k, k-1)$-CNF formulas.
\end{definition}

\begin{figure}[ht]
  \centering
  \includegraphics[trim={0 2.5cm 0 2.5cm},clip,width=0.5\textwidth]{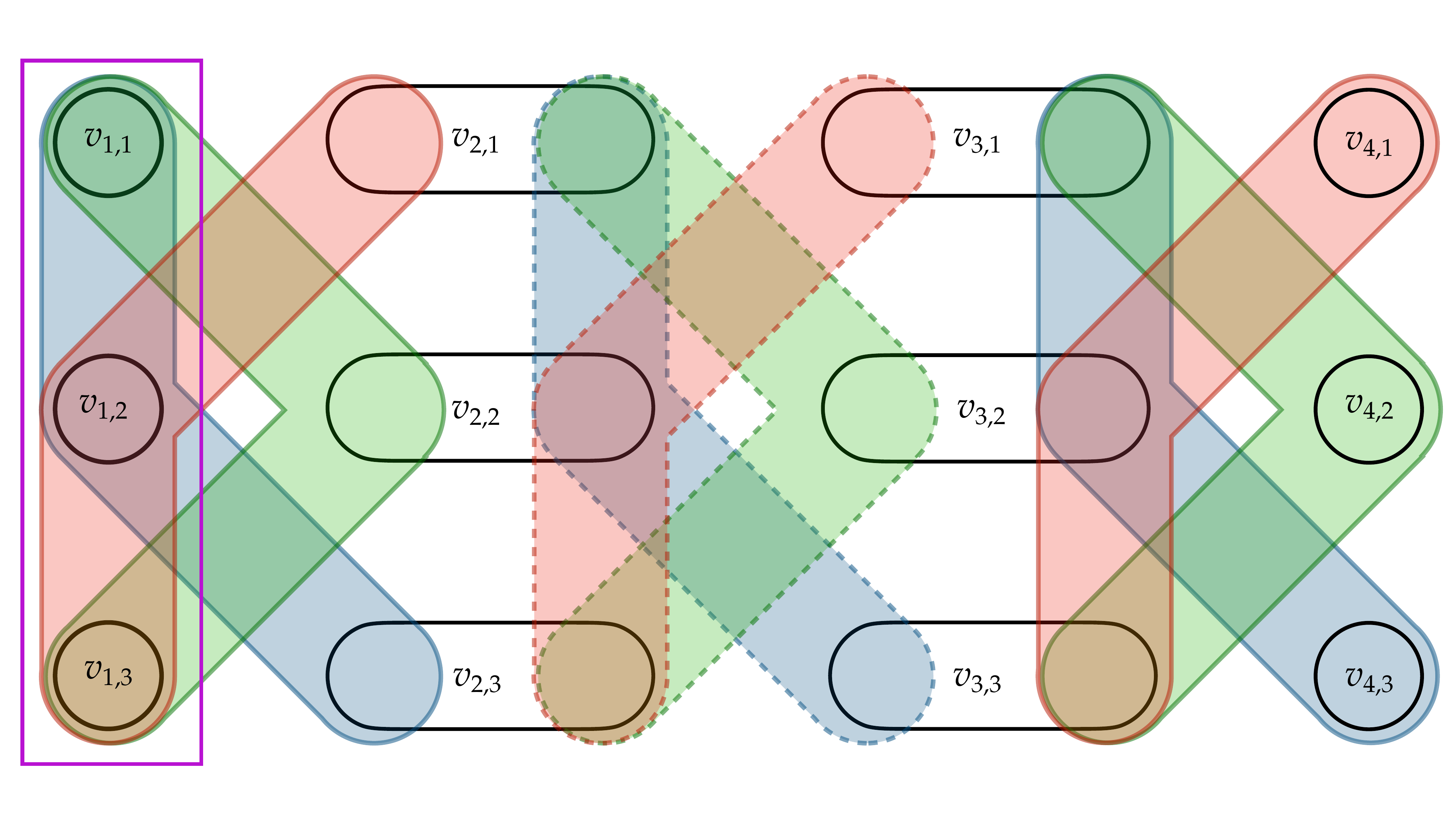}
  \caption{An illustration of the depth-$4$ gadgets for $k = 3$, 
  where black-bordered shapes denote variables $v_{i,j}$ and colored shapes denote clauses. Clauses with solid borders forbid all-True assignments 
  and clauses with dashed borders forbid all-False assignments. 
  The leftmost clause with a purple boundary is the restricted clause $c$.
  For clarity, variables $v_{2,\cdot}$ and $v_{3,\cdot}$ in the second and third layers are intentionally widened to better display the hyperedges.}
  \label{fig:gadgets}
\end{figure}

Next, we provide some basic properties about the unrestricted gadgets and restricted gadgets.

\begin{lemma}
  \label{lemma:lower-bound-properties-gadgets}
  Given an integer $\ell > 0$, 
  let $\Omega_u$ be the set of satisfying assignments of an unrestricted depth-$\ell$ gadget 
  and $\Omega_r$ be the set of satisfying assignments of a restricted depth-$\ell$ gadget. 
  We have that $\Omega_r\subseteq \Omega_u$ and the following bounds hold:
  \begin{equation*}
    1 - 2^{-(k-2) \ell}\le \frac{|\Omega_r|}{|\Omega_u|} \le 1 - 2^{-k\ell}.
  \end{equation*}
\end{lemma}
\begin{proof}
  Recall that in \Cref{definition:lower-bound-gadgets}, 
  the clause set of the restricted depth-$\ell$ gadget is a superset of the clause set of the unrestricted one,
  which implies that $\Omega_r \subseteq \Omega_u$.
  Moreover, we claim that $|\Omega_u \setminus \Omega_r| = 1$.
  To verify this, observe that for any $\sigma \in \Omega_u \setminus \Omega_r$,
  all variables in the first layer are assigned True. 
  By construction, this forces all variables in the second layer to be assigned False. 
  One can verify that all odd layers are assigned True,
  and all even layers are assigned False. 
  Therefore, there is exactly one satisfying assignment in $\Omega_u \setminus \Omega_r$.

  Then, by the fact that $|\Omega_u|\le 2^{k\ell}$, 
  $\frac{|\Omega_r|}{|\Omega_u|} \le 1 - 2^{-k\ell}$ holds directly.
  We next lower bound $|\Omega_u|$. 
  For each layer $i$, if $i$ is an odd number, 
  we assign False to the first variable and last variable in the $i$-th layer; 
  otherwise, 
  we assign True to the first variable and last variable in the $i$-th layer. 
  Note that after fixing these $2\ell$ variables, 
  all clauses in the unrestricted gadget are satisfied. 
  Therefore, we have $|\Omega_u| \ge 2^{(k-2)\ell}$, 
  which implies that  $1 - 2^{-(k-2)\ell} \le \frac{|\Omega_r|}{|\Omega_u|}$.
\end{proof}

We use the gadgets in \Cref{definition:lower-bound-gadgets} to construct a set $\+X$ of $(k,k,k-1)$-CNF formulas. 
Then we can define the uniform distribution over all CNF formulas in $\+X$
to use Fano's inequality.

\begin{definition}[Set of hard CNF formulas $\+X$]\label{definition:lower-bound-set-of-hard-cnf-formulas}
  Let $k,\ell,m \geq 1$ be three integers. Let $V$ be a set of variables with size $mk\ell$. 
  Let $U_1\uplus U_2\uplus \dots \uplus U_m$ be a partition of $V$ into $m$ subsets, 
  where each $U_i$ has size $k\ell$. 
  The set $\+X \triangleq \set{\Phi_i=(V,\+C_i)\mid 0\leq i < 2^m}$ is a set of $(k,k,k-1)$-CNF formulas, 
  where for each $0 \leq i < 2^m$, the CNF formula $\Phi_i = (V,\+C_i)$ is constructed as follows:
  \begin{itemize}
    \item write the integer $i$ as a binary string of length $m$, 
    let $i_j \in \{0,1\}$ be the $j$-th bit of $i$;
    \item for any $1 \leq j \leq m$, if $i_j = 0$, 
    then construct an unrestricted depth-$\ell$ gadget on the variables in $U_j$; 
    otherwise, construct a restricted depth-$\ell$ gadget on the variables in $U_j$.
  \end{itemize}
\end{definition}
\begin{figure}[ht]
  \centering
  \includegraphics[trim={0 12cm 0 12cm},clip,width=0.9\textwidth]{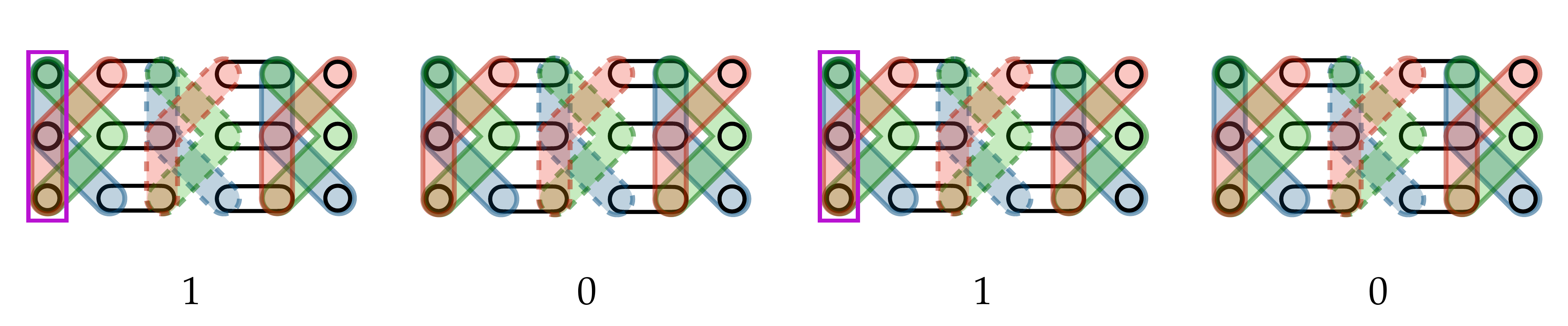}
  \caption{An illustration of hard CNF formulas $\Phi_i$ for $k = 3$, $\ell = 4$, $m=4$ and $i=(1010)_2$.}
  \label{fig:hard-cnf-formulas}
\end{figure}

Let $\Phi_i$ be a CNF formula in $\+X$, 
$\Omega_i$ be the set of all satisfying assignments of $\Phi_i$,
and $\mu_i \triangleq \mu_{\Phi_i}$ be the uniform distribution on $\Omega_i$.
For two integer $i, j \ge 0$, let $\binaryDistance(i, j)$ be 
the number of different locations for the binary representation of $i$ and $j$. 
Define $\gamma \triangleq \frac{|\Omega_r|}{|\Omega_u|}$, 
where $\Omega_r$ and $\Omega_u$ in \Cref{lemma:lower-bound-properties-gadgets} 
is the set of satisfying assignments of the restricted and unrestricted depth-$\ell$ gadget, respectively.

We have the following lemma about the total variation distance between $\mu_i$ and $\mu_j$.
\begin{lemma}
  \label{lemma:lower-bound-tv-distance-binary-distance}
  For $0 \leq i, j < 2^m$, if $m\cdot 2^{-k\ell}<\frac{1}{2}$, 
  then it holds that $\dtv(\mu_i, \mu_j) \ge \binaryDistance(i, j)\cdot 2^{-k \ell - 2}$.
\end{lemma}
\begin{proof}
  Let $p = \binaryDistance(i, j)$, and we assume that $i$ and $j$ differ at the first $p$ bits without loss of generality.
  Let $\Omega = \Omega_i \cup \Omega_j$. Next, we define $p$ disjoint subsets $(S_t)_{t\in [p]}$ of $\Omega$, 
  where $S_t \subseteq \Omega$ is the subset of assignments $\sigma \in \Omega$ satisfying the following two conditions.
  \begin{itemize}
    \item The first layer variables for the previous $t-1$ gadgets are not fully assigned True. 
    Formally, for any $s \leq t- 1$, in the gadget constructed on $U_s$, 
    let $U_{s1} \subseteq U_s$ be the first layer variables where $|U_{s1}| = k$. 
    Then, there exists $v \in U_{s1}$ such that $\sigma(v) = \text{False}$.
    \item The first layer variables of the $t$-th gadget are assigned True.
    Formally,  for any $v \in U_{t1}$, where $U_{t1} \subseteq U_t$ is the first layer variables, 
    it holds that $\sigma(v) = \text{True}$ for all $v \in U_{t1}$.
  \end{itemize}

  By the definition of total variation distance, 
  $\dtv(\mu_i, \mu_j) \ge \frac{1}{2}\sum_{t=1}^p \sum_{\sigma \in S_t} |\mu_i(\sigma) - \mu_j(\sigma)|$.
  Note that for each $t\in [p]$, due to the second step of the above construction, 
  it holds that either $\mu_i(S_t) = 0$ or $\mu_j(S_t)=0$ 
  because the all-True assignment in the first layer violates the restricted gadget.
  So $\sum_{\sigma\in S_t}|\mu_i(\sigma) - \mu_j(\sigma)| = \max\{\mu_i(S_t), \mu_j(S_t)\}$.
  To lower bound $\dtv(\mu_i, \mu_j)$, it suffices to lower bound $\max\{\mu_i(S_t), \mu_j(S_t)\}$ for each $t\in [p]$.
  Recall that $\gamma = \frac{|\Omega_r|}{|\Omega_u|}$.
  We claim that
  \[\max\{\mu_i(S_t), \mu_j(S_t)\} \ge \gamma^{t-1}(1 - \gamma).\]
  To verify the above inequality, suppose $\mu_i(S_t) > 0$, 
  then the $t$-th gadget in $\Phi_i$ must be unrestricted (which contributes a factor of $1-\gamma$) 
  and the worst case is that all first $t-1$ gadgets are unrestricted (which contributes a factor of $\gamma^{t-1}$).
  Therefore,
  \begin{equation*}
    \dtv(\mu_i, \mu_j) \ge \frac{1}{2} \sum_{t=1}^p \gamma^{t-1}(1-\gamma) 
    = \frac{1 - \gamma}{2} \cdot \frac{1 - \gamma^{p}}{1 - \gamma} = \frac{1-\gamma^p}{2}.
  \end{equation*}
  By \Cref{lemma:lower-bound-properties-gadgets}, it holds that  $\gamma \le 1 - 2^{-k\ell}$.
  It follows that
  $1 - \gamma^p \ge 1 - (1 - 2^{-k\ell})^p \ge 1 - \exp\tp{-p 2^{-k\ell}}$.
  By the assumption $m 2^{-k \ell} < 1 / 2$ in the lemma, 
  we have $p 2^{-k \ell} < 1 / 2$. Hence $\exp(-p 2^{-k \ell}) \le 1 - p 2^{-k \ell} / 2$ and 
  $1 - \gamma^p \ge p 2^{-k \ell - 1}$. 
  Therefore, $\dtv(\mu_i, \mu_j) \ge p 2^{-k\ell - 2}$.
\end{proof}

Now, we are ready to prove the lower bound in \Cref{theorem:lower-bound-general-thm}.
Fix a constant integer $k \geq 2$. Fix a constant error bound $\varepsilon_0 \in (0, \frac{1}{200\cdot 2^k})$.
For any sufficiently large integer $m \geq m_0(k,\varepsilon_0)$, define
\begin{align}\label{eq:def-m-ell}
  \ell \triangleq \floor{\frac{1}{k} \log \frac{m}{100 \varepsilon_0}} \geq 2.
\end{align}
Note that $\varepsilon_0 \leq \frac{1}{200\cdot 2^k}$. 
It holds that $m\cdot 2^{-k\ell} \leq m \cdot 2^{-\log \frac{m}{100 \varepsilon_0} + k} 
= 2^k \cdot 100 \varepsilon_0 < \frac{1}{2}$,
which satisfies the condition in \Cref{lemma:lower-bound-tv-distance-binary-distance}.

Using \Cref{definition:lower-bound-set-of-hard-cnf-formulas} with parameter $m$, $\ell$, and $k$, we construct a set $\+X = \{\Phi_i: 0 \leq i < 2^m\}$ of $(k,k,k-1)$-CNFs. Note that the number of variables in each CNF formula is
\begin{align*}
  n = m\cdot k\cdot \ell, \quad\text{where } n \to \infty \text{ as } m \to \infty.
\end{align*}

We prove the following lower bound on learning CNF formulas in $\+X$.

\begin{lemma}\label{lemma:lower-bound-general-thm}
  Fix a constant integer $k \geq 2$ and 
  a constant error bound $\varepsilon_0 \in (0, \frac{1}{200 \cdot 2^k})$. 
  For any sufficiently large  $m \geq m_0(k,\varepsilon_0)$, 
  let $\ell$ be defined in \eqref{eq:def-m-ell}, 
  the following results hold for $\+X$ in \Cref{definition:lower-bound-set-of-hard-cnf-formulas}.

  Let $0 \leq K < 2^m$ be a uniform random integer. 
  Let $X_1,X_2,\ldots,X_T$ be $T$ \iid samples from $\mu_{\Phi_K}$. 
  Any algorithm such that given $X_1,X_2,\ldots,X_T$, 
  outputs a CNF formula $\Phi_{\tilde{K}} \in \+X$ satisfying $\dtv(\mu_{\Phi_K}, \mu_{\Phi_{\tilde{K}}}) \le \varepsilon_0$ 
  with probability at least $\frac{1}{3}$ requires at least 
  $T = \frac{1}{25\cdot 2^k}\cdot (\frac{n}{100 \varepsilon_0 \log(\frac{n}{100 \cdot \varepsilon_0})})^{\frac{k-2}{k}}$ samples, 
  where $n = mk\ell$ is the number of variables for $(k,k,k-1)$-CNF formulas in $\+X$.
\end{lemma}

Assuming the correctness of \Cref{lemma:lower-bound-general-thm}, 
we can already prove \Cref{theorem:lower-bound-general-thm}.

\lowerApproxPoly*

\begin{proof}
  Note that the algorithm in this theorem can output an arbitrary CNF formula $\tilde{\Phi}$, 
  rather than a CNF formula in $\+X$.

  Fix an integer $k \ge 2$ and an error bound $\varepsilon_0$ satisfying 
  $0 < 2\varepsilon_0 < \frac{1}{200 \cdot 2^k}$.
  For any sufficiently large integer $m \ge m_0(k, 2\varepsilon_0)$,
  let $\ell$ be defined in~\eqref{eq:def-m-ell}, and set $n = mk\ell$.
  Suppose the algorithm $\+A$ in \Cref{theorem:lower-bound-general-thm} exists 
  and $\+A$ uses less than 
  $T = \frac{1}{25\cdot 2^k}\cdot 
  (\frac{n}{100 \cdot 2 \varepsilon_0 \log(\frac{n}{100 \cdot 2 \varepsilon_0})})^{\frac{k-2}{k}}$ samples 
  to learn a $(k,k,k-1)$-CNF formula within total variation distance at most $\varepsilon_0$, 
  then we show a contradiction to \Cref{lemma:lower-bound-general-thm}.
  Let $\Phi_{K}$ be a uniform random CNF formula in $\+X$. 
  We run $\+A$ with \iid samples $X_1,X_2,\ldots,X_T$ from $\mu_{\Phi_K}$. 
  Let $\tilde{\Phi}$ be the CNF formula that $\+A$ outputs. 
  Next, we enumerate all $2^m$ CNF formulas in $\+X$ 
  and find a $\Phi_{\tilde{K}}$ that minimizes 
  $\dtv(\mu_{\tilde{\Phi}}, \mu_{\Phi_{\tilde{K}}})$ among all $\Phi_{\tilde{K}} \in \+X$. 
  Finally, we output $\Phi_{\tilde{K}}$. 
  By the assumption of $\+A$, 
  given any $\Phi_K \in \+X$, with probability at least $\frac{1}{3}$, 
  $\dtv(\mu_{\Phi_K}, \mu_{\tilde{\Phi}}) \le \varepsilon_0$. 
  Since $\Phi_{\tilde{K}} \in \+X$, 
  we have $\dtv(\mu_{\tilde{\Phi}}, \mu_{\Phi_{\tilde{K}}}) 
  \leq \dtv(\mu_{\tilde{\Phi}}, \mu_{\Phi_{K}}) \leq \varepsilon_0$. 
  By the triangle inequality, it holds that
  \begin{align*}
    \dtv(\mu_{\Phi_K}, \mu_{\Phi_{\tilde{K}}}) 
    \le \dtv(\mu_{\Phi_K}, \mu_{\tilde{\Phi}}) + \dtv(\mu_{\tilde{\Phi}}, \mu_{\Phi_{\tilde{K}}}) \le 2\varepsilon_0.
  \end{align*}
  This contradicts to \Cref{lemma:lower-bound-general-thm} with the error bound $2\varepsilon_0$, 
  which proves the $\Omega((\frac{n}{\log n})^{\frac{k-2}{k}})$  
  sample complexity lower bound.
\end{proof}

Finally, we use \Cref{lemma:distance-fano-inequality} to prove \Cref{lemma:lower-bound-general-thm}.

\begin{proof}[Proof of \Cref{lemma:lower-bound-general-thm}]
  Suppose there exists an algorithm $\+A$ that given $X_1,X_2,\ldots,X_T$, 
  outputs a CNF formula $\Phi_{\tilde{K}} \in \+X$. 
  Consider the following Markov chain:
  \begin{align*}
    \Phi_K \to \tp{X_1,X_2,\ldots,X_T }\to \Phi_{\tilde{K}} , 
    \text{ equivalently } K \to \tp{X_1,X_2,\ldots,X_T } \to \tilde{K}.
  \end{align*}
  We show that if 
  $T = \frac{1}{25\cdot 2^k}\cdot (\frac{n}{100 \varepsilon_0 \log(\frac{n}{100 \cdot \varepsilon_0})})^{\frac{k-2}{k}}$, 
  then with probability at least $\frac{9}{10}$, 
  $\dtv(\mu_{\Phi_K}, \mu_{\Phi_{\tilde{K}}}) > \varepsilon_1$, 
  where $\varepsilon_1 = \frac{m\cdot 2^{-k\ell}}{100}$. 
  By definition in~\eqref{eq:def-m-ell}, 
  we have $\varepsilon_0 \leq \varepsilon_1 \leq 2^k \varepsilon_0$ and $m \cdot 2^{-k \ell} < 1 / 2$. 
  By \Cref{lemma:lower-bound-tv-distance-binary-distance}, 
  we know that if $\binaryDistance(K, \tilde{K}) > \frac{m}{25}$, 
  then it must hold that $\dtv(\mu_{\Phi_K}, \mu_{\Phi_{\tilde{K}}}) > \varepsilon_1$.
  Hence, it suffices to show
  \begin{align}\label{eq:lowerdis}
    \Pr{\binaryDistance(K, \tilde{K}) > \frac{m}{25}} > \frac{9}{10}.
  \end{align}

  We use distance-based Fano's inequality in \Cref{lemma:distance-fano-inequality} to prove the claim.
  We set up all parameters for the distance-based Fano's inequality.
  Let the function $\rho(\cdot,\cdot)$ be $\binaryDistance(\cdot,\cdot)$. We set the threshold $t = \frac{m}{25}$.
  To use the inequality, we need to verify $|\+X| - N_t^{\min} > N_t^{\max}$, give a lower bound on $\ln\left( \frac{|\+X|-N_t^{\min}}{N_t^{\max}}\right)$ and upper bound on $I(\sigma_1,\sigma_2,\dots, \sigma_T; K)$.

  We claim that $N^{\max}_t \le \sum_{j=0}^{m/25} \binom{m}{j}\le e^{m H_b(1/25)}$ where $H_b(x) = -x\ln(x) - (1-x)\ln(1-x)$.
  To verify the bound, let $X$ be the sum of $m$ \iid Bernoulli random variables with parameter $1 / 2$.
  Then $\sum_{j = 0}^{\alpha m} 2^{-m} \binom{m}{j} \le \Pr{X \le \alpha m}$.
  By the Chernoff bound, we have $\Pr{X \leq \alpha m} \le \exp\sqb{-m\tp{\alpha \ln (2\alpha) + (1 - \alpha) \ln(2(1 - \alpha))}}$.
  Hence $\sum_{j=0}^{\alpha m} \binom{m}{j} \le \exp(m H_b(\alpha))$ follows by rearranging the terms.
  Note that $N^{\min}_t \leq N^{\max}_t$.
  It can be verified that $2^m - e^{m H_b(1/25)} > e^{m H_b(1/25)}$ for $m \ge 2$, which implies that $|\+X| - N^{\min}_t > N^{\max}_t$.
  To give a lower bound on $\ln\left( \frac{|\+X|-N_t^{\min}}{N_t^{\max}}\right)$, we have that 
  $\frac{|\+X| - N_t^{\min}}{N_t^{\max}} \ge \frac{2^m}{\exp(m H_b(1/25))} - 1 \ge 2^{0.757 m} - 1 \ge 2^{0.75 m}$ where the last inequality holds when $m\ge 7$. Hence, $\ln\left( \frac{|\+X| - N_t^{\min}}{N_t^{\max}}\right) \ge 0.75m \ln 2$.
  % $\ln\left( \frac{|\+X| - N_t^{\min}}{N_t^{\max}}\right) \ge m(\ln2 - H_b(1/25)) \ge 0.75 m \ln 2$.

  Next, we upper bound $I(\sigma_1,\sigma_2,\dots, \sigma_T; K)$. By the chain rule, we have that $I(\sigma_1,\sigma_2,\dots, \sigma_T; K)= \sum_{i=1}^T I(\sigma_i; K ~|~ \sigma_1,\sigma_2,\dots, \sigma_{i-1})\le \sum_{i=1}^T I(\sigma_i; K)$. And by symmetry, it suffices to bound $I(\sigma; K)$, where $\sigma \sim \mu_{\Phi_K}$. Recall that for any $\Phi_i$, it consists of $m$ disjoint gadgets. For $j\in[m]$, we use $\sigma^{(j)}$ to denote the random assignment of variables in the $j$-th gadgets projected from $\sigma$. Also by the chain rule, we have $I(\sigma; K)= \sum_{j=1}^m I(\sigma^{(j)}; K~|~ \sigma^{(1)}, \sigma^{(2)},\dots, \sigma^{(j-1)})\le \sum_{j=1}^m I(\sigma^{(j)}; K)$.
  By symmetry, it suffices to bound $I(\sigma^{(1)}; K)$. For any fixed
  $0\leq j < 2^m$, let $p_j$ be the distribution of $\mu_{\Phi_j}$ projected on the variables in the first gadget. Let $\bar{p}$ be the averaged distribution, i.e., $\bar{p} = \frac{1}{2^m} \sum_{j=0}^{2^m-1} p_j$ when $0 \leq j < 2^m$ is sampled uniformly at random. Note that the mutual information can be written as the KL divergence between the joint distribution and the product of the marginal distributions.
  A simple calculation shows that
\begin{align*}
    I(\sigma^{(1)}; K) = \sum_{j=0}^{2^m - 1}\sum_{x \in \{\text{True,False}\}^{U_1}}\frac{p_j(x)}{2^m} \ln \frac{p_j(x)/2^m}{\bar{p}(x)/2^m} = \=E_{K}[\dkl{p_K}{\bar{p}}].
  \end{align*}
  Consider two cases: the first gadget is restricted or unrestricted, depending on the value of $K$. When the first gadget is restricted, let the distribution on variables in the first gadget be $p_r$. Similarly, let $p_u$ be the distribution when the first gadget is unrestricted. Then $p_K$ is either $p_r$ or $p_u$. We have
  \[\dkl{p_r}{\bar{p}} = \sum_{x\in \Omega_r} p_r(x)\ln\tp{\frac{p_r(x)}{\frac{1}{2}p_r(x) + \frac{1}{2}p_{u}(x)}} \le \ln \tp{\frac{1}{\frac{1}{2} + \frac{1}{2}\frac{|\Omega_r|}{|\Omega_u|}}},\]
  where $\Omega_r$ denotes the support of $p_r$ and note that the support of $\bar{p}$ is $\Omega_u \supseteq \Omega_r$.
  By \Cref{lemma:lower-bound-properties-gadgets}, we have $\dkl{p_r}{\bar{p}}\le \ln\tp{1 + \frac{|\Omega_u| - |\Omega_r|}{|\Omega_u| + |\Omega_r|}}\le \frac{|\Omega_u| - |\Omega_r|}{|\Omega_u| + |\Omega_r|} \leq 1 - \frac{|\Omega_r|}{|\Omega_u|}\le 2^{-(k-2)\ell}$.
  Similarly, we have
  \begin{align*}
    \dkl{p_u}{\bar{p}} =                                     & ~ \sum_{x\in \Omega_u} p_u(x)\ln\tp{\frac{p_u(x)}{\frac{1}{2}p_r(x) + \frac{1}{2}p_u(x)}}                                                                                                 \\
    =                                                        & ~ \frac{|\Omega_r|}{|\Omega_u|}\ln\tp{\frac{\frac{|\Omega_r|}{|\Omega_u|}}{\frac{1}{2} + \frac{1}{2}\frac{|\Omega_r|}{|\Omega_u|}}} + \frac{|\Omega_u| - |\Omega_r|}{|\Omega_u|}\ln\tp{2} \\
    \tp{\text{by }\frac{|\Omega_r|}{|\Omega_u|} < 1}\quad\le & ~ \tp{1 - \frac{|\Omega_r|}{|\Omega_u|}}\ln2.
  \end{align*}
  Also by \Cref{lemma:lower-bound-properties-gadgets}, it holds that $\dkl{p_u}{\bar{p}}\le 2^{-(k-2)\ell}\ln2$.
  Combining everything, we have the following bound on the mutual information
  \[I(\sigma_1, \sigma_2, \dots, \sigma_T; K)\le T\cdot m \cdot 2^{-(k-2)\ell}\ln2.\]

  Using distance-based Fano's inequality in \Cref{lemma:distance-fano-inequality}, we have %\yyxtodo{I modified the expression.}
  \begin{align*}
    \Pr{d_b(K,\tilde{K}) > \frac{m}{25}} \geq 1 - \frac{I(X; Y) + \ln2}{\ln\left( {(|\+X|-N_t^{\min})}/{N_t^{\max}} \right)} \geq 1 - \frac{T\cdot m \cdot 2^{-(k-2)\ell}\ln2 + \ln 2}{0.75m \ln 2}.
  \end{align*}
  Assume that $m$ is large enough. Then, if $T\le 0.01 \cdot 2^{(k-2)\ell}$, then $\Pr{d_b(K,\tilde{K}) > \frac{m}{25}} > \frac{9}{10}$.
  By our choices of parameter, $\ell \geq \frac{1}{k}\log\tp{\frac{m}{100\cdot \varepsilon_0}} - 1$ and then $2^{(k-2)\ell} \geq \frac{4}{2^k}(\frac{m}{100\varepsilon_0})^{\frac{k-2}{k}} =  \frac{4}{2^k}(\frac{n}{100\varepsilon_0k\ell})^{\frac{k-2}{k}}$, where we use the definition that $n = m k \ell$. Since $\ell \leq \frac{1}{k}\log\tp{\frac{m}{100\cdot \varepsilon_0}} \leq \frac{1}{k}\log\tp{\frac{n}{100\cdot \varepsilon_0}}$, we have if
  \begin{align*}
    T \leq \frac{1}{25\cdot 2^k}\cdot \tp{\frac{n}{100 \varepsilon_0 \log(\frac{n}{100 \cdot \varepsilon_0})}}^{\frac{k-2}{k}} \leq 0.01 \cdot 2^{(k-2)\ell},
  \end{align*}
  then $\Pr{d_b(K,\tilde{K}) > \frac{m}{25}} > \frac{9}{10}$.
  This verifies~\eqref{eq:lowerdis} and proves the lemma.
\end{proof}

\subsection{Sample complexity of exact learning CNF formulas in the local lemma regime}
Using the gadgets in \Cref{definition:lower-bound-gadgets}, 
we can also establish an exponential lower bound on sample complexity 
of exact learning CNF formulas in the local lemma regime.
\lowerExactExp*

\begin{proof}
  Fix $k \geq 2$. 
  For any $\ell$, construct restricted and unrestricted depth-$\ell$ gadgets $\Phi_r$ and $\Phi_u$. 
  The number of variables is $n = k \ell  = \Theta(\ell)$. 
  By \Cref{lemma:lower-bound-properties-gadgets}, 
  the total variation distance between $\mu_{\Phi_r}$ and $\mu_{\Phi_u}$ is at most $2^{-\Omega_k(n)}$. 
  If an algorithm can exact learn $\Phi_r$ and $\Phi_u$, 
  then it can distinguish between $\mu_{\Phi_r}$ and $\mu_{\Phi_u}$ from $T$ samples. 
  The total variation distance between $T$ \iid samples 
  from $\mu_{\Phi_r}$ and $T$ \iid samples from $\mu_{\Phi_u}$ is at most $T\cdot 2^{-\Omega_k(n)}$. 
  Hence, exact learning $(k,k,k-1)$-CNF formulas with constant probability 
  requires $\exp(\Omega_k(n))$ samples.
\end{proof}

\ifthenelse{\boolean{doubleblind}}{}{
\section*{Acknowledgements}
We thank Xue Chen, Zhe Hou, Eric Vigoda, and Yitong Yin for helpful discussions.
Weiming Feng acknowledges the support of ECS grant 27202725 from Hong Kong RGC.
Yixiao Yu acknowledges the support of the National Natural Science Foundation of China under Grant No. 62472212.
}

\printbibliography

\appendix
\section{A counterexample to the correlation lower bound}\label{sec:counterexample-correlation-lower-bound}

For any CNF formula $\Phi$ with the uniform distribution $\mu_{\Phi}$ on its satisfying assignments, 
in order to apply the techniques in~\cite[Theorem~4]{BreslerMS13}, 
we have to ensure that the following quantity has a positive lower bound for any $u, v \in V$ 
with $\{u, v\} \subseteq \vbl(c)$ for some clause $c \in \+C$:
\[
d_{\mathrm{C}}(u, v) 
  \defeq 
  \sum_{x_u, x_v \in \{\true, \false\}}
  \abs{
    \Pr[X \sim \mu_{\Phi}]{X(u)=x_u, X(v)=x_v}
    - \Pr[X \sim \mu_{\Phi}]{X(u)=x_u} \Pr[X \sim \mu_{\Phi}]{X(v)=x_v}
  }.
\]

Consider the CNF formula $\Phi$ that contains only two clauses:
\[
c_1 = v_1 \lor v_2 \lor \dots \lor v_{k-1} \lor v_k,
\qquad
c_2 = v_1 \lor v_2' \lor \dots \lor v_{k-1}' \lor \lnot v_k,
\]
where the variable sets $\{v_2, \dots, v_{k-1}\}$ and $\{v_2', \dots, v_{k-1}'\}$ are disjoint.
We claim that $d_{\mathrm{C}}(v_1, v_k) = 0$.

\paragraph{Counting argument.}
Let $a = 2^{k-2}$.
We enumerate all satisfying assignments of $\Phi$ and obtain:
$$N_{\true,\true} = a^2, \quad N_{\true,\false} = a^2, \quad N_{\false,\true} = a(a-1), \quad N_{\false,\false} = a(a-1).$$
where $N_{x_1, x_k}$ denotes the number of satisfying assignments 
with $v_1 = x_1$ and $v_k = x_k$.
Hence, the total number of satisfying assignments is $N = 2a^2 + 2a(a-1) = 2a(2a-1)$.
The corresponding marginal probabilities are $
\Pr{X(v_k) = \true}
  = \frac{a^2 + a(a-1)}{N}
  = \frac{1}{2}$, $
\Pr{X(v_1) = \true}
  = \frac{2a^2}{N}
  = \frac{a}{2a-1}$.
  
In particular,
\[
\Pr{X(v_1) = \true, X(v_k) = \true}
  = \frac{a^2}{N}
  = \frac{a}{2(2a-1)}
  = \Pr{X(v_1) = \true} \Pr{X(v_k) = \true},
\]
and by symmetry, the same equality holds for all other $(x_1, x_k) \in \{\true, \false\}^2$.
Therefore, $v_1$ and $v_k$ are independent under $\mu_{\Phi}$, 
and we conclude that $d_{\mathrm{C}}(v_1, v_k) = 0$.

Finally, we remark that this counterexample can be naturally extended into a large counterexample with $m$ clauses by adding symmetric structures on $c_1$ and $c_2$. Also by symmetry, one can verify that $d_{\mathrm{C}}(v_1, v_k) = 0$.

\end{document}